\documentclass[11pt,a4paper,oneside]{amsart}

\pdfoutput=1

\usepackage[colorlinks=true,citecolor=blue]{hyperref}
\usepackage{enumerate}
\usepackage[margin=3cm]{geometry}

\usepackage{graphicx}
\usepackage{multicol,multirow}
\usepackage{amsmath,amssymb,amsfonts,mathtools}
\usepackage{mathrsfs}
\usepackage{amsthm}
\usepackage{rotating}
\usepackage{appendix}
\usepackage[numbers]{natbib}
\usepackage{ifpdf}
\usepackage[T1]{fontenc}
\usepackage{newtxtext}
\usepackage{newtxmath}
\usepackage{textcomp}
\usepackage{tikz-cd}
\usepackage{tikz}
\usetikzlibrary{decorations,decorations.pathreplacing, math}
\usetikzlibrary{calc} 

\newtheorem{theorem}{Theorem}[section]
\newtheorem{lemma}[theorem]{Lemma}
\newtheorem{corollary}[theorem]{Corollary}
\newtheorem{conjecture}[theorem]{Conjecture}
\newtheorem{question}[theorem]{Question}
\newtheorem{claim}[theorem]{Claim}
\newtheorem{observation}[theorem]{Observation}
\newtheorem{proposition}[theorem]{Proposition}
\newtheorem*{theorem*}{Theorem}

\newtheorem{innercustomthm}{Theorem}
\newenvironment{customthm}[1]
  {\renewcommand\theinnercustomthm{#1}\innercustomthm}
  {\endinnercustomthm}

\theoremstyle{definition}

\newtheorem{example}[theorem]{Example}
\numberwithin{equation}{section}
\newtheorem*{remark*}{Remark}

\newcommand{\FF}{\mathcal{F}}
\newcommand{\GG}{\mathcal{G}}
\newcommand{\HH}{\mathcal{H}}
\newcommand{\NN}{\mathbb{N}}
\newcommand{\ZZ}{\mathbb{Z}}
\newcommand{\RR}{\mathbb{R}}

\newcommand{\univ}{\ensuremath{X}}
\newcommand{\grid}{\textnormal{Grd}}

\newcommand{\tbeta}{\tilde{\beta}}

\newcommand{\pth}[1]{\left(#1\right)}

\DeclareMathOperator{\stc}{sc}
\DeclareMathOperator{\supp}{supp}

\DeclareMathOperator{\bx}{box}
\DeclareMathOperator{\chain}{chn}

\newcommand{\shat}[2]{\phi^{(#1)}_{#2}}

\begin{document}

\author{Xavier Goaoc \and Andreas F. Holmsen \and Zuzana Pat\'akov\'a}

\thanks{X.G. and A.F.H. were supported by the INRIA associate teams FIP and DIPPS. A.F.H. was also supported by  the National Research Foundation of Korea (NRF) grants funded by the Ministry of Science and ICT (NRF-2020R1F1A1A01048490) and the Institute for Basic Science (IBS-R029-C1). Z.P. was supported by the Charles University project PRIMUS/21/SCI/014 and also by the GA\v CR grant no. 25-16847S}

\date{\today}

\address{Xavier Goaoc, Universit\'e de Lorraine, CNRS, INRIA, LORIA, Nancy, F-54000, France}
\email{xavier.goaoc@loria.fr}

\address{Andreas F. Holmsen, Department of Mathematical Sciences, KAIST, 
Daejeon, South Korea and Discrete Mathematics Group (DIMAG), Institute for Basic Science (IBS), Daejeon, South Korea. }
\email{andreash@kaist.edu}

\address{Zuzana Pat\'akov\'a, Department of Algebra, Faculty of Mathematics and Physics, Charles University, Sokolovsk\'a 49/83, Prague, 186 75, Czech Republic} \email{patakova@karlin.mff.cuni.cz}

\title[Intersection patterns in spaces with a forbidden homological minor]{Intersection patterns in spaces with a forbidden homological minor}

\begin{abstract}
In this paper we study generalizations of classical results on 
intersection patterns of set systems in $\mathbb{R}^d$, 
such as the fractional Helly theorem or the $(p,q)$-theorem, 
in the setting of arbitrary triangulable spaces with a forbidden homological minor. 

Given a simplicial complex $K$ and an integer $b$, 
we say that a family $\FF$ of subcomplexes of some simplicial complex 
$\univ$ is a \emph{$(K,b)$-free cover} if 
(i) $K$ is a forbidden homological minor of $\univ$, and 
(ii) the $j$th reduced Betti number $\tilde{\beta}_j(\bigcap_{S\in {\mathcal{G}}}S,\ZZ_2)$ is strictly less than $b$ for all $0\leq j < \dim K$ 
and all nonempty subfamilies $\GG\subseteq \FF$.

We show that for every $K$ and $b$, the fractional Helly number of a 
$(K,b)$-free cover is at most $\mu(K)+1$, 
where $\mu(K)$ is the maximum sum of the dimensions of two disjoint faces in~$K$. 
This implies that the  assertion of the $(p,q)$-theorem holds for every 
$p \ge q > \mu(K)$ and every $(K,b)$-free cover $\FF$. 
For $b=1$ and a suitable $K$ this recovers the original $(p,q)$-theorem and 
its generalization to good covers. 
Interestingly, our results show that  that the range of parameters $(p,q)$ 
for which the $(p,q)$-theorem holds is independent of $b$.

Our proofs use Ramsey-type arguments combined with the notion of stair 
convexity of Bukh et al. to construct (forbidden) homological minors in 
certain cubical complexes.
\end{abstract}


\maketitle

\section{Introduction}
\label{sec:intro}
Helly's theorem on the intersection of convex sets is one of the most 
well-known results of combinatorial convexity. 
Applications, generalizations, and variations of this theorem have been studied
extensively for nearly a century, and now comprise a significant area of discrete
geometry. We refer the reader 
to~\cite{DGKsurvey63, de2019discrete,Eckhoff:1993survey} 
for in-depth surveys and further references, and to the
textbooks~\cite{barany2021combinatorial,matousek2013lectures} 
for an introduction to the area. 

\subsection{Problem statement} 
\label{subsec:problem statement}
In this paper we are concerned with generalizations of Helly's theorem
that allow for more flexible intersection patterns and relax the convexity assumption.

A central result in this area is the celebrated 
\emph{$(p,q)$-theorem}~\cite{alon1992piercing}. 
It asserts that for a finite family of convex sets in $\RR^d$, 
if the family satisfies the $(p,q)$-condition -- meaning that 
\emph{among every $p$ members some $q$ intersect} -- then a 
\emph{constant} number of points suffice to intersect all the convex sets. 
The crucial feature is that the constant depends only on $p$, $q$ and $d$, 
but not on the size of the family. 
Note that Helly's theorem corresponds to the case $p=q=d+1$, 
which guarantees that a \emph{single} point intersects every set in the family. 

Two central problems in this line of research are to identify the 
weakest possible assumptions under which the classical theorems can be generalized, 
and to determine their key parameters, for instance the \emph{Helly number} 
($d+1$ for convex sets in $\RR^d$) or the range for which the $(p,q)$-theorem holds
(every $p \geq q \geq  d+1$ for convex sets in $\RR^d$).  

In this paper, we study the $(p,q)$ theorem (and other generalizations of Helly's theorem) in the broad topological setting of triangulable spaces with a  forbidden homological minor,  a notion introduced by Wagner~\cite{homMinors} as a higher-dimensional  analogue of the classical notion of graph minors~\cite{mohar}.

\subsection{Our setting}
\label{subsec:setting}
We now describe the general setting in which we establish our $(p,q)$-theorem.
Let $K$ be a finite simplicial complex. 
Throughout this paper, all complexes are assumed to be finite. 
Homology and chain complexes are taken with $\ZZ_2$-coefficients, 
that is, $C_{\bullet}(K)$ means $C_{\bullet}(K; \mathbb{Z}_2)$. 
We employ singular, simplicial, or cellular homology depending on the context.

\subsubsection*{Homological minors}\hfill\par
 
Our goal is to establish a $(p,q)$-theorem for some  general triangulable space $\univ$. 
It turns out that the key property of the space $\univ$ we utilize 
can be deduced from the simplicial complexes that, 
in a certain sense, do \emph{not} embed into $\univ$. 
For our purposes, this nonembeddability property is formalized via the concept of a forbidden homological minor, which we now define.

The \emph{support} of a chain $\sigma$, denoted $\supp(\sigma)$, 
in a simplicial complex is the set of simplices with nonzero coefficients in $\sigma$. 
We say that two chains $\sigma$ and $\tau$ have \emph{overlapping supports} 
if there exists a simplex in the support of $\sigma$ that intersects a 
simplex in the support of $\tau$; 
if no such pair of simplices exists we say that $\sigma$ and $\tau$ have 
\emph{nonoverlapping supports}. 

A chain map $f_\bullet \colon C_{\bullet}(K) \to C_{\bullet}(\univ)$ is 
\emph{nontrivial} if the image of every vertex of $K$ is a $0$-chain of $\univ$ 
supported on an odd number of vertices. 
The simplicial complex $K$ is a \emph{homological minor} of~$\univ$, 
written ${K \prec_H \univ}$, if there exists a nontrivial chain map 
$f_\bullet \colon C_{\bullet}(K) \to C_{\bullet}(\univ)$ 
such that disjoint simplices are mapped to chains with nonoverlapping supports. 
If no such chain map exists we say that $K$ is a \emph{forbidden homological minor} of $\univ$, 
and write $K \not\prec_H \univ$.  The notion of homological minor readily extends to any triangulable space: 
$K$ is a forbidden homological minor of a space $X$ if $K \not\prec_H T$ 
for every triangulation $T$ of $X$.

From now on, rather than focusing on some given space $\univ$, 
we fix a simplicial complex $K$ and 
consider the entire class of spaces that exclude $K$ as a homological minor.
This point of view requires a substitute for the 
dimension of the ambient space, which turns out to be the parameter
\[\mu(K) \coloneqq \max_{\stackrel{\sigma, \tau\in K}{\sigma\cap \tau = \emptyset}} 
\{\dim \sigma + \dim \tau\}.\] 
We note the obvious bounds $\dim K -1 \leq \mu(K) \leq 2\dim K$.

\begin{example}\label{ex:forbiddenRd}
A homological version of Radon's lemma~\cite[Lemma~15]{hb17}) asserts that $\partial \Delta_{d+1} \not\prec_H \RR^d$, that is, the boundary of the $(d+1)$-dimensional simplex is a forbidden homological minor of $\RR^d$. Note that $\mu\pth{\partial \Delta_{d+1}} = d = \dim\partial \Delta_{d+1}$. Similarly, a homological version of  the Van Kampen-Flores theorem~\cite[Corollary~14]{hb17}) asserts that $\Delta_{2k+2}^{(k)}\not\prec_H\RR^{2k}$, that is, the $k$-dimensional skeleton of the $(2k+2)$-dimensional simplex is a forbidden homological minor of $\RR^{2k}$. Note that $\mu\pth{\Delta_{2k+2}^{(k)}} = 2k = 2\dim \Delta_{2k+2}^{(k)}$.
\end{example}

\subsubsection*{$(K,b)$-free covers}\hfill\par

Since we have replaced $\mathbb{R}^d$ by a general triangulable space $\univ$ with a 
forbidden homological minor, the classical notion of a finite family 
of ``convex sets'' may no longer be applicable. Therefore, 
we need to replace and relax the notion of convexity as well.  

To this end, let $\tilde{\beta}_j(\cdot)$ denote the $j$th reduced Betti number, that is, 
the rank of the reduced homology group $\tilde{H}_j(\cdot)$. Given a finite simplicial complex $K$ 
of positive dimension and a positive integer $b$
we define a \emph{$(K,b)$-free cover} in a simplicial complex $\univ$ to be a 
finite family  $\FF$ of (not necessarily induced) subcomplexes of $\univ$ such that: 
\begin{enumerate}[(i)]
\item[(i)] $K$ is a forbidden homological minor of $\univ$ (i.e., $K\not\prec_H \univ$), and
\item[(ii)] the intersection of any nonempty subfamily has bounded complexity:
\[\tilde{\beta}_j \big( \bigcap_{S\in {\mathcal{G}}}S \big) <b\]
for all $0\leq j < \dim K$ and all $\emptyset \neq \mathcal{G}\subseteq \mathcal{F}$. 
\end{enumerate} 

\begin{example}\label{ex:coverRd}
Let $K_5$ denote the complete graph on 5 vertices, viewed as a 1-dimensional simplicial complex and let $X$ be a triangulation of a $2$-dimensional disk. As a special case of Example~\ref{ex:forbiddenRd}, we have that $K_5\not\prec_H X$.
\end{example}

\subsubsection*{Illustration: analyzing intersection patterns of planar compact convex sets}\hfill\par
Let us showcase how the intersection patterns of finite families of compact convex sets in $\mathbb{R}^2$ can be analyzed in our setting.

Consider some finite family $\FF$ of compact convex sets in $\mathbb{R}^2$. Let $\univ$ be a sufficiently fine triangulation of the one-point compactification of $\mathbb{R}^2$ with the property that any nonempty intersection of members of $\FF$ contains at least one vertex of $\univ$. For each $S\in \FF$, let $C_S$ denote the subcomplex of $\univ$ induced by the vertices contained in $S$, and define $\GG = \{C_S : S\in \FF\}$. Since the sets in $\FF$ are convex, by choosing $\univ$ sufficiently fine, we ensure that any nonempty intersection of subcomplexes in $\GG$ is \emph{contractible} (and thus connected).

\begin{figure}[h]
 \centering
\begin{tikzpicture}[scale=1]

\begin{scope}[scale = 0.8] 

\begin{scope}

\begin{scope}
\draw[teal]
(0,1) .. controls (0,1.5) and (1.5,2) ..
(2,2) .. controls (2.5,2) and (4,1.8) ..
(4,1) .. controls (4,.6) and (2,0) ..
(1,0) .. controls (0.5,0) and (0,0.5) ..
(0,1);
\fill[teal, opacity = .15]
(0,1) .. controls (0,1.5) and (1.5,2) ..
(2,2) .. controls (2.5,2) and (4,1.8) ..
(4,1) .. controls (4,.6) and (2,0) ..
(1,0) .. controls (0.5,0) and (0,0.5) ..
(0,1);
\end{scope}

\begin{scope}[rotate = -10, yshift = 0.7cm]
\draw[violet]
(3,1) .. controls (4,1) and (5,0.3) .. 
(5,-1) .. controls (5,-1.5) and (4,-2) ..  
(3,-2) .. controls (2.5,-2) and (2.2,1) ..
(3,1);
\fill[violet, opacity = .15]
(3,1) .. controls (4,1) and (5,0.3) .. 
(5,-1) .. controls (5,-1.5) and (4,-2) ..  
(3,-2) .. controls (2.5,-2) and (2.2,1) ..
(3,1);
\end{scope}

\begin{scope}[rotate = 10]
\draw[blue] 
(3,-1) .. controls (3.3,-1) and (4,-2) ..
(4,-3) .. controls (4,-3.8) and (-1,-3.6) ..
(-1,-3) .. controls (-1,-2.75) and (-0.5,-2.25) ..
(0,-2) .. controls (0.5,-1.75) and (2.5,-1) ..
(3,-1);
\fill[blue, opacity = .15] 
(3,-1) .. controls (3.3,-1) and (4,-2) ..
(4,-3) .. controls (4,-3.8) and (-1,-3.6) ..
(-1,-3) .. controls (-1,-2.75) and (-0.5,-2.25) ..
(0,-2) .. controls (0.5,-1.75) and (2.5,-1) ..
(3,-1);
\end{scope}

\begin{scope}[xshift = .6cm, yshift = -0.8cm]
\begin{scope}[scale = 1.2, rotate=55]
\draw[red]
(0,1) .. controls (0.2,1.2) and (0.7,1.3) ..
(1,1) .. controls (1.5,0.5) and (1.4,-1.6) ..
(1,-2) .. controls (0.5,-2.5) and (0.2,-2.2) ..
(0,-2) .. controls (-0.5,-1.5) and (-0.7,0.3) ..
(0,1);
 \fill[red, opacity = .15]
(0,1) .. controls (0.2,1.2) and (0.7,1.3) ..
(1,1) .. controls (1.5,0.5) and (1.4,-1.6) ..
(1,-2) .. controls (0.5,-2.5) and (0.2,-2.2) ..
(0,-2) .. controls (-0.5,-1.5) and (-0.7,0.3) ..
(0,1);
\end{scope}
\end{scope}

\begin{scope}[xshift = -0.94cm, yshift = -3.84cm]
\foreach \x in {1,...,20} {
	\foreach \y in {1,...,20} {
		\fill[black, opacity = .17] (\x *0.3, \y *0.3) circle (0.015cm) ; } }
\foreach \x in {1,...,19} {
	\foreach \y in {1,...,19} {
		\draw[black, opacity = .1, line width = .1] (\x *0.3, \y *0.3) --++ (0.3,0.3); } }		
\foreach \x in {1,...,19} {
	\foreach \y in {1,...,20} {
		\draw[black, opacity = .1, line width = .1] (\x *0.3, \y *0.3) --++ (0.3,0); } }
\foreach \x in {1,...,20} {
	\foreach \y in {1,...,19} {
		\draw[black, opacity = .1, line width = .1] (\x *0.3, \y *0.3) --++ (0,0.3); } }	
\end{scope}
\end{scope}

\begin{scope}[xshift = 8cm]
\begin{scope}[xshift = -0.94cm, yshift = -3.84cm]
\foreach \x in {1,...,20} {
	\foreach \y in {1,...,20} {
		\fill[black, opacity = .17] (\x *0.3, \y *0.3) circle (0.015cm) ; } }
\foreach \x in {1,...,19} {
	\foreach \y in {1,...,19} {
		\draw[black, opacity = .1, line width = .1] (\x *0.3, \y *0.3) --++ (0.3,0.3); } }		
\foreach \x in {1,...,19} {
	\foreach \y in {1,...,20} {
		\draw[black, opacity = .1, line width = .1] (\x *0.3, \y *0.3) --++ (0.3,0); } }
\foreach \x in {1,...,20} {
	\foreach \y in {1,...,19} {
		\draw[black, opacity = .1, line width = .1] (\x *0.3, \y *0.3) --++ (0,0.3); } }		

\foreach \x in {2,...,8} {
	\fill[blue, opacity = .2] (\x *0.3, 2 *0.3) --++(1*.3,1*.3)--++(0,-1 *.3)--cycle; }
\foreach \x in {2,...,13} {
	\fill[blue, opacity = .2] (\x *0.3, 3 *0.3) --++(1*.3,1*.3)--++(0,-1 *.3)--cycle; }
\foreach \x in {3,...,15} {
	\fill[blue, opacity = .2] (\x *0.3, 4 *0.3) --++(1*.3,1*.3)--++(0,-1 *.3)--cycle; }
\foreach \x in {3,...,16} {
	\fill[blue, opacity = .2] (\x *0.3, 5 *0.3) --++(1*.3,1*.3)--++(0,-1 *.3)--cycle; }
\foreach \x in {5,...,16} {
	\fill[blue, opacity = .2] (\x *0.3, 6 *0.3) --++(1*.3,1*.3)--++(0,-1 *.3)--cycle; }
\foreach \x in {7,...,16} {
	\fill[blue, opacity = .2] (\x *0.3, 7 *0.3) --++(1*.3,1*.3)--++(0,-1 *.3)--cycle; }
\foreach \x in {8,...,15} {
	\fill[blue, opacity = .2] (\x *0.3, 8 *0.3) --++(1*.3,1*.3)--++(0,-1 *.3)--cycle; }
\foreach \x in {10,...,14} {
	\fill[blue, opacity = .2] (\x *0.3, 9 *0.3) --++(1*.3,1*.3)--++(0,-1 *.3)--cycle; }
\foreach \x in {12,...,13} {
	\fill[blue, opacity = .2] (\x *0.3, 10 *0.3) --++(1*.3,1*.3)--++(0,-1 *.3)--cycle; }
\foreach \x in {2,...,9} {
	\fill[blue, opacity = .2] (\x *0.3, 2 *0.3) --++(1*.3,1*.3)--++(-1 *.3, 0)--cycle; }
\foreach \x in {3,...,14} {
	\fill[blue, opacity = .2] (\x *0.3, 3 *0.3) --++(1*.3,1*.3)--++(-1 *.3, 0)--cycle; }
\foreach \x in {3,...,16} {
	\fill[blue, opacity = .2] (\x *0.3, 4 *0.3) --++(1*.3,1*.3)--++(-1 *.3, 0)--cycle; }
\foreach \x in {4,...,16} {
	\fill[blue, opacity = .2] (\x *0.3, 5 *0.3) --++(1*.3,1*.3)--++(-1 *.3, 0)--cycle; }
\foreach \x in {6,...,16} {
	\fill[blue, opacity = .2] (\x *0.3, 6 *0.3) --++(1*.3,1*.3)--++(-1 *.3, 0)--cycle; }
\foreach \x in {8,...,16} {
	\fill[blue, opacity = .2] (\x *0.3, 7 *0.3) --++(1*.3,1*.3)--++(-1 *.3, 0)--cycle; }
\foreach \x in {9,...,15} {
	\fill[blue, opacity = .2] (\x *0.3, 8 *0.3) --++(1*.3,1*.3)--++(-1 *.3, 0)--cycle; }
\foreach \x in {11,...,14} {
	\fill[blue, opacity = .2] (\x *0.3, 9 *0.3) --++(1*.3,1*.3)--++(-1 *.3, 0)--cycle; }
\foreach \x in {13,...,13} {
	\fill[blue, opacity = .2] (\x *0.3, 10 *0.3) --++(1*.3,1*.3)--++(-1 *.3, 0)--cycle; }
\draw[blue, opacity =.5]
(2 *0.3, 2 *0.3)--++ (7*.3, 0 *.3)--++ (1*.3, 1 *.3) --++(4*.3, 0 *.3) --++(1*.3, 1 *.3) --++
(1*.3, 0 *.3)--++ (1*.3, 1 *.3) --++(0*.3, 3 *.3) --++(-1*.3, 0 *.3)--++ (0*.3, 1 *.3)--++
(-1*.3, 0 *.3)--++ (0*.3, 1 *.3) --++ (-1*.3, 0 *.3)--++ (0*.3, 1 *.3) --++ (-1*.3, 0 *.3) --++
(-1*.3, -1 *.3)--++ (-1*.3, 0 *.3) --++ (-1*.3, -1 *.3)--++ (-1*.3, 0 *.3) --++ (-2*.3, -2 *.3)--++
(-1*.3, 0 *.3) --++ (-1*.3, -1 *.3)--++ (-1*.3, 0 *.3) --++ (-1*.3, -1 *.3)--++ (0*.3, -1 *.3)--++
 (-1*.3, -1 *.3)-- cycle;

\foreach \y in {12,...,13} {
	\fill[red, opacity = .2] (2 *0.3, \y *0.3) --++(1*.3,1*.3)--++(0,-1 *.3)--cycle; }
\foreach \y in {10,...,14} {
	\fill[red, opacity = .2] (3 *0.3, \y *0.3) --++(1*.3,1*.3)--++(0,-1 *.3)--cycle; }
\foreach \y in {9,...,14} {
	\fill[red, opacity = .2] (4 *0.3, \y *0.3) --++(1*.3,1*.3)--++(0,-1 *.3)--cycle; }
\foreach \y in {9,...,14} {
	\fill[red, opacity = .2] (5 *0.3, \y *0.3) --++(1*.3,1*.3)--++(0,-1 *.3)--cycle; }
\foreach \y in {8,...,14} {
	\fill[red, opacity = .2] (6 *0.3, \y *0.3) --++(1*.3,1*.3)--++(0,-1 *.3)--cycle; }
\foreach \y in {7,...,13} {
	\fill[red, opacity = .2] (7 *0.3, \y *0.3) --++(1*.3,1*.3)--++(0,-1 *.3)--cycle; }
\foreach \y in {7,...,12} {
	\fill[red, opacity = .2] (8 *0.3, \y *0.3) --++(1*.3,1*.3)--++(0,-1 *.3)--cycle; }
\foreach \y in {6,...,12} {
	\fill[red, opacity = .2] (9 *0.3, \y *0.3) --++(1*.3,1*.3)--++(0,-1 *.3)--cycle; }
\foreach \y in {6,...,11} {
	\fill[red, opacity = .2] (10 *0.3, \y *0.3) --++(1*.3,1*.3)--++(0,-1 *.3)--cycle; }
\foreach \y in {6,...,10} {
	\fill[red, opacity = .2] (11 *0.3, \y *0.3) --++(1*.3,1*.3)--++(0,-1 *.3)--cycle; }
\foreach \y in {6,...,9} {
	\fill[red, opacity = .2] (12 *0.3, \y *0.3) --++(1*.3,1*.3)--++(0,-1 *.3)--cycle; }
\foreach \y in {12,...,12} {
	\fill[red, opacity = .2] (2 *0.3, \y *0.3) --++(1*.3,1*.3)--++(-1 *.3, 0)--cycle; }
\foreach \y in {10,...,14} {
	\fill[red, opacity = .2] (3 *0.3, \y *0.3) --++(1*.3,1*.3)--++(-1 *.3, 0)--cycle; }
\foreach \y in {9,...,14} {
	\fill[red, opacity = .2] (4 *0.3, \y *0.3) --++(1*.3,1*.3)--++(-1 *.3, 0)--cycle; }
\foreach \y in {9,...,14} {
	\fill[red, opacity = .2] (5 *0.3, \y *0.3) --++(1*.3,1*.3)--++(-1 *.3, 0)--cycle; }
\foreach \y in {8,...,14} {
	\fill[red, opacity = .2] (6 *0.3, \y *0.3) --++(1*.3,1*.3)--++(-1 *.3, 0)--cycle; }
\foreach \y in {7,...,13} {
	\fill[red, opacity = .2] (7 *0.3, \y *0.3) --++(1*.3,1*.3)--++(-1 *.3, 0)--cycle; }
\foreach \y in {7,...,12} {
	\fill[red, opacity = .2] (8 *0.3, \y *0.3) --++(1*.3,1*.3)--++(-1 *.3, 0)--cycle; }
\foreach \y in {6,...,12} {
	\fill[red, opacity = .2] (9 *0.3, \y *0.3) --++(1*.3,1*.3)--++(-1 *.3, 0)--cycle; }
\foreach \y in {6,...,11} {
	\fill[red, opacity = .2] (10 *0.3, \y *0.3) --++(1*.3,1*.3)--++(-1 *.3, 0)--cycle; }
\foreach \y in {6,...,10} {
	\fill[red, opacity = .2] (11 *0.3, \y *0.3) --++(1*.3,1*.3)--++(-1 *.3, 0)--cycle; }
\foreach \y in {6,...,9} {
	\fill[red, opacity = .2] (12 *0.3, \y *0.3) --++(1*.3,1*.3)--++(-1 *.3, 0)--cycle; }
\foreach \y in {7,...,7} {
	\fill[red, opacity = .2] (13 *0.3, \y *0.3) --++(1*.3,1*.3)--++(-1 *.3, 0)--cycle; }
\draw[red, opacity = .5] 
(13*.3,6*.3)--++(0*.3,1*.3)--++(1*.3,1*.3)--++(-1*.3,0*.3)--++(0*.3,2*.3)--++(-1*.3,0*.3)
--++(0*.3,1*.3)--++(-1*.3,0*.3)--++(0*.3,1*.3)--++(-1*.3,0*.3)--++(0*.3,1*.3)--++(-2*.3,0*.3)
--++(0*.3,1*.3)--++(-1*.3,0*.3)--++(0*.3,1*.3)--++(-4*.3,0*.3)--++(0*.3,-1*.3)--++(-1*.3,-1*.3)
--++(0*.3,-1*.3)--++(1*.3,0*.3)--++(0*.3,-2*.3)--++(1*.3,0*.3)--++(0*.3,-1*.3)--++(2*.3,0*.3)
--++(0*.3,-1*.3)--++(1*.3,0*.3)--++(0*.3,-1*.3)--++(2*.3,0*.3)--++(0*.3,-1*.3)--cycle;

\foreach \x in {6,...,7} {
	\fill[teal, opacity = .2] (\x *0.3, 13 *0.3) --++(1*.3,1*.3)--++(0,-1 *.3)--cycle; }
\foreach \x in {5,...,11} {
	\fill[teal, opacity = .2] (\x *0.3, 14 *0.3) --++(1*.3,1*.3)--++(0,-1 *.3)--cycle; }
\foreach \x in {4,...,14} {
	\fill[teal, opacity = .2] (\x *0.3, 15 *0.3) --++(1*.3,1*.3)--++(0,-1 *.3)--cycle; }
\foreach \x in {4,...,15} {
	\fill[teal, opacity = .2] (\x *0.3, 16 *0.3) --++(1*.3,1*.3)--++(0,-1 *.3)--cycle; }
\foreach \x in {4,...,14} {
	\fill[teal, opacity = .2] (\x *0.3, 17 *0.3) --++(1*.3,1*.3)--++(0,-1 *.3)--cycle; }
\foreach \x in {7,...,12} {
	\fill[teal, opacity = .2] (\x *0.3, 18 *0.3) --++(1*.3,1*.3)--++(0,-1 *.3)--cycle; }
\foreach \x in {6,...,8} {
	\fill[teal, opacity = .2] (\x *0.3, 13 *0.3) --++(1*.3,1*.3)--++(-1 *.3, 0)--cycle; }
\foreach \x in {5,...,12} {
	\fill[teal, opacity = .2] (\x *0.3, 14 *0.3) --++(1*.3,1*.3)--++(-1 *.3, 0)--cycle; }
\foreach \x in {4,...,15} {
	\fill[teal, opacity = .2] (\x *0.3, 15 *0.3) --++(1*.3,1*.3)--++(-1 *.3, 0)--cycle; }
\foreach \x in {4,...,15} {
	\fill[teal, opacity = .2] (\x *0.3, 16 *0.3) --++(1*.3,1*.3)--++(-1 *.3, 0)--cycle; }
\foreach \x in {5,...,14} {
	\fill[teal, opacity = .2] (\x *0.3, 17 *0.3) --++(1*.3,1*.3)--++(-1 *.3, 0)--cycle; }
\foreach \x in {8,...,12} {
	\fill[teal, opacity = .2] (\x *0.3, 18 *0.3) --++(1*.3,1*.3)--++(-1 *.3, 0)--cycle; }
\draw[teal, opacity = .5]
(6 *0.3, 13 *0.3) --++ (2 * 0.3, 0 * 0.3) --++ (1 * 0.3, 1 * 0.3) --++ (3 * 0.3, 0 * 0.3) --++
(1 * 0.3, 1 * 0.3) --++ (2 * 0.3, 0 * 0.3) --++ (1 * 0.3, 1 * 0.3) --++ (0 * 0.3, 1 * 0.3) --++
(-1 * 0.3, 0 * 0.3) --++ (0 * 0.3, 1 * 0.3) --++ (-2 * 0.3, 0 * 0.3) --++ (0 * 0.3, 1 * 0.3) --++
(-5 * 0.3, 0 * 0.3) --++ (-1 * 0.3, -1 * 0.3) --++ (-2 * 0.3, 0 * 0.3) --++ (-1 * 0.3, -1 * 0.3)--++
(0 * 0.3, -2* 0.3) --++ (1 * 0.3, 0 * 0.3)--++ (0 * 0.3, -1* 0.3) --++ (1 * 0.3, 0 * 0.3) -- cycle;

\foreach \x in {12,...,15} {
	\fill[violet, opacity = .2] (\x *0.3, 7 *0.3) --++(1*.3,1*.3)--++(0,-1 *.3)--cycle; }
\foreach \x in {12,...,17} {
	\fill[violet, opacity = .2] (\x *0.3, 8 *0.3) --++(1*.3,1*.3)--++(0,-1 *.3)--cycle; }
\foreach \x in {12,...,18} {
	\fill[violet, opacity = .2] (\x *0.3, 9 *0.3) --++(1*.3,1*.3)--++(0,-1 *.3)--cycle; }
\foreach \x in {12,...,18} {
	\fill[violet, opacity = .2] (\x *0.3, 10 *0.3) --++(1*.3,1*.3)--++(0,-1 *.3)--cycle; }
\foreach \x in {12,...,18} {
	\fill[violet, opacity = .2] (\x *0.3, 11 *0.3) --++(1*.3,1*.3)--++(0,-1 *.3)--cycle; }
\foreach \x in {12,...,17} {
	\fill[violet, opacity = .2] (\x *0.3, 12 *0.3) --++(1*.3,1*.3)--++(0,-1 *.3)--cycle; }
\foreach \x in {12,...,17} {
	\fill[violet, opacity = .2] (\x *0.3, 13 *0.3) --++(1*.3,1*.3)--++(0,-1 *.3)--cycle; }
\foreach \x in {12,...,16} {
	\fill[violet, opacity = .2] (\x *0.3, 14 *0.3) --++(1*.3,1*.3)--++(0,-1 *.3)--cycle; }
\foreach \x in {13,...,14} {
	\fill[violet, opacity = .2] (\x *0.3, 15 *0.3) --++(1*.3,1*.3)--++(0,-1 *.3)--cycle; }
\foreach \x in {12,...,16} {
	\fill[violet, opacity = .2] (\x *0.3, 7 *0.3) --++(1*.3,1*.3)--++(-1 *.3, 0)--cycle; }
\foreach \x in {12,...,18} {
	\fill[violet, opacity = .2] (\x *0.3, 8 *0.3) --++(1*.3,1*.3)--++(-1 *.3, 0)--cycle; }
\foreach \x in {12,...,18} {
	\fill[violet, opacity = .2] (\x *0.3, 9 *0.3) --++(1*.3,1*.3)--++(-1 *.3, 0)--cycle; }
\foreach \x in {12,...,18} {
	\fill[violet, opacity = .2] (\x *0.3, 10 *0.3) --++(1*.3,1*.3)--++(-1 *.3, 0)--cycle; }
\foreach \x in {12,...,18} {
	\fill[violet, opacity = .2] (\x *0.3, 11 *0.3) --++(1*.3,1*.3)--++(-1 *.3, 0)--cycle; }
\foreach \x in {12,...,17} {
	\fill[violet, opacity = .2] (\x *0.3, 12 *0.3) --++(1*.3,1*.3)--++(-1 *.3, 0)--cycle; }
\foreach \x in {12,...,17} {
	\fill[violet, opacity = .2] (\x *0.3, 13 *0.3) --++(1*.3,1*.3)--++(-1 *.3, 0)--cycle; }
\foreach \x in {13,...,16} {
	\fill[violet, opacity = .2] (\x *0.3, 14 *0.3) --++(1*.3,1*.3)--++(-1 *.3, 0)--cycle; }
\foreach \x in {13,...,14} {
	\fill[violet, opacity = .2] (\x *0.3, 15 *0.3) --++(1*.3,1*.3)--++(-1 *.3, 0)--cycle; }
\draw[violet, opacity =.5] 
(12 * 0.3, 7 *0.3) --++ (4 * 0.3,0 * 0.3) --++ (1 * 0.3,1 * 0.3) --++ (1 * 0.3,0 * 0.3) --++ 
(1 * 0.3,1 * 0.3) --++ (0 * 0.3,3 * 0.3) --++ (-1 * 0.3, 0 * 0.3) --++ (0 * 0.3,2 * 0.3) --++ 
(-1 * 0.3, 0 * 0.3) --++ (0 * 0.3, 1 * 0.3) --++ (-2 * 0.3,0 * 0.3) --++ (0 * 0.3, 1 * 0.3) --++
(-2 * 0.3,0 * 0.3) --++ (0 * 0.3, -1 * 0.3) --++ (-1 * 0.3,-1 * 0.3)--cycle;
\end{scope}
\end{scope}
\end{scope}
\end{tikzpicture}
\caption{On the left: A family of compact convex sets in $\mathbb{R}^2$. On the right: The corresponding $(K_5,1)$-free cover with equivalent intersection pattern}
\end{figure}

As discussed in Examples~\ref{ex:forbiddenRd} and~\ref{ex:coverRd}, we have $K_5\not\prec_H \univ$. Moreover, for every nonempty subfamily $\mathcal{G'}\subset \mathcal{G}$ and every $0\leq j < \dim K_5 = 1$, we have  \[\tilde{\beta}_j\big( \bigcap_{S \in \GG'} S \big) = 0  < 1.\]
Thus $\GG$ is a $(K_5, 1)$-free cover. Moreover, $\GG$ has the same intersection pattern as $\FF$ in the sense that for every integer $k \ge 1$ and every subfamily $\{S_1,S_2, \ldots, S_k\} \subseteq \FF$, we have
\[ \bigcap_{i=1}^k S_i  \neq \emptyset \qquad \Leftrightarrow \qquad \bigcap_{i=1}^k C_{S_i}  \neq \emptyset.\]

\begin{example} Every finite family of compact convex sets in $\RR^d$ is a $(\partial\Delta_{d+1},1)$-free cover. Indeed, by mimicking the previous construction, we obtain a triangulation of the one-point compactification of $\RR^d$ that preserves intersection patterns; the rest follows from Example~\ref{ex:forbiddenRd}. More generally,  every finite \emph{good cover} -- that is, a finite family of open sets in $\RR^d$ for which every nonempty intersection is contractible -- is also a $(\partial\Delta_{d+1},1)$-free cover. Similarly,  Example~\ref{ex:forbiddenRd} implies that every good cover in $\RR^{2k}$ is also a $(\Delta_{2k+2}^{(k)}, 1)$-free cover.\end{example}

\subsection{Main results}
\label{subsec:main results}
It is known that the Helly number of a $(K,b)$-free cover is bounded from above in 
terms of $K$ and~$b$~\cite{hb17}, as is the Radon 
number~\cite[Proposition 3.7]{rb20}.\footnote{The bound on the Helly number of a 
$(K,b)$-free cover directly follows from  a combination of Proposition 30 and 
Lemma 26 in~\cite{hb17}.} 
This implies that $(K,b)$-free covers also satisfy a fractional Helly 
theorem~\cite{boundedRadon_fractHelly} and a $(p,q)$-theorem~\cite{AKMM}. 

However, this approach yields a $(p,q)$-theorem, but only for $p\geq q \geq m$, where for $\dim K >1$, the threshold 
$m$ is an extremely large constant derived from successive iterations of Ramsey's 
theorem~\cite{rb20}.\footnote{For $\dim K = 1$, 
$m$ is in fact linear in $b$ and the bound is optimal for $b=1$~\cite{rb20}.} 
Our main result improves this significantly:

\begin{theorem}[General $(p,q)$-theorem]\label{t:pq for Kbfree}
For every finite simplicial complex $K$ and positive integer $b$, the $(p,q)$-theorem holds for the class of $(K,b)$-free covers for all $p\geq q > \mu(K)$.
\end{theorem}

\noindent 
Let us make four remarks:

\begin{itemize}
\item It should be clear from the discussion of Section~\ref{subsec:setting} that for $K$ the boundary of the $(d+1)$-simplex and~$b=1$, Theorem~\ref{t:pq for Kbfree} already contains the classical $(p,q)$ theorem for convex sets in $\RR^d$. In particular, in that special case the range of pairs $(p,q)$ is sharp.

\item Theorem~\ref{t:pq for Kbfree} implies that for every $p\geq q \geq 3$ and $b\geq 1$, the assertion of the $(p,q)$-theorem holds for any finite family of sets in the plane, such that the intersection of every subfamily has at most $b$ path-connected components (it is a $(K_5,b)$-free cover). This was known before only for open sets and $b=1$ while for $b \geq 2$ one had to assume that $q \geq 2b+2$ \cite{rb20}.

\item The bound on the \emph{piercing number} given by Theorem \ref{t:pq for Kbfree} depends on $p$, $q$, $K$ and $b$, but is independent of the size of the cover.  

\item While the Helly number of a $(K,b)$-free cover can grow with $b$ 
(it is at least $b(\mu(K)+2)$~\cite[Example 2]{hb17}), 
the range for which the $(p,q)$-theorem holds is independent of~$b$. 
This is similar to the gap observed for \emph{convex lattice sets} (sets of the form $C\cap \ZZ^d$ where $C$ is a convex set in $\RR^d$), which have Helly number $2^d$  \cite{doignon1973, scarf1977} and satisfies a $(p,q)$-theorem for every $p \geq q \geq  d+1$~\cite{baranymatousek}. 

\end{itemize}

We now describe our main technical result, which serves as the cornerstone of 
the proof of Theorem \ref{t:pq for Kbfree}. 
For a finite family $\FF$ of subsets of some 
(finite or infinite) ground set, let $\pi_m(\FF)$ denote the number of $m$-element
subfamilies of $\FF$ with nonempty intersection. 

A straightforward counting argument
shows that ``positive density propagates downwards'' in the sense that if $\pi_m(\FF)
\geq \alpha\binom{|\FF|}{m}$, then $\pi_{m-1}(\FF)\geq \alpha \binom{|\FF|}{m-1}$. 
(A sharper bound follows from the Kruskal--Katona theorem.) 
We show that, in essence, positive density propagates \emph{upwards} for 
$(K,b)$-free covers:

\begin{theorem}[Stepping-up theorem]\label{t:basic_propogation}
Fix a simplicial complex $K$, a real number $\delta \in (0,1]$, and integers 
$b\geq 1$ and $m > \mu(K)$. If $\mathcal{F}$ is a sufficiently large $(K,b)$-free 
cover such that $\pi_m(\mathcal{F})\geq \delta \binom{ | \mathcal{F} | }{m}$, 
then $\pi_{m+1}(\FF)\geq \gamma\binom{ | \mathcal{F} |}{m+1}$, where $\gamma > 0$ 
is a constant depending only on $\delta$, $b$, $m$ and $K$. 
\end{theorem}

One immediate application of Theorem~\ref{t:basic_propogation} is the reduction of 
\emph{fractional Helly numbers}. For instance, it easily improves a theorem of 
Pat\'akov\'a~\cite[Theorem~2.3]{rb20} into:\footnote{\cite[Theorem~2.3]{rb20} was not 
phrased in terms of $(K,b)$-free covers but readily generalizes to that setting.
}

\begin{theorem}[Fractional Helly theorem]\label{t:fh}
Let $K$ be a finite simplicial complex and $b \geq 1$ an integer. If a positive fraction
of the $(\mu(K)+1)$-tuples of a $(K,b)$-free cover $\FF$ have nonempty intersection, 
then there is a subfamily with nonempty intersection containing a positive fraction of the members of $\FF$.
\end{theorem}

\noindent
In the terminology of combinatorial convexity, Theorem~\ref{t:fh} states that the 
\emph{fractional Helly number} for $(K,b)$-free covers is at most $\mu(K)+1$. 
Specifically, for any $\alpha > 0$, if a $(K,b)$-free cover $\FF$ satisfies 
$\pi_{\mu(K)+1}(\FF) \geq \alpha \binom{ | \mathcal{F} |}{\mu(K)+1}$, 
then there exists a subfamily $\GG \subset \FF$ of size at 
least $\beta |\FF|$ with nonempty intersection, where $\beta>0$ depends only on 
$K$, $b$, and $\alpha$. 

Theorem~\ref{t:basic_propogation} generalizes a classical observation for convex sets. Theorems~\ref{t:pq for Kbfree} and~\ref{t:fh} generalize and unify a number of results in the area. See Section~\ref{s:outlook} for further discussion.

\section{Proof outline}
\label{sec:proof outline}

We now outline the proof of our main results and introduce some necessary terminology. We write $\NN = \{1,2,\ldots\}$ for the set of positive integers and 
$\NN_0 = \{0,1,\ldots\}$ for the set of non-negative integers.  We write 
$[n] = \{1,2,\ldots, n\}$ and $\binom{[n]}{k}$ for the set of $k$-element 
subsets of~$[n]$. For a simplical (or cell) complex $K$, we denote by $K^{(t)}$ 
the $t$-skeleton of $K$.

\subsection{Improving fractional Helly and $(p,q)$-theorems using stepping-up}

Let us first explain why the stepping-up Theorem~\ref{t:basic_propogation} can be easily combined with the existing literature to yield Theorems~\ref{t:pq for Kbfree} and~\ref{t:fh}.

\subsubsection*{Improved fractional Helly from stepping-up}\hfill\par

A fractional Helly theorem for $(K,b)$-free covers was previously established by the third author~\cite{rb20} with a worse fractional Helly number: the assumption is that a positive fraction
of the $m$-tuples have nonempty intersection, where for $\dim K>1$, $m$ is a hypergraph Ramsey number depending on~$b$ and~$K$. 

To prove Theorem~\ref{t:fh} it therefore suffices to show that if a positive fraction
of the $(\mu(K)+1)$-tuples intersect, then a positive fraction of the $m$-tuples 
intersect. This follows from successive applications of 
Theorem~\ref{t:basic_propogation}. 
(Note that the existence of a finite $m$ established in~\cite[Theorem 2.3]{rb20} 
is essential to ensure this process terminates;
additionally, the implicit bound on $\beta$ derived  
from~\cite[Theorem~2.3]{rb20} changes during this process.)

\subsubsection*{Improved $(p,q)$-theorem from improved fractional Helly}\hfill\par

Theorem~\ref{t:pq for Kbfree} is an immediate consequence of Theorem~\ref{t:fh} and 
the results of Alon et al.~\cite{AKMM}. In essence, Alon et al. show that if an
\emph{intersection closed} set system has fractional Helly number at most $m$, then the
set system also satisfies a $(p,q)$-theorem for the range $p \geq q \geq m$ 
(see ~Theorems~8(i) and~9 and the discussion in $\mathsection 2.1$ in~\cite{AKMM}).
Thus, Theorem~\ref{t:pq for Kbfree} follows by applying Theorem~\ref{t:fh} 
to the intersection closure $\FF^\cap \coloneqq \{ \cap_{S \in \GG} S \colon \GG 
\subset \FF\}$, observing that $\FF^\cap$ is a $(K,b)$-free cover if and only if 
$\FF$ is.

\subsection{Proof outline for the stepping up theorem: reduction to a weak colorful Helly}
\label{subsec:strategy}

Let us now turn our attention to our main technical contribution, the stepping-up Theorem~\ref{t:basic_propogation}. 
At a high-level, it is a consequence of a 
suitable \emph{colorful Helly theorem}, combined with \emph{ hypergraph supersaturation}. 
This approach has been successfully applied to various problems in combinatorial 
geometry, see e.g. \cite{abfk1992}, \cite[chapters 23, 29, 33]{barany2021combinatorial},
\cite{baranymatousek},  \cite[chapter 9]{matousek2013lectures}. We now present its main ingredients.

\subsubsection*{A hypergraph supersaturation argument}\hfill\par

For a more detailed explanation, it is convenient to reformulate the problem in 
terms of hypergraphs. Given a $(K,b)$-free cover $\FF$, we consider a 
hypergraph $\HH$ whose vertices correspond to the members of $\FF$ and whose edges 
correspond to the intersecting $m$-tuples in $\FF$. This makes $\HH$ a \emph{ dense} 
$m$-uniform hypergraph on $|\FF|$ vertices and  $\pi_m(\FF) > \delta\binom{|\FF|}{m}$ 
edges. The Erd\H{o}s--Simonovits (supersaturation) theorem \cite{Erdos-Simonovits} 
implies that there exist at least $\rho\binom{|\FF|}{mt}$ distinct vertex subsets of size
$mt$ that each span a \emph{ complete $m$-partite} subhypergraph with parts of size
$t$. Here $t$ is any fixed positive integer, and crucially, the fraction $\rho > 0$ 
depends only on $m$, $t$, and $\delta$. In other words, the positive edge density of 
$\HH$ ensures a positive density of complete $m$-partite subhypergraphs with parts of 
size $t$. 

Returning to the $(K,b)$-free cover $\FF$, each complete $m$-partite subhypergraph in
$\HH$ corresponds to a collection of disjoint subfamilies $\FF_1, \dots, \FF_m \subset \FF$, with $|\FF_i| = t$, such that $S_1\cap \cdots \cap S_m\neq \emptyset$ for every choice $S_i\in \FF_i$. This is a familiar intersection pattern originating in the \emph{colorful Helly theorem}~\cite{baranys-caratheodory}.

\subsubsection*{Colored $(K,b)$-free covers}\hfill\par

This leads us to introduce a colorful analogue of $(K,b)$-free covers. 
Given a collection of nonempty set systems $\FF_1, \FF_2, \ldots, \FF_m$, 
a subfamily $\GG \subset \bigcup_{i=1}^m \FF_i$ is called 
\emph{colorful} if $\GG$ contains at most one member from each $\FF_i$. 
We define an \emph{$m$-colored $(K,b)$-free cover} in a simplicial complex $\univ$ 
to be a finite family  $\FF = \FF_1 \sqcup \FF_2 \sqcup \cdots \sqcup \FF_m$ of 
subcomplexes of $\univ$ such that:
\begin{enumerate}[(i)]
\item[(i)] $K$ is a forbidden homological minor of $\univ$, and 
\item[(ii)]  $\tilde{\beta}_j(\bigcap_{S \in 
\GG}S) < b$  for all $0\leq j < \dim(K)$ and all nonempty colorful subfamilies $\GG 
\subseteq \FF$. 
 \end{enumerate}

Returning to the relation between the family $\FF$ and the $m$-regular hypergraph $\HH$, every subset of $mt$ vertices that spans a complete $m$-partite subhypergraph with parts of size $t$ corresponds to a subfamily of $\FF$ that can be colored to form an $m$-colored $(K,b)$-free cover where each color class has size~$t$ and every colorful subfamily has nonempty intersection. We prove that each such family has the following property:

\begin{theorem}[weak colorful Helly theorem]\label{t:weak colors}
For any  simplicial complex $K$ and integers $b\geq 1$ and $m> \mu(K)$, 
there exists an integer $t = t(b,K,m)$ with the following property: 
If $\FF = \FF_1 \sqcup \cdots \sqcup \FF_m$ is an $m$-colored $(K,b)$-free cover 
where each color class has size~$t$ 
and every colorful subfamily has nonempty intersection, then $\FF$ contains some
$2m-\mu(K)$ members with nonempty intersection. 
\end{theorem}

Note that the definition of colored $(K,b)$-free cover does not impose any condition on the non-colorful subfamilies, as this is not needed for the proof of Theorem~\ref{t:weak colors}. 

\subsubsection*{Stepping up from weak colorful Helly and supersaturation}\hfill\par

Our main result, Theorem~\ref{t:basic_propogation}, now follows from 
Theorem~\ref{t:weak colors} via double counting and supersaturation. 
Specifically, let $\FF$ be a  $(K,b)$-free cover with at least 
$\delta \binom{|\FF|}{m}$ intersecting $m$-tuples, where $m > \mu(K)$, 
and let $t = t(b,K,m)$ be the integer provided by Theorem \ref{t:weak colors}. 
By the supersaturation argument outlined above, there are at least  
$\rho \binom{|\FF|}{mt}$ $m$-colored subfamilies of $\FF$ that satisfy the 
hypothesis of Theorem \ref{t:weak colors}. 

For each such $m$-colored subfamily, we apply Theorem \ref{t:weak colors} 
to find an intersecting subfamily of size $2m-\mu(K)$. Since $2m-\mu(K) \geq m+1$, 
this implies the existence of an intersecting $(m+1)$-tuple within that subfamily. 
By summing over all such subfamilies and correcting for the number of times a given 
intersecting $(m+1)$-tuple is counted, we conclude that there are at least 
$\gamma \binom{|\FF|}{m+1}$ intersecting $(m+1)$-tuples, 
where $\gamma>0$ depends only on $\delta$, $b$, $m$, and $K$. 

The full details of this argument are provided in Section~\ref{s:wrapup}.

\subsection{Proof outline for the weak colorful Helly theorem}

The proof of Theorem~\ref{t:weak colors} is quite involved and builds on the method 
of \emph{constrained chain maps} developed in~\cite{hb17,rb20}.
A major part of this paper, specifically Sections~\ref{s:sc vkf}, \ref{s:cip homin} and~\ref{s:hominor in grid}, 
is devoted to adapt this machinery to handle the $m$-partite structure of 
colorful intersection patterns. Here we introduce the key ideas on a simple example, then give some indication on how these ideas can be extended to the general case.

\subsubsection*{The example}\hfill\par
Let us consider a simplified situation 
where we are given a 3-colored family $\FF = \FF_1 \sqcup \FF_2 \sqcup \FF_3$ 
of open connected sets in $\mathbb{R}^2$ such that the intersection 
of any colorful subfamily is nonempty and connected. 
We show that if $|\FF_i| =  5$ for every $1\leq i \leq 3$, 
then there exists four members in $\FF$ with nonempty intersection. 

\subsubsection*{Encoding the family in a graph}\hfill\par

The colorful subfamilies of $\FF$ are encoded by a graph $G = (V,E)$, 
defined on the vertex set $V = \{(a_1,a_2,a_3)\in \mathbb{N}^3 : 1\leq a_i \leq 5\}$. 
Vertices $u$ and $v$ are connected by an edge, $e = uv$, 
whenever their Euclidean distance satisfies $|u-v|=1$. 
It is convenient to view $G$ as a 1-dimensional complex realized 
in $\mathbb{R}^3$ with its vertices on the integer lattice, 
and edges parallel to the coordinate axes, 
forming a 3-dimensional grid (see figure \ref{fig:proofsketch}, left).  

We introduce a map $\psi$ that assigns a colorful subfamily of $\FF$ 
to every cell (vertex or edge) in $G$. 
First, we label the sets in $\FF$ such that 
$\FF_j = \{ S_{1,j}, \dots, S_{5,j}\}$ for $1\leq j \leq 3$.  
The mapping $\psi$ is then defined as follows:
\begin{itemize}
    \item[] {\bf For a vertex} $v = (a,b,c)$, $\psi$ selects one set from each family: $\psi(v) = \{S_{a,1}, S_{b,2}, S_{c,3}\}$.
    \item[] {\bf For an edge} $uv$, $\psi$ assigns the intersection of the 
    endpoint values: $\psi(uv) = \psi(u) \cap \psi(v)$.
\end{itemize}
Note that $\psi$ establishes a bijection between the vertices of $G$ 
and the maximal colorful subfamilies of $\FF$, and for any edge $e$, we have $|\psi(e)| = 2$. 
Furthermore, for any two edges $e$ and $e'$, 
the size $|\psi(e) \cup \psi(e')|$ is determined by 
the relative geometric alignment of $e$ and $e'$ in the grid: 
\begin{itemize}
\item[] If $e$ and $e'$ are contained in a \emph{common axis-parallel line}, 
then $\psi(e)  = \psi(e')$, and so $|\psi(e) \cup \psi(e')| = 2$. 

\smallskip
\item[] If $e$ and $e'$ are contained in a \emph{common axis-parallel plane}, 
but not in a common axis-parallel line, then $\psi(e)$ and $\psi(e')$ 
have exactly one member in common, and so $|\psi(e) \cup \psi(e')| = 3$. 
\end{itemize}
Most important for us is the following property:
\begin{itemize}
\item[] \emph{If $e$ and $e'$ are not contained in any common axis-parallel plane, 
then $\psi(e)$ and $\psi(e')$ are disjoint, and so $|\psi(e) \cup \psi(e')| = 4$.}
\end{itemize}

\subsubsection*{(Chain) maps}\hfill\par

To prove that some four members of $\FF$ intersect, we construct maps relating $K_5$, $G$, and $\RR^2$. In this simple example, these could be (PL) maps, but it turns out to be convenient in the general case to work with chain maps. We thus use chain maps also for this simple example.

\[
\begin{tikzcd}
  & C_\bullet(G) \arrow[dr, "f_\bullet"] \\
C_\bullet(K_5) \arrow[ur, "g_\bullet"] \arrow[rr, "h_\bullet"] & & C_\bullet(\mathbb{R}^2) 
\end{tikzcd}
\]

\subsubsection*{The first chain map}\hfill\par

The first is a map $f_\bullet \colon C_\bullet(G) \to C_\bullet(\mathbb{R}^2)$, 
defined as follows:
\begin{itemize}
    \item[] {\bf On vertices (0-chains):} For a vertex $v$ in $G$ we select a representative point $p_v \in \RR^2$ contained in the common intersection of the sets in $\psi(v)$. We set $f_0(v) = p_v$.

    \smallskip
    
    \item[] {\bf On edges (1-chains):} Consider an edge $uv$ in $G$. The points $p_u$ and $p_v$ are both contained in the intersection of the subfamily $\psi(uv)$. Since this intersection is assumed to be connected, there exists a path $\gamma_{uv}$ connecting $p_u$ to $p_v$ entirely within $\bigcap_{S\in \psi(uv)}S$. We set $f_1(uv) = \gamma_{uv}$.
\end{itemize}
The construction of $f_\bullet$ is completed by extending linearly to arbitrary chains.
A key feature of this construction is that for every cell $\sigma \in G$, 
the image $f_\bullet(\sigma)$ is contained within the intersection 
$\bigcap_{S\in \psi(\sigma)} S$. 
This containment property is what we refer to when we say that $f_\bullet$ 
is \emph{constrained} by the colorful family $\FF$ (see Figure \ref{fig:proofsketch}).

\subsubsection*{The second chain map}\hfill\par

We now define the second chain map $g_\bullet \colon K_5 \to G$, which sends the 
vertices $v_1, \dots, v_5$ of $K_5$ to the vertices $(1,1,1), \dots, (5,5,5)$ 
along the main diagonal in the geometric representation of $G$. 
Each edge $v_iv_j$ in $K_5$ is mapped to a specific path in $G$ connecting the 
corresponding diagonal vertices. These paths are constructed to satisfy the following 
crucial geometric property: 
\begin{itemize}
\item[] \emph{For any two disjoint edges $e$ and $e'$ in $K_5$,  
if an axis-parallel plane $P$ contains an edge $\tau$ from the path $g_1(e)$, 
then $P$ does not contain any edge $\tau'$ from the path $g_1(e')$.}      
\end{itemize}
We omit the explicit construction of these paths for the moment, 
referring instead to Figure \ref{fig:proofsketch}.

\subsubsection*{The new intersection appears}\hfill\par

To find four intersecting members in $\FF$, we compose $f_\bullet$ with $g_\bullet$ 
to obtain the chain map $h_\bullet \colon C_\bullet(K_5) \to C_\bullet(\mathbb{R}^2)$ as in the diagram above. Because $K_5$ is a forbidden homological minor of $\mathbb{R}^2$, $h_\bullet$ must result in two disjoint 
edges $e$ and $e'$ in $K_5$ whose images $h_1(e)$ and $h_1(e')$ overlap. 
Since $h_\bullet$ factors through $C_\bullet(G)$, there must exist specific edges 
$\tau$ from the path $g_1(e)$ and $\tau'$ from the path $g_1(e')$ whose images 
$f_1(\tau)$ and $f_1(\tau')$ intersect at some point $p\in \mathbb{R}^2$. 
Recalling that $f_\bullet$ is constrained by $\FF$, the point $p$ must lie in the 
intersection of all the sets in the subfamily  $\psi(\tau) \cup \psi(\tau')$. By the geometric property highlighted above, these grid edges $\tau$ and $\tau'$ do not lie in any common axis-parallel plane; thus $|\psi(\tau) \cup \psi(\tau')| = 4$. Therefore $p$ is contained in the intersection of four distinct members of $\FF$.

\begin{figure}
\centering
\begin{tikzpicture}

\begin{scope}[scale=.8] 

\begin{scope}
\fill[red, opacity =0.2] (8*0.3, 8*0.3) --++ (0*.3,-1*.3)--++ (1*.3,0*.3) --++ (0*.3,-1*.3)--++ (5*.3,0*.3) --++ (0*.3,-1*.3)--++ (1*.3,0*.3) --++ (0*.3,-1*.3)--++ (1*.3,0*.3) --++ (0*0.3, 2*0.3) --++ (1*0.3, 1*0.3) --++ (0*0.3, 2*0.3) --++ (1*0.3, 1*0.3) --++ (0*0.3, 6*0.3) --++ (-3*0.3, 0*0.3) --++ (0*0.3, 1*0.3) --++ (-1*0.3, 0*0.3) --++ (0*0.3, 1*0.3) --++ (-1*0.3, 0*0.3) --++ (-5*0.3, -5*0.3) --cycle;
\fill[white](11*0.3,11*0.3) --++ (3*.3,0*.3) --++ (0*0.3, -2*0.3) --++ (2*.3,0*.3) --++ (0*0.3, 6*0.3) --++ (-1*.3,0*.3) --cycle;
\end{scope}

\begin{scope} 
\foreach \x in {1,...,20} {
\foreach \y in {1,...,20} {
\fill[black, opacity = .17] (\x *0.3, \y *0.3) circle (0.015cm) ; } }
\foreach \x in {1,...,19} {
\foreach \y in {1,...,19} {
\draw[black, opacity = .1, line width = .1] (\x *0.3, \y *0.3) --++ (0.3,0.3); } }		
\foreach \x in {1,...,19} {
\foreach \y in {1,...,20} {
\draw[black, opacity = .1, line width = .1] (\x *0.3, \y *0.3) --++ (0.3,0); } }
\foreach \x in {1,...,20} {
\foreach \y in {1,...,19} {
\draw[black, opacity = .1, line width = .1] (\x *0.3, \y *0.3) --++ (0,0.3); } }	
\end{scope}

\begin{scope} 
\draw[teal, opacity =0.5] (4*0.3, 5*0.3) --++ (5*.3,0*.3) --++ (2*.3,2*.3)
--++ (0*.3,2*.3) --++ (-1*.3,0*.3) --++ (0*.3,3*.3) --++ (1*.3,1*.3) --++ (0*.3,1*.3)
--++ (-1*.3,0*.3)--++ (0*.3,1*.3)--++ (-1*.3,0*.3)--++ (0*.3,1*.3) 
--++ (-2*.3,0*.3) --++ (-2*.3,-2*.3) --++ (-1*.3,0*.3) -- cycle;
\fill[teal, opacity =0.2] (4*0.3, 5*0.3) --++ (5*.3,0*.3) --++ (2*.3,2*.3)
--++ (0*.3,2*.3) --++ (-1*.3,0*.3) --++ (0*.3,3*.3) --++ (1*.3,1*.3) --++ (0*.3,1*.3)
--++ (-1*.3,0*.3)--++ (0*.3,1*.3)--++ (-1*.3,0*.3)--++ (0*.3,1*.3) 
--++ (-2*.3,0*.3) --++ (-2*.3,-2*.3) --++ (-1*.3,0*.3) -- cycle;

\draw[red, opacity =0.5] (8*0.3, 8*0.3) --++ (0*.3,-1*.3)--++ (1*.3,0*.3) --++ (0*.3,-1*.3)--++ (5*.3,0*.3) --++ (0*.3,-1*.3)--++ (1*.3,0*.3) --++ (0*.3,-1*.3)--++ (1*.3,0*.3) --++ (0*0.3, 2*0.3) --++ (1*0.3, 1*0.3) --++ (0*0.3, 2*0.3) --++ 
(1*0.3, 1*0.3) --++ (0*0.3, 6*0.3) --++ (-3*0.3, 0*0.3) --++ (0*0.3, 1*0.3) --++ 
(-1*0.3, 0*0.3) --++ (0*0.3, 1*0.3) --++ (-1*0.3, 0*0.3) --++ (-5*0.3, -5*0.3) --cycle 
(11*0.3,11*0.3) --++ (3*.3,0*.3) --++ (0*0.3, -2*0.3)
--++ (2*.3,0*.3) --++ (0*0.3, 6*0.3) --++ (-1*.3,0*.3)
--cycle;

\draw[blue, opacity =0.5] (7*0.3, 3*0.3) --++ (0*0.3, 3*0.3) --++ (4*0.3, 4*0.3)  
--++ (4*0.3, 0*0.3) --++ (0*0.3, -3*0.3) --++ (-2*0.3, 0*0.3) --++ (-4*0.3, -4*0.3) --cycle;
\fill[blue, opacity =0.2] (7*0.3, 3*0.3) --++ (0*0.3, 3*0.3) --++ (4*0.3, 4*0.3)  
--++ (4*0.3, 0*0.3) --++ (0*0.3, -3*0.3) --++ (-2*0.3, 0*0.3) --++ (-4*0.3, -4*0.3) --cycle;

\draw[blue, opacity =0.5] (7*0.3,13*0.3) --++(2*0.3,0*0.3) --++(0*0.3,-1*0.3) --++(4*0.3,0*0.3) --++(0*0.3, 3*0.3) --++(3*0.3,3*0.3) --++(-1*0.3,0*0.3) --++(0*0.3,1*0.3) --++(-2*0.3,0*0.3) --++(-2*0.3,-2*0.3) --++(-2*0.3,0*0.3) --++(-2*0.3,-2*0.3) --cycle;
\fill[blue, opacity =0.2] (7*0.3,13*0.3) --++(2*0.3,0*0.3) --++(0*0.3,-1*0.3) --++(4*0.3,0*0.3) --++(0*0.3, 3*0.3) --++(3*0.3,3*0.3) --++(-1*0.3,0*0.3) --++(0*0.3,1*0.3) --++(-2*0.3,0*0.3) --++(-2*0.3,-2*0.3) --++(-2*0.3,0*0.3) --++(-2*0.3,-2*0.3) --cycle;


\draw[opacity = .7, dash pattern=on 0.7pt off 0.7pt, line cap=rect, line width=0.44pt] (10 *0.3, 7 *0.3) --++ (0 *0.3, 1 *0.3) --++ (-1 *0.3, 0 *0.3) --++ (0 *0.3, 4 *0.3) --++ (1 *0.3, 1 *0.3);
\fill[white, opacity =.7] 
(10 *0.3, 7 *0.3) circle (0.03cm)
(10 *0.3, 13 *0.3) circle (0.03cm);
\draw[black, opacity = .7] 
(10 *0.3, 7 *0.3) circle (0.03cm)
(10 *0.3, 13 *0.3) circle (0.03cm);
\end{scope}

\begin{scope}[xshift = -9.5cm, yshift = 1.3cm, scale=.9] 
\tikzmath{\x1 = 1;  \y1 = -2/9; \x2 =4/5; \y2 =2/7; \z1= 9/11; } 

\foreach \b in {0,...,4}{
    \foreach \a in {0,...,4} {
        \draw[line width = 0.15, opacity = .7] (\a*\x2 , \b*\z1+\a*\y2) --++ (4*\x1,4*\y1);}}

\foreach \b in {0,...,4}{
    \foreach \a in {0,...,4} {
        \draw[line width = 0.15, opacity = .7] (\a*\x1 , \b*\z1+\a*\y1) --++ (4*\x2,4*\y2);}}

\foreach \b in {0,...,4}{
    \foreach \a in {0,...,4} {
        \draw[line width = 0.15, opacity = .7] (\a*\x2+\b*\x1 , \a*\y2+\b*\y1) --++ (0,4*\z1);}}

\foreach \b in {1,...,4}{
    \foreach \a in {1,...,\b}{
        \draw[blue, line width = 0.89pt]
    (4*\x1 + 4*\x2 -\b*\x1 -\b*\x2, 4*\y1 + 4*\y2 + 4*\z1 - \b*\y1 -\b*\y2 -\b*\z1) 
    --++(0*\x1, \a*\z1) --++(\a*\x2, \a*\y2) --++(\a*\x1, \a*\y1);} }

\foreach \a in {0,...,4}{
    \foreach \b in {0,...,4}{
        \foreach \c in {0,...,4}{
        \fill[black] (\a*\x1 + \b*\x2, \a*\y1 + \b*\y2 + \c*\z1) circle (0.047cm);} } }
\end{scope}

\end{scope}

    \end{tikzpicture}
    \caption{On the left: A planar projection of the graph $G$ from the $[5]\times [5] \times [5]$ grid in $\mathbb{Z}^3$. The  edges of $K_5$ are mapped by $g_\bullet$ to the blue paths. On the right: Connecting points $p_u$ and $p_v$ by a constrained path $\gamma_{uv}$}
    \label{fig:proofsketch}
\end{figure}

\subsubsection*{Roadmap for the general case}\hfill\par
Stepping back, we see that the only property of $\mathbb{R}^2$ that we used above 
was that $K_5$ is a forbidden homological minor. In fact, the same argument demonstrates that 
for $K = K_5$, $b=1$, and $m=3$, the parameter $t(b,K,m)$ in Theorem 
\ref{t:weak colors} satisfies $t\leq 5$.
The proof for general $K$, $b$, and $m$ (given in Section \ref{s:wrapup}) follows the same basic structure. Let us highlight the main steps that need to be established along the way: 

\subsubsection*{The grid complex}\hfill\par
The colorful subfamilies of $\FF$ are encoded by a suitable \emph{grid complex} whose vertices lie on the integer grid $\mathbb{Z}^m$. We introduce the general grid complex $G[n]^m$ in Section \ref{grid section}, where we also establish terminology and basic properties that will be used throughout the paper. The crucial notion of a \emph{subgrid} is introduced in Section \ref{ss:defsubgrid}, while the encoding of colorful subfamilies via grid complexes is given in Section \ref{ss:grid encoding}. 

\subsubsection*{The first chain map: Constrained chain maps}\hfill\par
Our goal is to build chain maps as in the diagram from the proof sketch above, though certain technical obstacles arise that were absent in the simplified setting.

For instance, when $b>1$, the intersection of a colorful subfamily may have several connected components, making it impossible to find a constrained path $\gamma_{uv}$ as we did previously. However, since we have an absolute bound on the number of connected components (bounded by $b$), we can ensure -- provided the color classes $\FF_i$ are sufficiently large -- that many of the vertex images lie in the same connected component. This allows us to construct \emph{some} constrained paths. 

Our strategy is to construct enough paths (as well as higher-dimensional chains) to form a sufficiently large subgrid $G'$ from which we can construct a chain map $f_\bullet \colon C_\bullet(G') \to C_\bullet(\univ)$ that is \emph{constrained} by $\FF$. This notion is formally defined in Section \ref{s:ccd}.

The existence of such a constrained chain map (Lemma \ref{l:colored constrained}) is achieved using Ramsey-type arguments -- specifically, the Gallai--Witt Theorem~\cite{witt} on monochromatic subsets of $\mathbb{Z}^m$ -- which we use to establish a vanishing lemma for subgrids in Section \ref{s:subgrid}.  We then prove the existence of constrained chain maps in Section \ref{ss:exist_ccm}.

\subsubsection*{The second chain map: Generic chain maps into grid complexes}\hfill\par
To complete the diagram, we also require a chain map $g_\bullet \colon C_\bullet(K) 
\to C_\bullet( G')$. This map should be \emph{generic} in the sense that disjoint 
faces of $K$ are mapped to chains with no common alignments along axis-parallel 
hyperplanes. As in the sketch above, this allows us to control the size of the union of the colorful subfamilies represented by different cells $\sigma$ and $\tau$ in $G'$. See Section \ref{subsec:generic chain maps} for the precise statements (Lemma \ref{l:K_generic_grid_minor} and Corollary \ref{t:cubical_and_excluded}). 

We give a canonical method for constructing $g_\bullet$ based on the \emph{stair 
convexity} of Bukh et al.~\cite{stairConv}. Stair convexity was originally introduced 
to study extremal problems on point configurations in $\mathbb{R}^d$; 
to our knowledge, it has not 
previously been used in the context of grid complexes and chain maps. This adaptation is defined in Section \ref{ss:stair convex chains}.

Our application of stair convexity requires 
(a somewhat tedious) verification that it is compatible with the boundary 
operator in grid complexes. This verification and the construction of $g_\bullet$ 
are carried out in Sections \ref{subsec:nonrecursive} and \ref{s:boundary}.

\section{Grid complexes and encoding of colorful subfamilies}\label{s:sc vkf}
In this section, we introduce the general grid complex $G[n]^m$. 
Combinatorially, we use it to encode the colorful subfamilies of an 
$m$-colored family $\FF = \FF_1 \sqcup \cdots \sqcup \FF_m$,
where $|\FF_i| = n$ for every $1\leq i \leq m$. 
Topologically, it induces a chain complex that plays a crucial role in
constructing the chain maps needed for the proof of Theorem 
\ref{t:weak colors}.

\subsection{Grid complexes} \label{grid section}
Let $G[n]$ denote the $1$-dimensional cell complex whose vertices
($0$-cells) are the singletons $\{1\}, \{2\}, \dots , \{n\}$ and whose
closed $1$-cells are the closed intervals $[1,2], [2,3], \dots, [n-1, n]$.
For $m\geq 1$, define the \emph{grid complex} $G[n]^m$ as the $m$-fold product
\[G[n]^m \coloneqq \underbrace{G[n] \times \cdots \times G[n]}_{\text{$m$-fold}},\]
equipped with the product topology. We can think of $G[n]^m$ abstractly, or geometrically realized in $\RR^m$ as a collection of unit cubes with vertices on the integer lattice $\mathbb{Z}^m$. Their union forms an $m$-dimensional axis-parallel cube with side length $n-1$.
\begin{remark*}
We observe that our grid complex $G[n]^m$
is a particular instance of a cubical complex
(see e.g. Davis~\cite{davis} or Kaczynski~et~al.~\cite{kaczynski2006computational}), 
and therefore naturally gives rise to a chain complex.
While this material is standard, we review it here to
establish the notation and terminology used throughout 
the paper. 
\end{remark*}

\subsubsection*{Cells and chains}\hfill\par

For every integer $a \in [n]$ we use interchangeably the notations $[a,a] = \{a\}$ to denote the corresponding $0$-cell in $G[n]$. For
all integers $a,b \in [n]$ with $a <b$, we let $[a,b] = [b,a]$ denote
the 1-chain with $\ZZ_2$ coefficients
\[\left.\begin{array}{r} {[a,b]}\\{[b,a]}\end{array}\right\}
\coloneqq [a,a+1] + [a+1,a+2] + \cdots + [b-1,b].\]
Notice that for any \emph{distinct} integers $a,b,c \in [n]$ we have
$[a,c] = [a,b] + [b,c]$, as we work with $\ZZ_2$-coefficients. 
Every $k$-cell $\sigma$ in $G[n]^m$ can be written as the
product of exactly $(m-k)$ $0$-cells and $k$ $1$-cells
\[\sigma = [a_1,b_1] \times [a_2,b_2] \times \cdots \times [a_m,b_m],\]
where $1\leq a_i \leq b_i \leq a_i+1 \leq n$.
When the $i$-th factor of a cell $\sigma$ is a $0$-cell, 
we say that the $i$-th coordinate of $\sigma$ is \emph{constant}.
Note that $G[n]^m$ is a regular cell complex of dimension~$m$.

The terminology introduced for simplicial complexes carries 
directly over to grid complexes.
A \emph{$k$-chain} is a sum of $k$-cells in $G[n]^m$ with coefficients in $\ZZ_2$. 
The \emph{support} of a chain $\sigma$, denoted $\supp(\sigma)$, 
is the set of cells with nonzero coefficients in $\sigma$, and two chains, 
$\sigma$ and $\tau$, have \emph{overlapping supports} if there is a 
cell in $\supp(\sigma)$ which intersects a cell in $\supp(\tau)$.
We formulate the following simple observation for future reference.

\begin{observation}\label{o:hyperplane}
For any cells $\sigma, \tau$ in $G[n]^m$ such that 
$\dim \sigma + \dim \tau < m$, there is at least one coordinate that is constant 
for both. Moreover, if  $\sigma$ and $\tau$ intersect, 
then they must be contained in a common axis-parallel hyperplane 
$\{(x_1, \dots, x_m)\in \mathbb{R}^m : x_i = c\}$.
\end{observation}
\begin{proof}
This follows from the Pigeonhole Principle. A $k$-dimensional cell of
$G[n]^m$ has exactly $m-k$ constant coordinates. Since $\dim \sigma + \dim \tau < m$,
the sum of the number of constant coordinates for $\sigma$ and $\tau$ exceeds $m$.
Thus, they must share a constant coordinate index. If the cells intersect, the value
at this coordinate must be identical, placing them in the same hyperplane.
\end{proof}

\subsubsection*{Products and boundaries}\hfill\par

The usual (Cartesian) \emph{product} $\times$ of a $k_1$-cell of $G[n]^{m_1}$ and a
$k_2$-cell of $G[n]^{m_2}$ is a $(k_1+k_2)$-cell of
$G[n]^{m_1+m_2}$. We extend this operation to chains by putting
\[ (\sigma_1 + \cdots + \sigma_{\ell_1}) \times (\tau_1 + \cdots + \tau_{\ell_2}) \; \coloneqq \; \sum_{i=1}^{\ell_1}\sum_{j=1}^{\ell_2} \sigma_i \times \tau_j.\]
We denote the null chain (with empty support) by $0$ and clarify that
for any chain $\sigma$ we have $\sigma \times 0 = 0 \times \sigma =
0$. We can now define the \emph{boundary} of a cell of $G[n]^m$
recursively, as follows (working over $\mathbb{Z}_2$):

\begin{center}
\begin{tabular}{rrcl}
  {\bf ($0$-cells)}  & $\partial \{a\}$ & $\coloneqq$ & $0$ \qquad\qquad\qquad\qquad (the null chain)\\
  {\bf ($1$-cells)} & $\partial [a,a+1]$ & $\coloneqq$ & $\{a\}+\{a+1\}$\\
  {\bf ($\ge 2$-cells)} & $\partial(\sigma \times \tau)$ & $\coloneqq$ & $\partial \sigma \times \tau + \sigma \times \partial \tau$\\
\end{tabular}
\end{center}

\noindent
The definition of $\partial$ extends from $k$-cells to $k$-chains by linearity. 
A simple induction on the dimension yields that $\partial \circ \partial = 0$.
In fact, $\partial$ coincides with the standard boundary operator on $G[n]^m$ 
when viewed as a regular cell complex. 
For  a subcomplex $G$ of a  grid complex, we write $C_\bullet\pth{G}$ 
for the chain complex defined by the chains of $G$
equipped with $\partial$.

\subsection{Subgrids} \label{ss:defsubgrid}
Here we formalize the notion of a subgrid. 
Given integers $1\leq \ell \leq n$ and $m\geq 1$, 
we define a \emph{subgrid in $G[n]^m$ of size $\ell$} to be a  vertex map
\[\Gamma \colon V(G[\ell]^m)\to V(G[n]^m)\]
of the form $(x_1, \dots, x_m) \longmapsto (\gamma_1(x_1), \dots, \gamma_m(x_m))$, 
where $\gamma_1, \dots, \gamma_m$ are \emph{strictly increasing} 
functions from $[\ell]$ to $[n]$. 

A subgrid $\Gamma$ in $G[n]^m$ of size $\ell$ induces a natural chain map 
\[\Gamma_\bullet: C_\bullet(G[\ell]^m) \to C_\bullet(G[n]^m),\] 
defined as follows.
Let $\gamma_1, \dots, \gamma_m$ be the increasing functions that define $\Gamma$. 
For any $a,b \in [\ell]$ we set
$\gamma_i(\{a\}) \coloneqq \{\gamma_i(a)\}$, 
$\gamma_i(\{a,b\}) \coloneqq \{\gamma_i(a),\gamma_i(b)\}$, 
and $\gamma_i([a,b]) \coloneqq [\gamma_i(a),\gamma_i(b)]$. 

For a cell $\sigma = \sigma_1 \times \cdots \times \sigma_m$ in $G[\ell]^m$, 
we define
\[ \Gamma_\bullet(\sigma) = \gamma_1(\sigma_1) \times \cdots \times \gamma_m(\sigma_m) \]
and extend this linearly. 
This gives a chain map because for any cell 
$\sigma = \sigma_1 \times \dots \times \sigma_m$ of $G[\ell]^m$, 
letting $S \coloneqq \{i : \dim \sigma_i = 1\}$, we have (working with $\mathbb{Z}_2$ coefficients)
\[\begin{array}{rclcl}
  \Gamma_\bullet( \partial \sigma )  &   =  &  \Gamma_\bullet \pth{\displaystyle\sum_{i=1}^m \sigma_1 \times \cdots \times \partial \sigma_i \times \cdots \times \sigma_m } & & \\[2ex]
  &   =  &  \Gamma_\bullet\pth{\displaystyle\sum_{i \in S} \sigma_1 \times \cdots \times \partial \sigma_i \times \cdots \times \sigma_m } & & \\[2ex]
    & = & \displaystyle \sum_{i\in S} \gamma_1(\sigma_1) \times \cdots \times \gamma_i(\partial \sigma_i) \times \cdots \times \gamma_m(\sigma_m) & & \\[2ex]
    & = & \displaystyle \sum_{i\in S} \gamma_1(\sigma_1) \times \cdots \times \partial\gamma_i(\sigma_i) \times \cdots \times \gamma_m(\sigma_m) & = & \partial\Gamma_\bullet(\sigma).
\end{array}\]

\subsection{Grid encoding of colorful subfamilies} \label{ss:grid encoding}
Fix positive integers $m$, $t$, and a simplicial complex $\univ$. 
Let $\FF = \FF_1 \sqcup \cdots \sqcup \FF_m$ 
be an $m$-colored family of subcomplexes of $\univ$ 
where each color class $\FF_i$ has size $t$. 
(Note that, for the moment, we do not make any assumptions on 
colorful intersections or 
on forbidden homological minors in $\univ$.)

The colorful subfamilies in $\FF$ will be encoded by a grid complex 
which we denote as $\grid_\FF \coloneqq G[t]^m$. 
The encoding goes as follows. 
Label the members of each $\FF_j$ arbitrarily as 
$\FF_j = \{S_{1,j}, \dots, S_{t,j}\}$. 
We associate to each set $S_{i,j}$ the hyperplane 
$\Pi_j(i) = \{(x_1, \dots, x_m)\in \mathbb{R}^m : x_j=i\}$, 
and for a subset $Y \subset \grid_\FF$ we put 
\[\psi(Y) \coloneqq \{S_{i,j} \colon Y \subset \Pi_j(i) \}.\] 
This defines a bijection $A \mapsto \psi(A)$ between the axis-parallel $k$-flats 
intersecting $V(\grid_\FF)$ and the colorful subfamilies of $\FF$ of size $m-k$. 
(See Figure~\ref{fig:g(v)}.)

\begin{figure}[h]
\begin{center}
\begin{tikzpicture}
\tikzmath{\x1 = 7/8; \x2 =-1/4; \y1 = 3/4; \y2 =3/8; \z1= 0; \z2=15/16; } 

\begin{scope}[scale = .8]


\draw[-latex] (0,0) -- (2*\x1, 2*\x2);
\draw[-latex] (0,0) -- (2*\y1, 2*\y2);
\draw[-latex] (0,0) -- (2*\z1, 2*\z2);


\fill[opacity = .2, blue] (1*\x1+2*\y1+1*\z1, 1*\x2+2*\y2+1*\z2) 
--++ (3*\x1+0*\y1+0*\z1, 3*\x2+0*\y2+0*\z2) 
--++ (0*\x1+0*\y1+3*\z1, 0*\x2+0*\y2+3*\z2)
--++ (-3*\x1+0*\y1+0*\z1, -3*\x2+0*\y2+0*\z2) ; 


\draw[opacity=.35, red, line width = 3.1] (3*\x1+1*\y1+2*\z1, 3*\x2+1*\y2+2*\z2) --++ (0*\x1+3*\y1+0*\z1, 0*\x2+3*\y2+0*\z2);


\draw[green, opacity =.8, line width = 1.1] (2*\x1+3*\y1+4*\z1, 2*\x2+3*\y2+4*\z2) circle (3pt);
\filldraw[green, opacity = 0.8] (2*\x1+3*\y1+4*\z1, 2*\x2+3*\y2+4*\z2) circle (2pt);


\draw[line width = 0.6, opacity=0.8] (1*\x1+1*\y1+1*\z1, 1*\x2+1*\y2+1*\z2) --++ (3*\x1+0*\y1+0*\z1, 3*\x2+0*\y2+0*\z2);
\draw[line width = 0.6, opacity=0.8] (1*\x1+1*\y1+2*\z1, 1*\x2+1*\y2+2*\z2) --++ (3*\x1+0*\y1+0*\z1, 3*\x2+0*\y2+0*\z2);
\draw[line width = 0.6, opacity=0.8] (1*\x1+1*\y1+3*\z1, 1*\x2+1*\y2+3*\z2) --++ (3*\x1+0*\y1+0*\z1, 3*\x2+0*\y2+0*\z2);
\draw[line width = 0.6, opacity=0.8] (1*\x1+1*\y1+4*\z1, 1*\x2+1*\y2+4*\z2) --++ (3*\x1+0*\y1+0*\z1, 3*\x2+0*\y2+0*\z2);

\draw[line width = 0.2, opacity=0.2] (1*\x1+2*\y1+1*\z1, 1*\x2+2*\y2+1*\z2) --++ (3*\x1+0*\y1+0*\z1, 3*\x2+0*\y2+0*\z2);
\draw[line width = 0.2, opacity=0.3] (1*\x1+2*\y1+2*\z1, 1*\x2+2*\y2+2*\z2) --++ (3*\x1+0*\y1+0*\z1, 3*\x2+0*\y2+0*\z2);
\draw[line width = 0.4, opacity=0.6] (1*\x1+2*\y1+3*\z1, 1*\x2+2*\y2+3*\z2) --++ (3*\x1+0*\y1+0*\z1, 3*\x2+0*\y2+0*\z2);
\draw[line width = 0.6, opacity=0.8] (1*\x1+2*\y1+4*\z1, 1*\x2+2*\y2+4*\z2) --++ (3*\x1+0*\y1+0*\z1, 3*\x2+0*\y2+0*\z2);

\draw[line width = 0.2, opacity=0.2] (1*\x1+3*\y1+1*\z1, 1*\x2+3*\y2+1*\z2) --++ (3*\x1+0*\y1+0*\z1, 3*\x2+0*\y2+0*\z2);
\draw[line width = 0.2, opacity=0.3] (1*\x1+3*\y1+2*\z1, 1*\x2+3*\y2+2*\z2) --++ (3*\x1+0*\y1+0*\z1, 3*\x2+0*\y2+0*\z2);
\draw[line width = 0.4, opacity=0.6] (1*\x1+3*\y1+3*\z1, 1*\x2+3*\y2+3*\z2) --++ (3*\x1+0*\y1+0*\z1, 3*\x2+0*\y2+0*\z2);
\draw[line width = 0.6, opacity=0.8] (1*\x1+3*\y1+4*\z1, 1*\x2+3*\y2+4*\z2) --++ (3*\x1+0*\y1+0*\z1, 3*\x2+0*\y2+0*\z2);

\draw[line width = 0.2, opacity=0.2] (1*\x1+4*\y1+1*\z1, 1*\x2+4*\y2+1*\z2) --++ (3*\x1+0*\y1+0*\z1, 3*\x2+0*\y2+0*\z2);
\draw[line width = 0.2, opacity=0.3] (1*\x1+4*\y1+2*\z1, 1*\x2+4*\y2+2*\z2) --++ (3*\x1+0*\y1+0*\z1, 3*\x2+0*\y2+0*\z2);
\draw[line width = 0.4, opacity=0.6] (1*\x1+4*\y1+3*\z1, 1*\x2+4*\y2+3*\z2) --++ (3*\x1+0*\y1+0*\z1, 3*\x2+0*\y2+0*\z2);
\draw[line width = 0.6, opacity=0.8] (1*\x1+4*\y1+4*\z1, 1*\x2+4*\y2+4*\z2) --++ (3*\x1+0*\y1+0*\z1, 3*\x2+0*\y2+0*\z2);
\draw[line width = 0.2, opacity=0.3] (1*\x1+1*\y1+1*\z1, 1*\x2+1*\y2+1*\z2) --++ (0*\x1+3*\y1+0*\z1, 0*\x2+3*\y2+0*\z2);
\draw[line width = 0.2, opacity=0.3] (2*\x1+1*\y1+1*\z1, 2*\x2+1*\y2+1*\z2) --++ (0*\x1+3*\y1+0*\z1, 0*\x2+3*\y2+0*\z2);
\draw[line width = 0.2, opacity=0.3] (3*\x1+1*\y1+1*\z1, 3*\x2+1*\y2+1*\z2) --++ (0*\x1+3*\y1+0*\z1, 0*\x2+3*\y2+0*\z2);
\draw[line width = 0.6, opacity=0.8] (4*\x1+1*\y1+1*\z1, 4*\x2+1*\y2+1*\z2) --++ (0*\x1+3*\y1+0*\z1, 0*\x2+3*\y2+0*\z2);

\draw[line width = 0.3, opacity=0.4] (1*\x1+1*\y1+2*\z1, 1*\x2+1*\y2+2*\z2) --++ (0*\x1+3*\y1+0*\z1, 0*\x2+3*\y2+0*\z2);
\draw[line width = 0.3, opacity=0.4] (2*\x1+1*\y1+2*\z1, 2*\x2+1*\y2+2*\z2) --++ (0*\x1+3*\y1+0*\z1, 0*\x2+3*\y2+0*\z2);
\draw[line width = 0.3, opacity=0.4] (3*\x1+1*\y1+2*\z1, 3*\x2+1*\y2+2*\z2) --++ (0*\x1+3*\y1+0*\z1, 0*\x2+3*\y2+0*\z2);
\draw[line width = 0.6, opacity=0.8] (4*\x1+1*\y1+2*\z1, 4*\x2+1*\y2+2*\z2) --++ (0*\x1+3*\y1+0*\z1, 0*\x2+3*\y2+0*\z2);

\draw[line width = 0.4, opacity=0.6] (1*\x1+1*\y1+3*\z1, 1*\x2+1*\y2+3*\z2) --++ (0*\x1+3*\y1+0*\z1, 0*\x2+3*\y2+0*\z2);
\draw[line width = 0.4, opacity=0.6] (2*\x1+1*\y1+3*\z1, 2*\x2+1*\y2+3*\z2) --++ (0*\x1+3*\y1+0*\z1, 0*\x2+3*\y2+0*\z2);
\draw[line width = 0.4, opacity=0.6] (3*\x1+1*\y1+3*\z1, 3*\x2+1*\y2+3*\z2) --++ (0*\x1+3*\y1+0*\z1, 0*\x2+3*\y2+0*\z2);
\draw[line width = 0.6, opacity=0.8] (4*\x1+1*\y1+3*\z1, 4*\x2+1*\y2+3*\z2) --++ (0*\x1+3*\y1+0*\z1, 0*\x2+3*\y2+0*\z2);

\draw[line width = 0.6, opacity=0.8] (1*\x1+1*\y1+4*\z1, 1*\x2+1*\y2+4*\z2) --++ (0*\x1+3*\y1+0*\z1, 0*\x2+3*\y2+0*\z2);
\draw[line width = 0.6, opacity=0.8] (2*\x1+1*\y1+4*\z1, 2*\x2+1*\y2+4*\z2) --++ (0*\x1+3*\y1+0*\z1, 0*\x2+3*\y2+0*\z2);
\draw[line width = 0.6, opacity=0.8] (3*\x1+1*\y1+4*\z1, 3*\x2+1*\y2+4*\z2) --++ (0*\x1+3*\y1+0*\z1, 0*\x2+3*\y2+0*\z2);
\draw[line width = 0.6, opacity=0.8] (4*\x1+1*\y1+4*\z1, 4*\x2+1*\y2+4*\z2) --++ (0*\x1+3*\y1+0*\z1, 0*\x2+3*\y2+0*\z2);
\draw[line width = 0.6, opacity=0.8] (1*\x1+1*\y1+1*\z1, 1*\x2+1*\y2+1*\z2) --++ (0*\x1+0*\y1+3*\z1, 0*\x2+0*\y2+3*\z2);
\draw[line width = 0.6, opacity=0.8] (2*\x1+1*\y1+1*\z1, 2*\x2+1*\y2+1*\z2) --++ (0*\x1+0*\y1+3*\z1, 0*\x2+0*\y2+3*\z2);
\draw[line width = 0.6, opacity=0.8] (3*\x1+1*\y1+1*\z1, 3*\x2+1*\y2+1*\z2) --++ (0*\x1+0*\y1+3*\z1, 0*\x2+0*\y2+3*\z2);
\draw[line width = 0.6, opacity=0.8] (4*\x1+1*\y1+1*\z1, 4*\x2+1*\y2+1*\z2) --++ (0*\x1+0*\y1+3*\z1, 0*\x2+0*\y2+3*\z2);

\draw[line width = 0.4, opacity=0.6] (1*\x1+2*\y1+1*\z1, 1*\x2+2*\y2+1*\z2) --++ (0*\x1+0*\y1+3*\z1, 0*\x2+0*\y2+3*\z2);
\draw[line width = 0.4, opacity=0.6] (2*\x1+2*\y1+1*\z1, 2*\x2+2*\y2+1*\z2) --++ (0*\x1+0*\y1+3*\z1, 0*\x2+0*\y2+3*\z2);
\draw[line width = 0.4, opacity=0.6] (3*\x1+2*\y1+1*\z1, 3*\x2+2*\y2+1*\z2) --++ (0*\x1+0*\y1+3*\z1, 0*\x2+0*\y2+3*\z2);
\draw[line width = 0.6, opacity=0.8] (4*\x1+2*\y1+1*\z1, 4*\x2+2*\y2+1*\z2) --++ (0*\x1+0*\y1+3*\z1, 0*\x2+0*\y2+3*\z2);

\draw[line width = 0.2, opacity=0.4] (1*\x1+3*\y1+1*\z1, 1*\x2+3*\y2+1*\z2) --++ (0*\x1+0*\y1+3*\z1, 0*\x2+0*\y2+3*\z2);
\draw[line width = 0.2, opacity=0.4] (2*\x1+3*\y1+1*\z1, 2*\x2+3*\y2+1*\z2) --++ (0*\x1+0*\y1+3*\z1, 0*\x2+0*\y2+3*\z2);
\draw[line width = 0.4, opacity=0.6] (3*\x1+3*\y1+1*\z1, 3*\x2+3*\y2+1*\z2) --++ (0*\x1+0*\y1+3*\z1, 0*\x2+0*\y2+3*\z2);
\draw[line width = 0.6, opacity=0.8] (4*\x1+3*\y1+1*\z1, 4*\x2+3*\y2+1*\z2) --++ (0*\x1+0*\y1+3*\z1, 0*\x2+0*\y2+3*\z2);

\draw[line width = 0.1, opacity=0.1] (1*\x1+4*\y1+1*\z1, 1*\x2+4*\y2+1*\z2) --++ (0*\x1+0*\y1+3*\z1, 0*\x2+0*\y2+3*\z2);
\draw[line width = 0.2, opacity=0.4] (2*\x1+4*\y1+1*\z1, 2*\x2+4*\y2+1*\z2) --++ (0*\x1+0*\y1+3*\z1, 0*\x2+0*\y2+3*\z2);
\draw[line width = 0.4, opacity=0.6] (3*\x1+4*\y1+1*\z1, 3*\x2+4*\y2+1*\z2) --++ (0*\x1+0*\y1+3*\z1, 0*\x2+0*\y2+3*\z2);
\draw[line width = 0.6, opacity=0.8] (4*\x1+4*\y1+1*\z1, 4*\x2+4*\y2+1*\z2) --++ (0*\x1+0*\y1+3*\z1, 0*\x2+0*\y2+3*\z2);


\filldraw[black, opacity = 0.99] (1*\x1+1*\y1+1*\z1, 1*\x2+1*\y2+1*\z2) circle (1.6pt);
\filldraw[black, opacity = 0.89] (1*\x1+2*\y1+1*\z1, 1*\x2+2*\y2+1*\z2) circle (1.4pt);
\filldraw[black, opacity = 0.79] (1*\x1+3*\y1+1*\z1, 1*\x2+3*\y2+1*\z2) circle (1.2pt);
\filldraw[black, opacity = 0.69] (1*\x1+4*\y1+1*\z1, 1*\x2+4*\y2+1*\z2) circle (1.0pt);

\filldraw[black, opacity = 0.99] (1*\x1+1*\y1+2*\z1, 1*\x2+1*\y2+2*\z2) circle (1.6pt);
\filldraw[black, opacity = 0.89] (1*\x1+2*\y1+2*\z1, 1*\x2+2*\y2+2*\z2) circle (1.4pt);
\filldraw[black, opacity = 0.79] (1*\x1+3*\y1+2*\z1, 1*\x2+3*\y2+2*\z2) circle (1.2pt);
\filldraw[black, opacity = 0.69] (1*\x1+4*\y1+2*\z1, 1*\x2+4*\y2+2*\z2) circle (1.0pt);

\filldraw[black, opacity = 0.99] (1*\x1+1*\y1+3*\z1, 1*\x2+1*\y2+3*\z2) circle (1.6pt);
\filldraw[black, opacity = 0.89] (1*\x1+2*\y1+3*\z1, 1*\x2+2*\y2+3*\z2) circle (1.4pt);
\filldraw[black, opacity = 0.79] (1*\x1+3*\y1+3*\z1, 1*\x2+3*\y2+3*\z2) circle (1.2pt);
\filldraw[black, opacity = 0.69] (1*\x1+4*\y1+3*\z1, 1*\x2+4*\y2+3*\z2) circle (1.0pt);

\filldraw[black, opacity = 0.99] (1*\x1+1*\y1+4*\z1, 1*\x2+1*\y2+4*\z2) circle (1.6pt);
\filldraw[black, opacity = 0.89] (1*\x1+2*\y1+4*\z1, 1*\x2+2*\y2+4*\z2) circle (1.4pt);
\filldraw[black, opacity = 0.79] (1*\x1+3*\y1+4*\z1, 1*\x2+3*\y2+4*\z2) circle (1.2pt);
\filldraw[black, opacity = 0.69] (1*\x1+4*\y1+4*\z1, 1*\x2+4*\y2+4*\z2) circle (1.0pt);
\filldraw[black, opacity = 0.99] (2*\x1+1*\y1+1*\z1, 2*\x2+1*\y2+1*\z2) circle (1.6pt);
\filldraw[black, opacity = 0.89] (2*\x1+2*\y1+1*\z1, 2*\x2+2*\y2+1*\z2) circle (1.4pt);
\filldraw[black, opacity = 0.79] (2*\x1+3*\y1+1*\z1, 2*\x2+3*\y2+1*\z2) circle (1.2pt);
\filldraw[black, opacity = 0.69] (2*\x1+4*\y1+1*\z1, 2*\x2+4*\y2+1*\z2) circle (1.0pt);

\filldraw[black, opacity = 0.99] (2*\x1+1*\y1+2*\z1, 2*\x2+1*\y2+2*\z2) circle (1.6pt);
\filldraw[black, opacity = 0.89] (2*\x1+2*\y1+2*\z1, 2*\x2+2*\y2+2*\z2) circle (1.4pt);
\filldraw[black, opacity = 0.79] (2*\x1+3*\y1+2*\z1, 2*\x2+3*\y2+2*\z2) circle (1.2pt);
\filldraw[black, opacity = 0.69] (2*\x1+4*\y1+2*\z1, 2*\x2+4*\y2+2*\z2) circle (1.0pt);

\filldraw[black, opacity = 0.99] (2*\x1+1*\y1+3*\z1, 2*\x2+1*\y2+3*\z2) circle (1.6pt);
\filldraw[black, opacity = 0.89] (2*\x1+2*\y1+3*\z1, 2*\x2+2*\y2+3*\z2) circle (1.4pt);
\filldraw[black, opacity = 0.79] (2*\x1+3*\y1+3*\z1, 2*\x2+3*\y2+3*\z2) circle (1.2pt);
\filldraw[black, opacity = 0.69] (2*\x1+4*\y1+3*\z1, 2*\x2+4*\y2+3*\z2) circle (1.0pt);

\filldraw[black, opacity = 0.99] (2*\x1+1*\y1+4*\z1, 2*\x2+1*\y2+4*\z2) circle (1.6pt);
\filldraw[black, opacity = 0.89] (2*\x1+2*\y1+4*\z1, 2*\x2+2*\y2+4*\z2) circle (1.4pt);
\filldraw[black, opacity = 0.79] (2*\x1+3*\y1+4*\z1, 2*\x2+3*\y2+4*\z2) circle (1.2pt);
\filldraw[black, opacity = 0.69] (2*\x1+4*\y1+4*\z1, 2*\x2+4*\y2+4*\z2) circle (1.0pt);
\filldraw[black, opacity = 0.99] (3*\x1+1*\y1+1*\z1, 3*\x2+1*\y2+1*\z2) circle (1.6pt);
\filldraw[black, opacity = 0.89] (3*\x1+2*\y1+1*\z1, 3*\x2+2*\y2+1*\z2) circle (1.4pt);
\filldraw[black, opacity = 0.79] (3*\x1+3*\y1+1*\z1, 3*\x2+3*\y2+1*\z2) circle (1.2pt);
\filldraw[black, opacity = 0.69] (3*\x1+4*\y1+1*\z1, 3*\x2+4*\y2+1*\z2) circle (1.0pt);

\filldraw[black, opacity = 0.99] (3*\x1+1*\y1+2*\z1, 3*\x2+1*\y2+2*\z2) circle (1.6pt);
\filldraw[black, opacity = 0.89] (3*\x1+2*\y1+2*\z1, 3*\x2+2*\y2+2*\z2) circle (1.4pt);
\filldraw[black, opacity = 0.79] (3*\x1+3*\y1+2*\z1, 3*\x2+3*\y2+2*\z2) circle (1.2pt);
\filldraw[black, opacity = 0.69] (3*\x1+4*\y1+2*\z1, 3*\x2+4*\y2+2*\z2) circle (1.0pt);

\filldraw[black, opacity = 0.99] (3*\x1+1*\y1+3*\z1, 3*\x2+1*\y2+3*\z2) circle (1.6pt);
\filldraw[black, opacity = 0.89] (3*\x1+2*\y1+3*\z1, 3*\x2+2*\y2+3*\z2) circle (1.4pt);
\filldraw[black, opacity = 0.79] (3*\x1+3*\y1+3*\z1, 3*\x2+3*\y2+3*\z2) circle (1.2pt);
\filldraw[black, opacity = 0.69] (3*\x1+4*\y1+3*\z1, 3*\x2+4*\y2+3*\z2) circle (1.0pt);

\filldraw[black, opacity = 0.99] (3*\x1+1*\y1+4*\z1, 3*\x2+1*\y2+4*\z2) circle (1.6pt);
\filldraw[black, opacity = 0.89] (3*\x1+2*\y1+4*\z1, 3*\x2+2*\y2+4*\z2) circle (1.4pt);
\filldraw[black, opacity = 0.79] (3*\x1+3*\y1+4*\z1, 3*\x2+3*\y2+4*\z2) circle (1.2pt);
\filldraw[black, opacity = 0.69] (3*\x1+4*\y1+4*\z1, 3*\x2+4*\y2+4*\z2) circle (1.0pt);
\filldraw[black, opacity = 0.99] (4*\x1+1*\y1+1*\z1, 4*\x2+1*\y2+1*\z2) circle (1.6pt);
\filldraw[black, opacity = 0.89] (4*\x1+2*\y1+1*\z1, 4*\x2+2*\y2+1*\z2) circle (1.4pt);
\filldraw[black, opacity = 0.79] (4*\x1+3*\y1+1*\z1, 4*\x2+3*\y2+1*\z2) circle (1.2pt);
\filldraw[black, opacity = 0.69] (4*\x1+4*\y1+1*\z1, 4*\x2+4*\y2+1*\z2) circle (1.0pt);

\filldraw[black, opacity = 0.99] (4*\x1+1*\y1+2*\z1, 4*\x2+1*\y2+2*\z2) circle (1.6pt);
\filldraw[black, opacity = 0.89] (4*\x1+2*\y1+2*\z1, 4*\x2+2*\y2+2*\z2) circle (1.4pt);
\filldraw[black, opacity = 0.79] (4*\x1+3*\y1+2*\z1, 4*\x2+3*\y2+2*\z2) circle (1.2pt);
\filldraw[black, opacity = 0.69] (4*\x1+4*\y1+2*\z1, 4*\x2+4*\y2+2*\z2) circle (1.0pt);

\filldraw[black, opacity = 0.99] (4*\x1+1*\y1+3*\z1, 4*\x2+1*\y2+3*\z2) circle (1.6pt);
\filldraw[black, opacity = 0.89] (4*\x1+2*\y1+3*\z1, 4*\x2+2*\y2+3*\z2) circle (1.4pt);
\filldraw[black, opacity = 0.79] (4*\x1+3*\y1+3*\z1, 4*\x2+3*\y2+3*\z2) circle (1.2pt);
\filldraw[black, opacity = 0.69] (4*\x1+4*\y1+3*\z1, 4*\x2+4*\y2+3*\z2) circle (1.0pt);

\filldraw[black, opacity = 0.99] (4*\x1+1*\y1+4*\z1, 4*\x2+1*\y2+4*\z2) circle (1.6pt);
\filldraw[black, opacity = 0.89] (4*\x1+2*\y1+4*\z1, 4*\x2+2*\y2+4*\z2) circle (1.4pt);
\filldraw[black, opacity = 0.79] (4*\x1+3*\y1+4*\z1, 4*\x2+3*\y2+4*\z2) circle (1.2pt);
\filldraw[black, opacity = 0.69] (4*\x1+4*\y1+4*\z1, 4*\x2+4*\y2+4*\z2) circle (1.0pt);
\end{scope}

\begin{scope}[xshift = 7cm, scale = 0.8]
\fill[green, opacity=.5, rounded corners=2] (2.25,2.75) rectangle (2.25+0.75,2.75+0.55);
\fill[green, opacity=.5, rounded corners=2] (3.16,1.75) rectangle (3.16+0.75,1.75+0.55);
\fill[green, opacity=.5, rounded corners=2] (4.06,0.75) rectangle (4.06+0.75,0.75+0.55);

\fill[red, opacity=.5, rounded corners=2] (3.16,2.75) rectangle (3.16+0.75,2.75+0.55);
\fill[red, opacity=.5, rounded corners=2] (2.25,0.75) rectangle (2.25+0.75,0.75+0.55);

\fill[blue, opacity=.3, rounded corners=2] (2.25,1.75) rectangle (2.25+0.75,1.75+0.55);
    
\node[right] at (0,3) {$\FF_1 = \{$ };
\node[right] at (1.2,3){$S_{1,1},$};
\node[right] at (2.1,3){$S_{2,1},$};
\node[right] at (3.0,3){$S_{3,1},$};
\node[right] at (3.9,3){$S_{4,1}$};
\node[right] at (4.6,3){$\}$};

\node[right] at (0,2) {$\FF_2 = \{$ };
\node[right] at (1.2,2){$S_{1,2},$};
\node[right] at (2.1,2){$S_{2,2},$};
\node[right] at (3.0,2){$S_{3,2},$};
\node[right] at (3.9,2){$S_{4,2}$};
\node[right] at (4.6,2){$\}$};

\node[right] at (0,1) {$\FF_3 = \{$ };
\node[right] at (1.2,1){$S_{1,3},$};
\node[right] at (2.1,1){$S_{2,3},$};
\node[right] at (3.0,1){$S_{3,3},$};
\node[right] at (3.9,1){$S_{4,3}$};
\node[right] at (4.6,1){$\}$};
\end{scope}
\end{tikzpicture}
\caption{A $3$-colored family $\FF = \FF_1 \sqcup \FF_2 \sqcup \FF_3$ with $|\FF_i| =4$ and the corresponding grid $G[4]^3$. The point $(2,3,4)$ in green, the line $(3,t,2)$ in red, and the plane $(s,2,t)$ in blue, and the corresponding colorful subfamilies $\{S_{2,1}, S_{3,2}, S_{4,3}\}$, $\{S_{3,1}, S_{2,3}\}$, and $\{S_{2,2}\}$, respectively}
\label{fig:g(v)}
\end{center}
\end{figure}

Next, for any chain $\alpha \in C_\bullet(\grid_\FF)$, 
we define $\psi(\alpha)$ to be $\psi(\supp(\alpha))$.  
Note that if $\alpha$ is a $k$-cell of $\grid_\FF$, then $|\psi(\alpha)|=m-k$. 

For any nonempty subfamily $\GG \subseteq \FF$ 
we write $\bigcap \GG$ to mean $\cap_{S \in \GG}S$.
The definition of $\psi(\cdot)$ implies the following three straightforward properties:

\begin{claim}\label{c:no trick}
Let $\ell \le t$ and let $\Gamma$ be a subgrid of size $\ell$ in $\grid_\FF$.
For every cell $\sigma$ in $G[\ell]^m$ we have:
\begin{enumerate}
\item[(i)] $\psi(\Gamma_\bullet(\sigma)) = \psi(\tau)$ for every cell $\tau$ in the support of $\Gamma_\bullet(\sigma)$, 
\item[(ii)] $\psi(\Gamma_\bullet(\partial\sigma)) =
\psi(\Gamma_\bullet(\sigma))$, and
\item[(iii)] \label{c:notrick3} if $\sigma'$ is another cell in $G[\ell]^m$ such that $\sigma, \sigma'$ do not lie in a common axis-parallel hyperplane, then the families $\psi(\Gamma_\bullet(\sigma))$ and $\psi(\Gamma_\bullet(\sigma'))$ are disjoint.
\end{enumerate}
\end{claim}

\section{Colorful intersection patterns via homological minors}
\label{s:cip homin}
The machinery of~\cite{hb17,rb20} for analyzing intersection patterns via homological minors was designed to analyze \emph{complete} intersection patterns. We now adapt this framework to \emph{complete multipartite} intersection patterns.

The main goal for this section is the following.
Consider an $m$-colored family $\FF$ of subcomplexes of $\univ$, where each color class $\FF_i$ has size $t$. 
We show that if $t$ is sufficiently large and the topological complexity of $\FF$ is bounded, then there exists a chain map from a subgrid of the encoding grid $\grid_\FF$ to $C_\bullet(\univ)$ which is \emph{constrained} by $\FF$. The precise definition of a constrained chain map and the formal statement of this result (Lemma \ref{l:colored constrained}) is given in Section \ref{s:ccd}.

A key step in proving the existence of constrained chain maps is to establish a purely 
Ramsey-theoretic statement, which can be viewed as a vanishing lemma for subgrids (Lemma \ref{l:subgrid}). This will be established in Section \ref{s:subgrid}, while the proof of Lemma \ref{l:colored constrained} is given in Section \ref{ss:exist_ccm}.

\subsection{Colorful constrained chain maps}\label{s:ccd}
Let $\FF$ be an $m$-colored family of subcomplexes of $\univ$ and suppose each color class $\FF_i$ has size $t$. Recall the notation $\grid_\FF = G[t]^m$ and the 
bijection $A \mapsto \psi(A)$ between the axis-parallel $k$-flats 
intersecting $V(\grid_\FF)$ and the colorful subfamilies of $\FF$ of size $m-k$, established in Section~\ref{ss:grid encoding}.

Fix positive integers $d$ and $n$, with  $d\leq m$ and $n\leq t$, 
and 
let $Y$ denote the $d$-skeleton of $G[n]^m$, that is, 
$Y = (G[n]^m)^{(d)}$. 

We say that a nontrivial chain map $f_{\bullet} : C_{\bullet}(Y) \to C_{\bullet}(\univ)$ 
is \emph{constrained by $\FF$} if there exists a subgrid $\Gamma : V(G[n]^m) \to V(\grid_\FF)$ such that
\begin{equation} \label{eq:constrained}
\supp f_{\bullet}(\sigma) \subset \bigcap \psi \big(\Gamma_{\bullet}(\sigma)\big) \quad
\text{for every } \sigma\in Y \subset G[n]^m.
\end{equation}
It is convenient to keep in mind the following diagram, where the vertical dashed arrow represents the support constraint \eqref{eq:constrained}:
\[
\begin{tikzcd}
    & C_\bullet(\grid_\FF) \arrow[d,dashed,"\bigcap \psi(\cdot)"] \\
    C_\bullet(Y) \arrow[r, "f_\bullet"'] \arrow[ru, "\Gamma_\bullet"] & C_\bullet(\univ)
\end{tikzcd}
\]

Our goal for the remainder of this section is to establish the following:

\begin{lemma}\label{l:colored constrained}
Let $b,d,m,n \in \mathbb{N}$ be given with $b\geq 1$ and $m\geq d$. 
Let $Y$ denote the $d$-skeleton of $G[n]^m$. 
There exists an integer $t = t(b,d,m,n)$ such that the following holds: 
Let $\FF = \FF_1 \sqcup \cdots \sqcup \FF_m$ be an $m$-colored family of 
subcomplexes of a simplicial complex $\univ$. Suppose $\FF$ satisfies:
\begin{enumerate}
\item[(i)] $|\FF_i| = t$ for every $1\leq i \leq m$,
\item[(ii)] $\bigcap \GG \neq \emptyset$ for every colorful subfamily 
$\GG \subset \FF$, and
\item[(iii)] $\tilde{\beta}_j(\bigcap \GG) < b$ for all $0\leq j < d$ 
and all nonempty colorful subfamilies $\GG \subset \FF$.
\end{enumerate}
Then there exists a nontrivial chain map $f_{\bullet} : 
C_{\bullet}(Y) \to C_{\bullet}(\univ)$ that is constrained by~$\FF$.
\end{lemma}

\subsection{The subgrid lemma}\label{s:subgrid}

Here we establish the main Ramsey-type statement needed to prove Lemma \ref{l:colored constrained}. 
Throughout this section, we view $C_k(G[n]^m)$ as a vector space over $\mathbb{Z}_2$. 
Given a linear map $h:C_k(G[n]^m) \to \pth{\ZZ_2}^b$, 
we say that $h$ \emph{vanishes} on a subgrid $\Gamma$ of size $\ell$ 
(where $\Gamma \colon V(G[\ell]^m) \to V(G[n]^m)$) if the composition $h\circ \Gamma_\bullet$
is the 0-map; that is, 
\[h(\Gamma_{\bullet}(c)) = 0 \; \text{ for every } c\in C_k(G[\ell]^m).\]

\begin{lemma}[Subgrid lemma]\label{l:subgrid}
  Let $b,k,m,\ell \in \NN$ be given, with $\ell \geq 2$. There exists a constant
  $N = N(b,k,m,\ell)$ such that for all $n\geq N$, the following holds: Every linear map $h \colon C_k(G[n]^m) \to \pth{\ZZ_2}^b$ vanishes on some subgrid $\Gamma$ in $G[n]^m$ of size $\ell$. 
\end{lemma}

\begin{figure}
 \centering
  \begin{tikzpicture}

   \begin{scope}[scale = .75]
    \draw[line width = 1pt, blue, opacity = .8] (1,0) -- (6,0);
    \draw[line width = 1pt, blue, opacity = .8] (1,3) -- (6,3);
    \draw[line width = 1pt, blue, opacity = .8] (1,4) -- (6,4);

    \draw[line width = 1pt, blue, opacity = .8] (1,0) -- (1,4);
    \draw[line width = 1pt, blue, opacity = .8] (3,0) -- (3,4);
    \draw[line width = 1pt, blue, opacity = .8] (6,0) -- (6,4);
   \end{scope}

   \begin{scope}[scale = .75]
    \draw[red!50!black] (0,1) -- (7,1);
    \draw[red!50!black] (0,2) -- (7,2);
    \draw[red!50!black] (0,5) -- (7,5);
    \draw[red!50!black] (0,6) -- (7,6);
    \draw[red!50!black] (0,7) -- (7,7);
    \draw[red!50!black] (0,0) -- (1,0) (6,0) -- (7,0);
    \draw[red!50!black] (0,3) -- (1,3) (6,3) -- (7,3);
    \draw[red!50!black] (0,4) -- (1,4) (6,4) -- (7,4);

    \draw[red!50!black] (0,0) -- (0,7);
    \draw[red!50!black] (2,0) -- (2,7);
    \draw[red!50!black] (4,0) -- (4,7);
    \draw[red!50!black] (5,0) -- (5,7);
    \draw[red!50!black] (7,0) -- (7,7);
    \draw[red!50!black] (1,4) -- (1,7);
    \draw[red!50!black] (3,4) -- (3,7);
    \draw[red!50!black] (6,4) -- (6,7);

    \foreach \x in {0,...,7} 
    \foreach \y in {0,...,7}{ 
    \filldraw[black] (\x,\y) circle (1pt); 
    \draw (\x,\y) circle (1pt);}
    \end{scope}
    
    \begin{scope}[scale = .75]
    \filldraw[white] (1,0) circle (2pt);
    \filldraw[white] (3,0) circle (2pt);
    \filldraw[white] (6,0) circle (2pt);
    \draw[thick, blue] (1,0) circle (3pt);
    \draw[thick, blue] (3,0) circle (3pt);
    \draw[thick, blue] (6,0) circle (3pt);

    \filldraw[white] (1,3) circle (2pt);
    \filldraw[white] (3,3) circle (2pt);
    \filldraw[white] (6,3) circle (2pt);
    \draw[thick, blue] (1,3) circle (3pt);
    \draw[thick, blue] (3,3) circle (3pt);
    \draw[thick, blue] (6,3) circle (3pt);

    \filldraw[white] (1,4) circle (2pt);
    \filldraw[white] (3,4) circle (2pt);
    \filldraw[white] (6,4) circle (2pt);
    \draw[thick, blue] (1,4) circle (3pt);
    \draw[thick, blue] (3,4) circle (3pt);
    \draw[thick, blue] (6,4) circle (3pt);
    \end{scope}
    
    \begin{scope}[scale = .75]
    \foreach \x in {0,...,6}
    \foreach \y in {0,...,7}{ 
    \filldraw[white, opacity = .9] (.5+ \x,\y) circle (2.9pt);
    \filldraw[white, opacity = .9] (\y, .5+\x) circle (2.9pt);}

    \node at (.5+0,0) {\tiny $1$}; 
    \node at (.5+1,0) {\tiny $1$};
    \node at (.5+2,0) {\tiny $1$}; 
    \node at (.5+3,0) {\tiny $1$};
    \node at (.5+4,0) {\tiny $0$}; 
    \node at (.5+5,0) {\tiny $1$};
    \node at (.5+6,0) {\tiny $0$}; 

    \node at (.5+0,1) {\tiny $1$}; 
    \node at (.5+1,1) {\tiny $1$};
    \node at (.5+2,1) {\tiny $0$}; 
    \node at (.5+3,1) {\tiny $1$};
    \node at (.5+4,1) {\tiny $1$}; 
    \node at (.5+5,1) {\tiny $0$};
    \node at (.5+6,1) {\tiny $1$}; 

    \node at (.5+0,2) {\tiny $0$}; 
    \node at (.5+1,2) {\tiny $0$};
    \node at (.5+2,2) {\tiny $1$}; 
    \node at (.5+3,2) {\tiny $1$};
    \node at (.5+4,2) {\tiny $0$}; 
    \node at (.5+5,2) {\tiny $1$};
    \node at (.5+6,2) {\tiny $0$}; 

    \node at (.5+0,3) {\tiny $1$}; 
    \node at (.5+1,3) {\tiny $0$};
    \node at (.5+2,3) {\tiny $0$}; 
    \node at (.5+3,3) {\tiny $1$};
    \node at (.5+4,3) {\tiny $0$}; 
    \node at (.5+5,3) {\tiny $1$};
    \node at (.5+6,3) {\tiny $0$}; 

    \node at (.5+0,4) {\tiny $0$}; 
    \node at (.5+1,4) {\tiny $1$};
    \node at (.5+2,4) {\tiny $1$}; 
    \node at (.5+3,4) {\tiny $1$};
    \node at (.5+4,4) {\tiny $1$}; 
    \node at (.5+5,4) {\tiny $0$};
    \node at (.5+6,4) {\tiny $1$}; 

    \node at (.5+0,5) {\tiny $0$}; 
    \node at (.5+1,5) {\tiny $1$};
    \node at (.5+2,5) {\tiny $0$}; 
    \node at (.5+3,5) {\tiny $0$};
    \node at (.5+4,5) {\tiny $0$}; 
    \node at (.5+5,5) {\tiny $0$};
    \node at (.5+6,5) {\tiny $1$}; 

    \node at (.5+0,6) {\tiny $1$}; 
    \node at (.5+1,6) {\tiny $1$};
    \node at (.5+2,6) {\tiny $0$}; 
    \node at (.5+3,6) {\tiny $0$};
    \node at (.5+4,6) {\tiny $0$}; 
    \node at (.5+5,6) {\tiny $0$};
    \node at (.5+6,6) {\tiny $0$}; 

    \node at (.5+0,7) {\tiny $0$}; 
    \node at (.5+1,7) {\tiny $1$};
    \node at (.5+2,7) {\tiny $1$}; 
    \node at (.5+3,7) {\tiny $1$};
    \node at (.5+4,7) {\tiny $1$}; 
    \node at (.5+5,7) {\tiny $0$};
    \node at (.5+6,7) {\tiny $1$}; 
    \node at (0,.5+0) {\tiny $1$}; 
    \node at (0,.5+1) {\tiny $1$}; 
    \node at (0,.5+2) {\tiny $0$}; 
    \node at (0,.5+3) {\tiny $1$}; 
    \node at (0,.5+4) {\tiny $0$}; 
    \node at (0,.5+5) {\tiny $1$}; 
    \node at (0,.5+6) {\tiny $1$}; 

    \node at (1,.5+0) {\tiny $0$}; 
    \node at (1,.5+1) {\tiny $1$}; 
    \node at (1,.5+2) {\tiny $1$}; 
    \node at (1,.5+3) {\tiny $0$}; 
    \node at (1,.5+4) {\tiny $0$}; 
    \node at (1,.5+5) {\tiny $0$}; 
    \node at (1,.5+6) {\tiny $0$}; 

    \node at (2,.5+0) {\tiny $0$}; 
    \node at (2,.5+1) {\tiny $0$}; 
    \node at (2,.5+2) {\tiny $0$}; 
    \node at (2,.5+3) {\tiny $0$}; 
    \node at (2,.5+4) {\tiny $0$}; 
    \node at (2,.5+5) {\tiny $1$}; 
    \node at (2,.5+6) {\tiny $1$}; 

    \node at (3,.5+0) {\tiny $1$}; 
    \node at (3,.5+1) {\tiny $1$}; 
    \node at (3,.5+2) {\tiny $0$}; 
    \node at (3,.5+3) {\tiny $0$}; 
    \node at (3,.5+4) {\tiny $0$}; 
    \node at (3,.5+5) {\tiny $0$}; 
    \node at (3,.5+6) {\tiny $1$}; 

    \node at (4,.5+0) {\tiny $0$}; 
    \node at (4,.5+1) {\tiny $0$}; 
    \node at (4,.5+2) {\tiny $1$}; 
    \node at (4,.5+3) {\tiny $1$}; 
    \node at (4,.5+4) {\tiny $1$}; 
    \node at (4,.5+5) {\tiny $0$}; 
    \node at (4,.5+6) {\tiny $1$}; 

    \node at (5,.5+0) {\tiny $0$}; 
    \node at (5,.5+1) {\tiny $1$}; 
    \node at (5,.5+2) {\tiny $1$}; 
    \node at (5,.5+3) {\tiny $0$}; 
    \node at (5,.5+4) {\tiny $0$}; 
    \node at (5,.5+5) {\tiny $1$}; 
    \node at (5,.5+6) {\tiny $0$}; 

    \node at (6,.5+0) {\tiny $0$}; 
    \node at (6,.5+1) {\tiny $0$}; 
    \node at (6,.5+2) {\tiny $0$}; 
    \node at (6,.5+3) {\tiny $0$}; 
    \node at (6,.5+4) {\tiny $1$}; 
    \node at (6,.5+5) {\tiny $1$}; 
    \node at (6,.5+6) {\tiny $1$}; 

    \node at (7,.5+0) {\tiny $0$}; 
    \node at (7,.5+1) {\tiny $1$}; 
    \node at (7,.5+2) {\tiny $1$}; 
    \node at (7,.5+3) {\tiny $0$}; 
    \node at (7,.5+4) {\tiny $0$}; 
    \node at (7,.5+5) {\tiny $1$}; 
    \node at (7,.5+6) {\tiny $0$}; 
    \end{scope}

\begin{scope}[yshift=1.4cm, xshift = 8cm]
\tikzmath{\x1 = 7/8; \x2 =-1/4; \y1 = 3/4; \y2 =3/8; \z1= 0; \z2=15/16; } 

\begin{scope}[scale = .8]

\draw[-latex] (0*\y1,0*\y2) -- (4*\x1, 4*\x2);
\draw[line width = 0.2, opacity = 0.2] (1*\y1,1*\y2) --++ (3*\x1, 3*\x2);
\draw[line width = 0.2, opacity = 0.2] (2*\y1,2*\y2) --++ (3*\x1, 3*\x2);
\draw[line width = 0.2, opacity = 0.2] (3*\y1,3*\y2) --++ (3*\x1, 3*\x2);

\draw[-latex] (0,0) -- (4*\y1, 4*\y2);
\draw[line width = 0.2, opacity = 0.2] (1*\x1,1*\x2) --++ (3*\y1, 3*\y2);
\draw[line width = 0.2, opacity = 0.2] (2*\x1,2*\x2) --++ (3*\y1, 3*\y2);
\draw[line width = 0.2, opacity = 0.2] (3*\x1,3*\x2) --++ (3*\y1, 3*\y2);

\draw[-latex] (0,0) -- (4*\z1, 4*\z2);
\draw[line width = 0.2, opacity = 0.2] (1*\y1,1*\y2) --++ (3*\z1, 3*\z2);
\draw[line width = 0.2, opacity = 0.2] (2*\y1,2*\y2) --++ (3*\z1, 3*\z2);
\draw[line width = 0.2, opacity = 0.2] (3*\y1,3*\y2) --++ (3*\z1, 3*\z2);

\draw[line width = 0.2, opacity = 0.2] (1*\z1,1*\z2) --++ (3*\y1, 3*\y2);
\draw[line width = 0.2, opacity = 0.2] (2*\z1,2*\z2) --++ (3*\y1, 3*\y2);
\draw[line width = 0.2, opacity = 0.2] (3*\z1,3*\z2) --++ (3*\y1, 3*\y2);

\draw[thin, dotted] (0*\x1+1*\y1+2*\z1, 0*\x2+1*\y2+2*\z2)
--++ (1*\x1+0*\y1+0*\z1, 1*\x2+0*\y2+0*\z2);
\draw[thin, dotted] (0*\x1+1*\y1+3*\z1, 0*\x2+1*\y2+3*\z2)
--++ (1*\x1+0*\y1+0*\z1, 1*\x2+0*\y2+0*\z2);

\draw[thin, dotted] (0*\x1+2*\y1+2*\z1, 0*\x2+2*\y2+2*\z2)
--++ (1*\x1+0*\y1+0*\z1, 1*\x2+0*\y2+0*\z2);
\draw[thin, dotted] (0*\x1+2*\y1+3*\z1, 0*\x2+2*\y2+3*\z2)
--++ (1*\x1+0*\y1+0*\z1, 1*\x2+0*\y2+0*\z2);

\draw[thin, dotted] (0*\x1+3*\y1+2*\z1, 0*\x2+3*\y2+2*\z2)
--++ (1*\x1+0*\y1+0*\z1, 1*\x2+0*\y2+0*\z2);
\draw[thin, dotted] (0*\x1+3*\y1+3*\z1, 0*\x2+3*\y2+3*\z2)
--++ (1*\x1+0*\y1+0*\z1, 1*\x2+0*\y2+0*\z2);

\draw[thin, dotted] (1*\x1+1*\y1+0*\z1, 1*\x2+1*\y2+0*\z2)
--++ (0*\x1+0*\y1+2*\z1, 0*\x2+0*\y2+2*\z2);
\draw[thin, dotted] (2*\x1+1*\y1+0*\z1, 2*\x2+1*\y2+0*\z2)
--++ (0*\x1+0*\y1+2*\z1, 0*\x2+0*\y2+2*\z2);
\draw[thin, dotted] (3*\x1+1*\y1+0*\z1, 3*\x2+1*\y2+0*\z2)
--++ (0*\x1+0*\y1+2*\z1, 0*\x2+0*\y2+2*\z2);

\draw[thin, dotted] (1*\x1+2*\y1+0*\z1, 1*\x2+2*\y2+0*\z2)
--++ (0*\x1+0*\y1+2*\z1, 0*\x2+0*\y2+2*\z2);
\draw[thin, dotted] (2*\x1+2*\y1+0*\z1, 2*\x2+2*\y2+0*\z2)
--++ (0*\x1+0*\y1+2*\z1, 0*\x2+0*\y2+2*\z2);
\draw[thin, dotted] (3*\x1+2*\y1+0*\z1, 3*\x2+2*\y2+0*\z2)
--++ (0*\x1+0*\y1+2*\z1, 0*\x2+0*\y2+2*\z2);

\draw[thin, dotted] (1*\x1+3*\y1+0*\z1, 1*\x2+3*\y2+0*\z2)
--++ (0*\x1+0*\y1+2*\z1, 0*\x2+0*\y2+2*\z2);
\draw[thin, dotted] (2*\x1+3*\y1+0*\z1, 2*\x2+3*\y2+0*\z2)
--++ (0*\x1+0*\y1+2*\z1, 0*\x2+0*\y2+2*\z2);
\draw[thin, dotted] (3*\x1+3*\y1+0*\z1, 3*\x2+3*\y2+0*\z2)
--++ (0*\x1+0*\y1+2*\z1, 0*\x2+0*\y2+2*\z2);

\draw[thin] (1*\x1+1*\y1+2*\z1, 1*\x2+1*\y2+2*\z2)
--++ (2*\x1+0*\y1+0*\z1, 2*\x2+0*\y2+0*\z2);
\draw[thin] (1*\x1+2*\y1+2*\z1, 1*\x2+2*\y2+2*\z2)
--++ (2*\x1+0*\y1+0*\z1, 2*\x2+0*\y2+0*\z2);
\draw[thin] (1*\x1+3*\y1+2*\z1, 1*\x2+3*\y2+2*\z2)
--++ (2*\x1+0*\y1+0*\z1, 2*\x2+0*\y2+0*\z2);

\draw[thin] (1*\x1+1*\y1+3*\z1, 1*\x2+1*\y2+3*\z2)
--++ (2*\x1+0*\y1+0*\z1, 2*\x2+0*\y2+0*\z2);
\draw[thin] (1*\x1+2*\y1+3*\z1, 1*\x2+2*\y2+3*\z2)
--++ (2*\x1+0*\y1+0*\z1, 2*\x2+0*\y2+0*\z2);
\draw[thin] (1*\x1+3*\y1+3*\z1, 1*\x2+3*\y2+3*\z2)
--++ (2*\x1+0*\y1+0*\z1, 2*\x2+0*\y2+0*\z2);

\draw[thin] (1*\x1+1*\y1+2*\z1, 1*\x2+1*\y2+2*\z2)
--++ (0*\x1+0*\y1+1*\z1, 0*\x2+0*\y2+1*\z2);
\draw[thin] (2*\x1+1*\y1+2*\z1, 2*\x2+1*\y2+2*\z2)
--++ (0*\x1+0*\y1+1*\z1, 0*\x2+0*\y2+1*\z2);
\draw[thin] (3*\x1+1*\y1+2*\z1, 3*\x2+1*\y2+2*\z2)
--++ (0*\x1+0*\y1+1*\z1, 0*\x2+0*\y2+1*\z2);

\draw[thin] (1*\x1+2*\y1+2*\z1, 1*\x2+2*\y2+2*\z2)
--++ (0*\x1+0*\y1+1*\z1, 0*\x2+0*\y2+1*\z2);
\draw[thin] (2*\x1+2*\y1+2*\z1, 2*\x2+2*\y2+2*\z2)
--++ (0*\x1+0*\y1+1*\z1, 0*\x2+0*\y2+1*\z2);
\draw[thin] (3*\x1+2*\y1+2*\z1, 3*\x2+2*\y2+2*\z2)
--++ (0*\x1+0*\y1+1*\z1, 0*\x2+0*\y2+1*\z2);

\draw[thin] (1*\x1+3*\y1+2*\z1, 1*\x2+3*\y2+2*\z2)
--++ (0*\x1+0*\y1+1*\z1, 0*\x2+0*\y2+1*\z2);
\draw[thin] (2*\x1+3*\y1+2*\z1, 2*\x2+3*\y2+2*\z2)
--++ (0*\x1+0*\y1+1*\z1, 0*\x2+0*\y2+1*\z2);
\draw[thin] (3*\x1+3*\y1+2*\z1, 3*\x2+3*\y2+2*\z2)
--++ (0*\x1+0*\y1+1*\z1, 0*\x2+0*\y2+1*\z2);

\draw[thin] (1*\x1+1*\y1+2*\z1, 1*\x2+1*\y2+2*\z2)
--++ (0*\x1+2*\y1+0*\z1, 0*\x2+2*\y2+0*\z2);
\draw[thin] (2*\x1+1*\y1+2*\z1, 2*\x2+1*\y2+2*\z2)
--++ (0*\x1+2*\y1+0*\z1, 0*\x2+2*\y2+0*\z2);
\draw[thin] (3*\x1+1*\y1+2*\z1, 3*\x2+1*\y2+2*\z2)
--++ (0*\x1+2*\y1+0*\z1, 0*\x2+2*\y2+0*\z2);

\draw[thin] (1*\x1+1*\y1+3*\z1, 1*\x2+1*\y2+3*\z2)
--++ (0*\x1+2*\y1+0*\z1, 0*\x2+2*\y2+0*\z2);
\draw[thin] (2*\x1+1*\y1+3*\z1, 2*\x2+1*\y2+3*\z2)
--++ (0*\x1+2*\y1+0*\z1, 0*\x2+2*\y2+0*\z2);
\draw[thin] (3*\x1+1*\y1+3*\z1, 3*\x2+1*\y2+3*\z2)
--++ (0*\x1+2*\y1+0*\z1, 0*\x2+2*\y2+0*\z2);

\fill[opacity = .2, blue] (1*\x1+1*\y1+2*\z1, 1*\x2+1*\y2+2*\z2)
--++ (2*\x1+0*\y1+0*\z1, 2*\x2+0*\y2+0*\z2)
--++ (0*\x1+0*\y1+1*\z1, 0*\x2+0*\y2+1*\z2) 
--++ (-2*\x1+0*\y1+0*\z1, -2*\x2+0*\y2+0*\z2);

\fill[opacity = .4, blue] (1*\x1+1*\y1+3*\z1, 1*\x2+1*\y2+3*\z2)
--++ (2*\x1+0*\y1+0*\z1, 2*\x2+0*\y2+0*\z2)
--++ (0*\x1+2*\y1+0*\z1, 0*\x2+2*\y2+0*\z2) 
--++ (-2*\x1+0*\y1+0*\z1, -2*\x2+0*\y2+0*\z2);

\fill[opacity = .3, blue] (3*\x1+1*\y1+2*\z1, 3*\x2+1*\y2+2*\z2)
--++ (0*\x1+2*\y1+0*\z1, 0*\x2+2*\y2+0*\z2)
--++ (0*\x1+0*\y1+1*\z1, 0*\x2+0*\y2+1*\z2) 
--++ (0*\x1-2*\y1+0*\z1, 0*\x2-2*\y2+0*\z2);
\end{scope}
\end{scope}

  \end{tikzpicture}
   \caption{On the left: A linear map $h: C_1(G[8]^2) \to \mathbb{Z}_2$ which vanishes on the blue subgrid of size $3$.  On the right: The 3-chain $\bx_3\big( (1,1,2), (3,3,3) \big)$}
   \label{fig:subgrid_lemma}
\end{figure}

The Subgrid lemma is illustrated on Figure \ref{fig:subgrid_lemma} (left); it will be derived from the  following consequence of the Gallai-Witt theorem~\cite[p. 40]{graham1990ramsey}:

\begin{proposition}\label{l:weak GW}
  Let $m, \ell, q \in \mathbb{N}$ be given, with $\ell \geq 2$. There exists an integer $N = N(m,\ell,q)$ such that the following holds:
  For every $q$-coloring of $V(G[N]^m)$, there is a monochromatic subgrid $\Gamma$ in $G[N]^m$ of size $\ell$.
\end{proposition}

\begin{remark*}
We do not use the full strength of the Gallai-Witt theorem, 
as we do not require the subgrid to be a ``homothet'' of $G[\ell]^m$. 
Furthermore, we make no attempt to optimize the constants in our Ramsey-type result. 
\end{remark*}

\begin{proof}[Proof of the Subgrid lemma]
The proof first handles the case $k=m$ and then uses it to deduce the case $k < m$. 
Note that for $k > m$ the lemma is trivial, as the chain group $C_k(G[n]^m)$ is zero. 

\subsubsection*{The case $k=m$}\hfill\par
Fix $m\geq 1$, $b\geq 1$, and $\ell \ge 2$. 
Let $n \geq \ell$ be an integer and let $\mathbf{1}$ 
denote $(1,1,\ldots, 1) \in \mathbb{Z}^m$. 
Given $x, y \in V(G[n]^m)$ we write $x \preceq y$ 
if $x_i \le y_i$ for $i=1,\ldots, m$. 
For any two vertices $x \preceq y$ we define the $m$-chain 
$\bx_m(x,y) \in C_m(G[n]^m)$ by
\[
 \bx_m(x,y) =
 \begin{cases}
   0 & \text{ if } x_i=y_i \text{ for some } i,\\
  [x_1,y_1] \times [x_2,y_2] \times \cdots \times [x_m,y_m] & \text{ otherwise.}
 \end{cases}
\]
For illustration see Figure \ref{fig:subgrid_lemma} (right). 
Note that $\bx_m(x,y)$ is nonzero if and only if $x$ and $y$ do not share any coordinate (i.e., they do not lie in a common axis-parallel hyperplane).

For the linear map $h:C_m(G[n]^m) \to \pth{\ZZ_2}^b$ 
to vanish on a subgrid $\Gamma$ of size $\ell$, it suffices to have
\begin{equation}\label{eq:kernel}
 h(\bx_m(\Gamma(z), \Gamma(z + \mathbf{1}))) = 0 
 \text{ for all } z \in V(G[\ell -1]^m).
\end{equation}
This is because $C_m(G[\ell]^m)$ is spanned by the set of elementary cells 
$\{\bx_m(z, z +  \mathbf{1}) \colon z \in V(G[\ell-1]^m)\}$, 
and $\bx_m(\Gamma(z), \Gamma(z + \mathbf{1})) = 
\Gamma_\bullet(\bx_m(z, z + \mathbf{1}))$.

To apply Proposition \ref{l:weak GW}, we consider the vertex coloring $\chi_h \colon V(G[n]^m) \to \pth{\ZZ_2}^b$ defined by 
\[\chi_h(y) \coloneqq h(\bx_m(\mathbf{1},y)).\]

\begin{claim}\label{cl:kernel}
 Let $\Gamma$ be a subgrid of size $\ell$ in $G[n]^m$. If $\chi_h\big(\Gamma(v)\big)$ attains the same value for all $v \in V(G[\ell]^m)$, then $h$ vanishes on $\Gamma$.
\end{claim}

\begin{proof}[Proof of claim]
By \eqref{eq:kernel}, it suffices to show that $h(\bx_m(u,v)) = 0$ 
for all $u = \Gamma(z)$ and $v=\Gamma(z+\mathbf{1})$ with $z \in V(G[\ell-1]^m)$. 
Let $W$ denote the set of corners of $\bx_m(u,v)$, that is, 
\[ W= \{(w_1,\ldots, w_m)\colon w_i \in \{u_i, v_i\} \text{ for every } i \in [m]\}. \]
By the inclusion-exclusion principle
(and since we work with coefficients in $\mathbb Z_2$), we have 
\[ \bx_m(u,v) = \sum_{w \in W} \bx_m(\mathbf{1}, w).\]
In other words, 
$\bx_m(u,v)$ is the sum of the boxes anchored at 
$\mathbf{1}$ with opposite corners in $W$. 
Applying the linear map $h$ to both sides, 
and using the hypothesis that $\chi_h$ is constant on the image of $\Gamma$, 
we get
\[ h(\bx_m(u,v)) = \sum_{w \in W} h(\bx_m(\mathbf{1}, w)) = 
\sum_{w \in W} \chi_h(w) = \chi_h(u)|W| = 0,\]
where the last equality holds because $|W| = 2^m$ is even (for $m\geq 1$). 
\end{proof}

To complete the proof of Lemma~\ref{l:subgrid} for the case $k=m$, 
let $N = N(b,m,m,\ell)$ be the constant $N(m,\ell, 2^{b})$ 
from Proposition~\ref{l:weak GW}. 
For any $n \ge N$ and any linear map $h:C_m(G[n]^m) \to (\ZZ_2)^b$, 
applying  Proposition~\ref{l:weak GW} to the vertex coloring $\chi_h$ 
ensures the existence of a monochromatic subgrid $\Gamma$ of size $\ell$ in $G[n]^m$. 
By Claim~\ref{cl:kernel}, $h$ vanishes on $\Gamma$.

\subsubsection*{The case $1 \leq k < m$}\hfill\par
The proof of the general case essentially reduces to multiple instances of the 
full-dimensional case by restricting attention to every possible 
axis-parallel $k$-flat that passes through vertices of $G[n]^m$. 

It will be convenient to introduce some further notation.
Fix integers $1\leq k<m$, $b\geq 1$, and $\ell \ge 2$.
For every vertex $y = (y_1, \dots, y_m) \in G[n]^m$ and an index set $I \in \binom{[m]}k$, 
we define the $k$-chain $\chain_I(y)$ as follows. Let $x = (x_1, \dots, x_m)$ be the vertex defined by 
\[x_j = \begin{cases} 1 & \text{if } j \in I, \\ y_j & 
\text{if } j \notin I. \end{cases} \] 
We then set 
\[ \chain_I(y) = \begin{cases} 0 & \text{if } y_j = 1 \text{ for some } j \in I,\\ [x_1,y_1] \times \cdots \times [x_m,y_m] & 
\text{otherwise.} \end{cases} \] 
For illustration see Figure \ref{fig:chn}.
Note that if $j \notin I$, then $[x_j, y_j] = [y_j, y_j]$ is simply the vertex $y_j$ (a 0-chain). Thus, $\chain_I(y)$ is a product of $k$ 1-chains (the intervals $[1, y_j]$ for $j \in I$) and $(m-k)$ 0-chains (the vertices $y_j$ for $j \notin I$), making it a $k$-chain. It is nonzero if and only if $y_j \neq 1$ for all $j \in I$. (Note that if we allowed $k=m$, then $\chain_{[m]}(y)$ would coincide with $\bx_m(\mathbf{1},y)$.)

\begin{figure}
\begin{center}
\begin{tikzpicture}
\tikzmath{\x1 = 7/8; \x2 =-1/4; \y1 = 3/4; \y2 =3/8; \z1= 0; \z2=15/16; } 

\begin{scope}[scale = .8]

\draw[-latex] (0*\y1,0*\y2) -- (4*\x1, 4*\x2);

\draw[line width = .2, opacity = .2]
(1*\y1,1*\y2) --++ (3*\x1, 3*\x2)
(2*\y1,2*\y2) --++ (3*\x1, 3*\x2)
(3*\y1,3*\y2) --++ (3*\x1, 3*\x2)
(4*\y1,4*\y2) --++ (3*\x1, 3*\x2);

\draw[-latex] (0,0) -- (5*\y1, 5*\y2);
\draw[line width = .2, opacity = .2]
(1*\x1,1*\x2) --++ (4*\y1, 4*\y2)
(2*\x1,2*\x2) --++ (4*\y1, 4*\y2)
(3*\x1,3*\x2) --++ (4*\y1, 4*\y2);

\draw[-latex] (0,0) -- (5*\z1, 5*\z2);
\draw[line width = .2, opacity = .2]
(1*\y1,1*\y2) --++ (4*\z1, 4*\z2)
(2*\y1,2*\y2) --++ (4*\z1, 4*\z2)
(3*\y1,3*\y2) --++ (4*\z1, 4*\z2)
(4*\y1,4*\y2) --++ (4*\z1, 4*\z2);

\draw[line width = .2, opacity = .2]
(1*\z1,1*\z2) --++ (4*\y1, 4*\y2)
(2*\z1,2*\z2) --++ (4*\y1, 4*\y2)
(3*\z1,3*\z2) --++ (4*\y1, 4*\y2)
(4*\z1,4*\z2) --++ (4*\y1, 4*\y2);

\draw[thin, dotted] (1*\y1+4*\z1,1*\y2+4*\z2) --++ (1*\x1, 1*\x2);
\draw[thin, dotted] (2*\y1+4*\z1,2*\y2+4*\z2) --++ (1*\x1, 1*\x2);
\draw[thin, dotted] (3*\y1+4*\z1,3*\y2+4*\z2) --++ (1*\x1, 1*\x2);
\draw[thin, dotted] (4*\y1+4*\z1,4*\y2+4*\z2) --++ (1*\x1, 1*\x2);

\draw[thin] (1*\x1+1*\y1+4*\z1,1*\x2+1*\y2+4*\z2) --++ (2*\x1, 2*\x2);
\draw[thin] (1*\x1+2*\y1+4*\z1,1*\x2+2*\y2+4*\z2) --++ (2*\x1, 2*\x2);
\draw[thin] (1*\x1+3*\y1+4*\z1,1*\x2+3*\y2+4*\z2) --++ (2*\x1, 2*\x2);
\draw[thin] (1*\x1+4*\y1+4*\z1,1*\x2+4*\y2+4*\z2) --++ (2*\x1, 2*\x2);

\draw[thin] (1*\x1+1*\y1+4*\z1,1*\x2+1*\y2+4*\z2) --++ (3*\y1, 3*\y2);
\draw[thin] (2*\x1+1*\y1+4*\z1,2*\x2+1*\y2+4*\z2) --++ (3*\y1, 3*\y2);
\draw[thin] (3*\x1+1*\y1+4*\z1,3*\x2+1*\y2+4*\z2) --++ (3*\y1, 3*\y2);

\draw[thin, dotted] (1*\x1+4*\y1+0*\z1,1*\x2+4*\y2+0*\z2) --++ (1*\z1, 1*\z2);
\draw[thin, dotted] (2*\x1+4*\y1+0*\z1,2*\x2+4*\y2+0*\z2) --++ (1*\z1, 1*\z2);
\draw[thin, dotted] (3*\x1+4*\y1+0*\z1,3*\x2+4*\y2+0*\z2) --++ (1*\z1, 1*\z2);

\draw[ultra thin] (1*\x1+4*\y1+1*\z1,1*\x2+4*\y2+1*\z2) --++ (3*\z1, 3*\z2);
\draw[thin] (1*\x1+4*\y1+1*\z1,1*\x2+4*\y2+1*\z2) --++ (1.53*\z1, 1.53*\z2);

\draw[thin] (2*\x1+4*\y1+1*\z1,2*\x2+4*\y2+1*\z2) --++ (2.26*\z1, 2.26*\z2);
\draw[ultra thin] (2*\x1+4*\y1+1*\z1,2*\x2+4*\y2+1*\z2) --++ (3*\z1, 3*\z2);
\draw[thin] (3*\x1+4*\y1+1*\z1,3*\x2+4*\y2+1*\z2) --++ (3*\z1, 3*\z2);

\draw[thin, dotted] (0*\x1+4*\y1+1*\z1,0*\x2+4*\y2+1*\z2) --++ (1*\x1, 1*\x2);
\draw[thin, dotted] (0*\x1+4*\y1+2*\z1,0*\x2+4*\y2+2*\z2) --++ (1*\x1, 1*\x2);
\draw[thin, dotted] (0*\x1+4*\y1+3*\z1,0*\x2+4*\y2+3*\z2) --++ (1*\x1, 1*\x2);
\draw[thin, dotted] (0*\x1+4*\y1+4*\z1,0*\x2+4*\y2+4*\z2) --++ (1*\x1, 1*\x2);

\draw[thin] (1*\x1+4*\y1+1*\z1,1*\x2+4*\y2+1*\z2) --++ (2*\x1, 2*\x2);
\draw[thin] (1*\x1+4*\y1+2*\z1,1*\x2+4*\y2+2*\z2) --++ (2*\x1, 2*\x2);
\draw[ultra thin] (1*\x1+4*\y1+3*\z1,1*\x2+4*\y2+3*\z2) --++ (2*\x1, 2*\x2);
\draw[thin] (3*\x1+4*\y1+3*\z1,3*\x2+4*\y2+3*\z2) --++ (-1.35*\x1, -1.35*\x2);
\draw[ultra thin] (1*\x1+4*\y1+4*\z1,1*\x2+4*\y2+4*\z2) --++ (2*\x1, 2*\x2);

\fill[opacity = .2, blue] (1*\x1+1*\y1+4*\z1, 1*\x2+1*\y2+4*\z2) 
--++ (2*\x1+0*\y1+0*\z1, 2*\x2+0*\y2+0*\z2) 
--++ (0*\x1+3*\y1+0*\z1, 0*\x2+3*\y2+0*\z2)
--++ (-2*\x1+0*\y1+0*\z1, -2*\x2+0*\y2+0*\z2) ; 

\fill[opacity = .2, red] (1*\x1+4*\y1+1*\z1, 1*\x2+4*\y2+1*\z2) 
--++ (2*\x1+0*\y1+0*\z1, 2*\x2+0*\y2+0*\z2) 
--++ (0*\x1+0*\y1+3*\z1, 0*\x2+0*\y2+3*\z2)
--++ (-2*\x1+0*\y1+0*\z1, -2*\x2+0*\y2+0*\z2) ; 
\end{scope}

\end{tikzpicture}
\caption{The 2-chain $\chain_{\{1,2\}}\big( (3,4,4) \big)$  in blue, and the 2-chain $\chain_{\{1,3\}}\big( (3,4,4) \big)$ in red}
\label{fig:chn}
\end{center}
\end{figure}

To apply Proposition \ref{l:weak GW}, 
we consider the vertex coloring $\rho_h \colon V(G[n]^m) \to (\ZZ_2)^{b\binom{m}{k}}$ defined by 
\[\rho_h(y) \coloneqq \big(h(\chain_I(y))\big)_{I \in \binom{[m]}k}. \]
We think of $\rho_h(y)$ as a vector of $\binom{m}{k}$ elements from $\pth{\ZZ_2}^b$, indexed by the subsets $I\in \binom{[m]}{k}$. 
We denote the component of $\rho_h(\cdot)$ corresponding to the index set $I$
by $[\rho_h(\cdot)]_I$.

\begin{claim}\label{cl:kernel2}
 Let $\Gamma$ be a subgrid of size $\ell$ in $G[n]^m$. If $\rho_h\big(\Gamma(v)\big)$ attains the same value for all $v \in V(G[\ell]^m)$, then $h$ vanishes on $\Gamma$.
\end{claim}

\begin{proof}[Proof of Claim]
For a point $a = (a_1, \dots, a_m) \in \mathbb{Z}^m$ and subset 
$I\in \binom{[m]}{k}$ define the axis-parallel $k$-flat $F_I(a)$ by
\[F_I(a) = \{(x_1, \dots, x_m) \in \mathbb{R}^m : x_i = a_i\; \forall i\notin I\}.\]
This is the $k$-flat passing through $a$ 
where the coordinates indexed by $I$ vary, while the others remain fixed. 
Note that for any vertex $z \in V(G[\ell]^m)$, 
the intersection $G[\ell]^m \cap F_I(z)$ is isomorphic to $G[\ell]^k$,
and similarly $G[n]^m\cap F_I(\Gamma(z))$ is isomorphic to $G[n]^k$. 

Furthermore, for any vertex  $w\in V(G[\ell]^m) \cap F_I(z)$, 
the support of the chain $\chain_I(\Gamma(w))$ is contained in 
$V(G[n])^m \cap F_I(\Gamma(z))$. 
Thus restricting $\Gamma$ and $h$ to the flats $F_I(z)$ and $F_I(\Gamma(z))$ 
reduces the problem to the full-dimensional case (within that $k$-flat). 
Since $[\rho_h(\Gamma(w))]_I = h\big(\chain_I(\Gamma(w))\big)$ is constant 
for all $w\in V(G[\ell]^m) \cap F_I(z)$, 
Claim \ref{cl:kernel} implies that $h$ vanishes on 
the restriction of $\Gamma$ to $F_I(z)$.  
Since the elementary $k$-cells in these flats span all of $C_k(G[\ell]^m)$ 
as we vary $z$ and $I$, we conclude that $h$ vanishes on all of $\Gamma$.
\end{proof}

Returning to the proof of Lemma \ref{l:subgrid}, let $N = N(m,\ell, q)$ 
be the constant from Proposition~\ref{l:weak GW} with $q = 2^{b\binom{m}{k}}$.
For any $n \ge N$ and for any linear map
$h : C_k(G[n]^m) \to \pth{\ZZ_2}^b$, 
applying Proposition~\ref{l:weak GW} to the vertex coloring $\rho_h$ 
ensures the existence monochromatic subgrid $\Gamma$ in $G[n]^m$ of size $\ell$. 
By Claim~\ref{cl:kernel2}, $h$ vanishes on $\Gamma$.  
This completes the proof of the Subgrid lemma.
\end{proof}

\subsection{The existence of colorful constrained chain maps} \label{ss:exist_ccm}

Here we prove Lemma \ref{l:colored constrained}. The chain map $f_\bullet$ will be constructed recursively on the $\ell$-skeleton for $\ell = 0,1,\dots, d$.

We start by defining constants $t_0 > t_1 > \cdots > t_{d}$ 
by setting $t_{d} \coloneqq n$ and, for $\ell = d-1, \dots, 0$, recursively setting
$t_{\ell} \coloneqq N(b, \ell+1, m, t_{\ell+1})$ 
where $N(\cdot,\cdot,\cdot,\cdot)$ is the function from the 
Subgrid Lemma (Lemma \ref{l:subgrid}). 
Note that the definition of the $t_{\ell}$'s ensures:
\begin{claim}\label{c:kernel induction}
For any $0 \le \ell \leq d-1$ and any linear map 
$h: C_{\ell+1}(G[t_\ell]^m) \to (\ZZ_2)^b$,
$h$ vanishes on some subgrid $\Gamma$ in $G[t_\ell]^m$ of size $t_{\ell+1}$. 
\end{claim}

We will show that Lemma \ref{l:colored constrained} holds with $t = t_0$. 
We therefore assume that each color class $\FF_i$ has size $t_0$, 
and we set $\grid_\FF = G[t_0]^m$.
For every $\ell = 0,\dots, d$ define $Y_\ell$ to be the $\ell$-skeleton of 
$G[t_\ell]^m$, that is, $Y_\ell = (G[t_\ell]^m)^{(\ell)}$. 
We will recursively use Claim \ref{c:kernel induction} to propagate the following 
property from $\ell=0$ to $\ell=d$:

\begin{quote}
\emph{$(P_\ell):$ there exists a subgrid 
$\Gamma^{(\ell)} : V(G[t_\ell]^m) \to V(\grid_\FF)$  
and a nontrivial chain map 
\[f^{(\ell)}_\bullet: C_\bullet(Y_\ell) \to C_\bullet(\univ)\] 
such that $\supp f^{(\ell)}_\bullet(\sigma)\subset 
\bigcap \psi\big(\Gamma^{(\ell)}_{\bullet}(\sigma)\big)$ 
for every cell $\sigma\in Y_\ell$.}
\end{quote}

Property $(P_\ell)$ essentially asserts the existence of a 
chain map $f^{(\ell)}_\bullet$ that is constrained by 
$\FF$ on the $\ell$-skeleton of the grid complex $G[t_\ell]^m$.

Observe that $Y_d$ is the $d$-skeleton of $G[n]^m$, which is precisely the 
$Y$ from the statement of Lemma~\ref{l:colored constrained}. 
Consequently, the subgrid $\Gamma = \Gamma^{(d)}$ and the chain map 
$f_\bullet = f^{(d)}_\bullet$ given by property $(P_d)$ 
will prove Lemma~\ref{l:colored constrained} with parameters $t(b,d,m,n) = t_0$.

\subsubsection*{Setting up the induction}\hfill\par
It remains to prove that property $(P_\ell)$ holds under the hypotheses 
of Lemma~\ref{l:colored constrained} for $\ell=0,1,\ldots, d$. 
The proof is by induction. 

First, for each colorful subfamily $\GG \subset \FF$,
we fix (arbitrarily) a basis 
for $\tilde{H}_j(\bigcap \GG)$, $0\leq j < d$. 
(Here, we view the reduced homology groups as $\mathbb{Z}_2$-vector spaces.)
Since $\tilde{\beta}_j(\bigcap \GG) < b$ by assumption, we may view each 
$\tilde{H}_j(\bigcap \GG)$ as a subspace of~$(\ZZ_2)^b$.
These bases (and subspaces) remain fixed for the remainder of the proof.

To verify the base case $(P_0)$, 
we let $\Gamma^{(0)} : V(G[t_0]^m)\to V(\grid_\FF)$ be the identity map. 
(Recall that $\grid_\FF = G[t_0]^m$.) 
For each vertex $v\in G[t_0]^m$ we fix a vertex $p_v$ in the intersection 
$\bigcap \psi(v)$ of the maximal colorful family associated with $v$ 
(which is nonempty by hypothesis). 
We define the chain map 
$f^{(0)}_\bullet : C_\bullet(Y_0) \to C_\bullet(\univ)$ 
by setting $f^{(0)}_{\bullet}(v) = p_v$ 
for every vertex $v$ of $G[t_0]^m$. 
This map is clearly nontrivial and constrained by $\FF$.

\subsubsection*{The induction step}\hfill\par
Let $0\leq \ell < d$ and suppose property $(P_\ell)$ is satisfied. 
That is, we are given a subgrid
$\Gamma^{(\ell)} : V(G[t_\ell]^m) \to V(\grid_\FF)$ 
and a nontrivial chain map 
\[f^{(\ell)}_\bullet: 
C_\bullet(Y_\ell) \to C_\bullet(\univ)\] 
such that $\supp f^{(\ell)}_\bullet(\sigma) \subset 
\bigcap \psi\big( \Gamma^{(\ell)}_{\bullet}(\sigma) \big)$ 
for every cell $\sigma\in Y_\ell$.

For any $(\ell+1)$-cell  $\sigma$ in~$G[t_\ell]^m$, 
the chain $\Gamma^{(\ell)}_{\bullet}(\sigma)$ is well-defined.
By Claim \ref{c:no trick}, the support of this chain has 
the same affine span in $\grid_\FF$ 
as the support of its boundary $\Gamma^{(\ell)}_{\bullet}(\partial \sigma)$.
Therefore we have
\begin{equation}
\supp f^{(\ell)}_{\bullet}(\partial\sigma) \subset 
\bigcap \psi\big(\Gamma^{(\ell)}_{\bullet}(\partial \sigma) \big) = 
\bigcap \psi\big( \Gamma^{(\ell)}_{\bullet}(\sigma) \big).
\end{equation}
Now we define a linear map $h: C_{\ell+1}(G[t_\ell]^m) \to (\ZZ_2)^b$ 
by setting, for any $(\ell+1)$-dimensional cell~$\sigma$,
\begin{equation}\label{eq:def h}
h(\sigma) \coloneqq \left[ f^{(\ell)}_\bullet(\partial\sigma) \right] \in 
\tilde{H}_{\ell}\pth{\bigcap \psi\big( \Gamma^{(\ell)}_\bullet(\sigma) \big) }
\subset (\ZZ_2)^b,
\end{equation}
and extending $h$ linearly. 
In other words, $h(\sigma)$ equals the homology class of the image 
$f^{(\ell)}_\bullet(\partial \sigma)$ in the $\ell$-dimensional 
(reduced) homology group of 
$\bigcap \psi\big( \Gamma^{(\ell)}_\bullet(\sigma) \big)$, 
which we can view as an element in~$(\ZZ_2)^b$. 

By Claim~\ref{c:kernel  induction}, $h$ vanishes on some subgrid 
$\Phi : V(G[t_{\ell+1}]^m) \to V(G[t_\ell]^m)$ of size $t_{\ell+1}$.
We set 
\[\Gamma^{(\ell+1)} \coloneqq \Gamma^{(\ell)}\circ \Phi,\] 
and note that $\Gamma^{(\ell+1)}$ is indeed a vertex map from 
$V(G[t_{\ell+1}]^m)$ to $V(\grid_\FF)$, 
making $\Gamma^{(\ell+1)}$ a subgrid of $\grid_\FF$ of size $t_{\ell+1}$. 
Note also that the induced chain map $\Gamma^{(\ell+1)}_\bullet$  
satisfies $\Gamma^{(\ell+1)}_\bullet = \Gamma_\bullet^{(\ell)} \circ \Phi_\bullet$.

Let $\tau\in G[t_{\ell+1}]^m$ be a cell of dimension $(\ell+1)$.
Since $h$ vanishes on $\Phi$, we have $h\big(\Phi_\bullet(\tau)\big) = 0$. 
By definition of $h$, this means the 
homology class of $f_\bullet^{(\ell)}\big(\partial \Phi_\bullet(\tau)\big)$ is zero, that is, 
\begin{equation}
\left[ f^{(\ell)}_\bullet\big(\partial\Phi_\bullet(\tau)\big) \right]  = 0 \in
\tilde{H}_{\ell}\pth{\bigcap \psi \big(\Gamma^{(\ell)}_\bullet(\Phi_\bullet(\tau))\big)}.
\end{equation}
Using that $\Phi_\bullet$ is a chain map and the fact that 
$\Gamma^{(\ell+1)}_\bullet = \Gamma_\bullet^{(\ell)} \circ \Phi_\bullet$, this
rewrites as
\begin{equation}\label{eq:boundary}
  \left[ f^{(\ell)}_\bullet \big( \Phi_\bullet(\partial\tau) \big) \right]=0\in 
  \tilde{H}_{\ell}\pth{\bigcap \psi \big( \Gamma^{(\ell+1)}_\bullet(\tau)\big)}.
\end{equation}
For a cell $\sigma \in G[t_{\ell+1}]^m$ of dimension at most $\ell$, we set
\begin{equation}\label{eq:equation13}
f_\bullet^{(\ell+1)}(\sigma) \coloneqq 
f_{\bullet}^{(\ell)}\big( \Phi_\bullet(\sigma) \big).
\end{equation}
For a cell $\tau\in G[t_{\ell+1}]^m$ of dimension $(\ell +1)$, 
equation~\eqref{eq:boundary} reveals that $f_\bullet^{(\ell+1)}(\partial\tau)$ is a 
boundary in $C_{\ell}\pth{ \bigcap \psi \big( \Gamma^{(\ell+1)}_\bullet(\tau) \big)  }$. 
We choose some arbitrary $\alpha \in C_{\ell+1} \pth{
\bigcap \psi\big( \Gamma^{(\ell+1)}_\bullet(\tau) \big) }$ 
such that $\partial \alpha = f_\bullet^{(\ell+1)}(\partial\tau)$, and set
\begin{equation}
f^{(\ell+1)}_\bullet(\tau) \coloneqq \alpha.
\end{equation}

Therefore $f^{(\ell+1)}_\bullet$ is indeed a chain map 
from $C_\bullet(Y_{\ell+1})$ to $C_\bullet\pth{\univ}$. 
Note that by \eqref{eq:equation13}, it follows that $f_\bullet^{(\ell+1)}$ 
is nontrivial because $f_\bullet^{(\ell)}$ is nontrivial and 
$\Phi$, being a subgrid, maps each vertex to a single vertex.
Moreover, by construction $f^{(\ell+1)}_\bullet$ and $\Gamma^{(\ell+1)}$ satisfy $(P_{\ell+1})$. This concludes the proof of the induction step and of Lemma~\ref{l:colored constrained}. \qed

\section{Homological minors in grid complexes}\label{s:hominor in grid}
Let $K$ be a simplicial complex on $n$ vertices.
The main goal in this section is to give a canonical construction of a \emph{generic} chain map 
\[g_\bullet \colon C_\bullet(K) \to C_\bullet(G[n]^m),\]
for any $m>\mu(K)$.
Here, generic means that disjoint faces of $K$ map to chains that avoid common alignments along axis-parallel hyperplanes. The precise statement (Lemma \ref{l:K_generic_grid_minor}) is given in Section \ref{subsec:generic chain maps}. 

The construction of $g_\bullet$ is based on an adaptation of stair convexity to 
grid complexes. In Sections \ref{ss:stair convex chains}--\ref{s:boundary}, 
we formally define stair convex chains and verify their compatibility 
with the boundary operator. 

\subsection{Generic chain maps into grid complexes}\label{subsec:generic chain maps}

Consider chains $\alpha\in C_k(G[n]^m)$ and $\beta \in C_\ell(G[n]^m)$. 
We say that the pair $\{\alpha, \beta\}$ is \emph{degenerate} 
if there exist a $k$-cell $\alpha_k \in \supp(\alpha)$ and an 
$\ell$-cell $\beta_\ell \in \supp(\beta)$ such that $\alpha_k$ and $\beta_\ell$ 
are contained in a common axis-parallel hyperplane 
$\Pi_j(a) = \{(x_1, \dots, x_m)\in \mathbb{R}^m : x_j = a\}$. 
If the supports of $\alpha$ and $\beta$ contain no such 
pair of cells, then we say that the pair 
$\{\alpha, \beta\}$  is \emph{generic}.

We extend the notion of genericity to chain maps as follows. 
Let $K$ be a simplicial complex and consider a nontrivial chain map
\[g_\bullet \colon C_\bullet(K) \to C_\bullet(G[n]^m).\]
We say that $g_\bullet$ is \emph{generic} if every pair of 
disjoint simplices of $K$ maps to a generic pair of chains of $C_\bullet(G[n]^m)$. 
In other words, $g_\bullet$ is generic if
\[ \sigma, \tau\in K, \sigma \cap \tau = 
\emptyset \implies \{g_\bullet(\sigma), g_\bullet(\tau)\} \text{ is generic.}  \]
Roughly speaking, the following \emph{``Picasso Lemma"}  asserts that any simplicial complex can be realized within a (sufficiently large) grid complex via a generic chain map. (See Figure \ref{fig:picasso}.)

\begin{lemma} \label{l:K_generic_grid_minor}
Let $K$ be a simplicial complex on $n$ vertices. For any $m > \mu(K)$ there exists a generic, nontrivial chain map $g_\bullet \colon C_\bullet(K) \to C_\bullet(G[n]^m)$.
\end{lemma}

Our proof of this lemma is based on stair convexity, which we introduce in 
Section~\ref{ss:stair convex chains}. Sections \ref{subsec:nonrecursive} and \ref{s:boundary} will establish auxiliary results and properties related to stair convexity, concluding with the proof of Lemma~\ref{l:K_generic_grid_minor}. First, we establish the following immediate consequence:

\begin{figure}
\begin{center}
\begin{tikzpicture}
\tikzmath{\x1 = 7/8; \x2 =-1/4; \y1 = 3/4; \y2 =3/8; \z1= 0; \z2=15/16; } 

\begin{scope}[scale = .8]
\draw[-latex] (0,0) -- (6*\x1, 6*\x2);
\draw[-latex] (0,0) -- (6*\y1, 6*\y2);
\draw[-latex] (0,0) -- (6*\z1, 6*\z2);

\draw[line width = 0.2, opacity = 0.2] 
(0*\x1+1*\y1+0*\z1,0*\x2+1*\y2+0*\z2) --++ (5*\x1+0*\y1+0*\z1,5*\x2+0*\y2+0*\z2)
(0*\x1+2*\y1+0*\z1,0*\x2+2*\y2+0*\z2) --++ (5*\x1+0*\y1+0*\z1,5*\x2+0*\y2+0*\z2)
(0*\x1+3*\y1+0*\z1,0*\x2+3*\y2+0*\z2) --++ (5*\x1+0*\y1+0*\z1,5*\x2+0*\y2+0*\z2)
(0*\x1+4*\y1+0*\z1,0*\x2+4*\y2+0*\z2) --++ (5*\x1+0*\y1+0*\z1,5*\x2+0*\y2+0*\z2)
(0*\x1+5*\y1+0*\z1,0*\x2+5*\y2+0*\z2) --++ (5*\x1+0*\y1+0*\z1,5*\x2+0*\y2+0*\z2);

\draw[line width = 0.2, opacity = 0.2] 
(1*\x1+0*\y1+0*\z1,1*\x2+0*\y2+0*\z2) --++ (0*\x1+5*\y1+0*\z1,0*\x2+5*\y2+0*\z2)
(2*\x1+0*\y1+0*\z1,2*\x2+0*\y2+0*\z2) --++ (0*\x1+5*\y1+0*\z1,0*\x2+5*\y2+0*\z2)
(3*\x1+0*\y1+0*\z1,3*\x2+0*\y2+0*\z2) --++ (0*\x1+5*\y1+0*\z1,0*\x2+5*\y2+0*\z2)
(4*\x1+0*\y1+0*\z1,4*\x2+0*\y2+0*\z2) --++ (0*\x1+5*\y1+0*\z1,0*\x2+5*\y2+0*\z2)
(5*\x1+0*\y1+0*\z1,5*\x2+0*\y2+0*\z2) --++ (0*\x1+5*\y1+0*\z1,0*\x2+5*\y2+0*\z2);

\draw[thin, dotted] (1*\x1+1*\y1,1*\x2+1*\y2) --++ (\z1,\z2);
\draw[thin, dotted] (2*\x1+2*\y1,2*\x2+2*\y2) --++ (2*\z1,2*\z2);
\draw[thin, dotted] (3*\x1+3*\y1,3*\x2+3*\y2) --++ (3*\z1,3*\z2);
\draw[thin, dotted] (4*\x1+4*\y1,4*\x2+4*\y2) --++ (4*\z1,4*\z2);
\draw[thin, dotted] (5*\x1+5*\y1,5*\x2+5*\y2) --++ (5*\z1,5*\z2);

\draw[thick, blue] (\x1+\y1+\z1, \x2+\y2+\z2) --++ (1*\z1, 1*\z2) --++ (1*\y1, 1*\y2) --++ (1*\x1, 1*\x2);
\draw[thick, blue] (\x1+\y1+\z1, \x2+\y2+\z2) --++ (2*\z1, 2*\z2) --++ (2*\y1, 2*\y2) --++ (2*\x1, 2*\x2);
\draw[thick, blue] (\x1+\y1+\z1, \x2+\y2+\z2) --++ (3*\z1, 3*\z2) --++ (3*\y1, 3*\y2) --++ (3*\x1, 3*\x2);
\draw[thick, blue] (\x1+\y1+\z1, \x2+\y2+\z2) --++ (4*\z1, 4*\z2) --++ (4*\y1, 4*\y2) --++ (4*\x1, 4*\x2);

\fill[white, opacity=.92] (2*\x1+2*\y1+3.63*\z1, 2*\x2+2*\y2+3.63*\z2) circle (1.3pt); 

\fill[white, opacity=.92] (1*\x1+3.17*\y1+4*\z1, 1*\x2+3.17*\y2+4*\z2) circle (1.3pt); 

\draw[thick, blue] (2*\x1+2*\y1+2*\z1, 2*\x2+2*\y2+2*\z2) --++ (1*\z1, 1*\z2) --++ (1*\y1, 1*\y2) --++ (1*\x1, 1*\x2);
\draw[thick, blue] (2*\x1+2*\y1+2*\z1, 2*\x2+2*\y2+2*\z2) --++ (2*\z1, 2*\z2) --++ (2*\y1, 2*\y2) --++ (2*\x1, 2*\x2);
\draw[thick, blue] (2*\x1+2*\y1+2*\z1, 2*\x2+2*\y2+2*\z2) --++ (3*\z1, 3*\z2) --++ (3*\y1, 3*\y2) --++ (3*\x1, 3*\x2);

\fill[white, opacity=.92] (3*\x1+3*\y1+4.63*\z1, 3*\x2+3*\y2+4.63*\z2) circle (1.3pt);

\draw[thick, blue] (3*\x1+3*\y1+3*\z1, 3*\x2+3*\y2+3*\z2) --++ (1*\z1, 1*\z2) --++ (1*\y1, 1*\y2) --++ (1*\x1, 1*\x2);
\draw[thick, blue] (3*\x1+3*\y1+3*\z1, 3*\x2+3*\y2+3*\z2) --++ (2*\z1, 2*\z2) --++ (2*\y1, 2*\y2) --++ (2*\x1, 2*\x2);

\draw[thick, blue] (4*\x1+4*\y1+4*\z1, 4*\x2+4*\y2+4*\z2) --++ (1*\z1, 1*\z2) --++ (1*\y1, 1*\y2) --++ (1*\x1, 1*\x2);


\filldraw[white] (1*\x1+1*\y1+1*\z1, 1*\x2+1*\y2+1*\z2) circle (1.4pt);
\draw[thick, blue] 	   (1*\x1+1*\y1+1*\z1, 1*\x2+1*\y2+1*\z2) circle (2pt); 

\filldraw[black] (1*\x1+1*\y1+2*\z1, 1*\x2+1*\y2+2*\z2) circle (1.4pt);
\draw 		   (1*\x1+1*\y1+2*\z1, 1*\x2+1*\y2+2*\z2) circle (1.4pt);

\filldraw[black] (1*\x1+1*\y1+3*\z1, 1*\x2+1*\y2+3*\z2) circle (1.4pt);
\draw 		   (1*\x1+1*\y1+3*\z1, 1*\x2+1*\y2+3*\z2) circle (1.4pt);

\filldraw[black] (1*\x1+1*\y1+4*\z1, 1*\x2+1*\y2+4*\z2) circle (1.4pt);
\draw 		   (1*\x1+1*\y1+4*\z1, 1*\x2+1*\y2+4*\z2) circle (1.4pt);

\filldraw[black] (1*\x1+1*\y1+5*\z1, 1*\x2+1*\y2+5*\z2) circle (1.4pt);
\draw 		   (1*\x1+1*\y1+5*\z1, 1*\x2+1*\y2+5*\z2) circle (1.4pt);


\filldraw[black] (1*\x1+2*\y1+2*\z1, 1*\x2+2*\y2+2*\z2) circle (1.3pt);
\draw 		   (1*\x1+2*\y1+2*\z1, 1*\x2+2*\y2+2*\z2) circle (1.3pt);

\filldraw[black] (1*\x1+2*\y1+3*\z1, 1*\x2+2*\y2+3*\z2) circle (1.3pt);
\draw 		   (1*\x1+2*\y1+3*\z1, 1*\x2+2*\y2+3*\z2) circle (1.3pt);

\filldraw[black] (1*\x1+2*\y1+4*\z1, 1*\x2+2*\y2+4*\z2) circle (1.3pt);
\draw 		   (1*\x1+2*\y1+4*\z1, 1*\x2+2*\y2+4*\z2) circle (1.3pt);

\filldraw[black] (1*\x1+2*\y1+5*\z1, 1*\x2+2*\y2+5*\z2) circle (1.3pt);
\draw 		   (1*\x1+2*\y1+5*\z1, 1*\x2+2*\y2+5*\z2) circle (1.3pt);


\filldraw[black] (1*\x1+3*\y1+3*\z1, 1*\x2+3*\y2+3*\z2) circle (1.2pt);
\draw 		   (1*\x1+3*\y1+3*\z1, 1*\x2+3*\y2+3*\z2) circle (1.2pt);

\filldraw[black] (1*\x1+3*\y1+4*\z1, 1*\x2+3*\y2+4*\z2) circle (1.2pt);
\draw 		   (1*\x1+3*\y1+4*\z1, 1*\x2+3*\y2+4*\z2) circle (1.2pt);

\filldraw[black] (1*\x1+3*\y1+5*\z1, 1*\x2+3*\y2+5*\z2) circle (1.2pt);
\draw 		   (1*\x1+3*\y1+5*\z1, 1*\x2+3*\y2+5*\z2) circle (1.2pt);


\filldraw[black] (1*\x1+4*\y1+4*\z1, 1*\x2+4*\y2+4*\z2) circle (1.1pt);
\draw 		   (1*\x1+4*\y1+4*\z1, 1*\x2+4*\y2+4*\z2) circle (1.1pt);

\filldraw[black] (1*\x1+4*\y1+5*\z1, 1*\x2+4*\y2+5*\z2) circle (1.1pt);
\draw 		   (1*\x1+4*\y1+5*\z1, 1*\x2+4*\y2+5*\z2) circle (1.1pt);


\filldraw[black] (1*\x1+5*\y1+5*\z1, 1*\x2+5*\y2+5*\z2) circle (1.0pt);
\draw 		   (1*\x1+5*\y1+5*\z1, 1*\x2+5*\y2+5*\z2) circle (1.0pt);


\filldraw[white] (2*\x1+2*\y1+2*\z1, 2*\x2+2*\y2+2*\z2) circle (1.3pt);
\draw[thick, blue]   (2*\x1+2*\y1+2*\z1, 2*\x2+2*\y2+2*\z2) circle (1.9pt); 

\filldraw[black] (2*\x1+2*\y1+3*\z1, 2*\x2+2*\y2+3*\z2) circle (1.3pt);
\draw 		   (2*\x1+2*\y1+3*\z1, 2*\x2+2*\y2+3*\z2) circle (1.3pt);

\filldraw[black] (2*\x1+2*\y1+4*\z1, 2*\x2+2*\y2+4*\z2) circle (1.3pt);
\draw 		   (2*\x1+2*\y1+4*\z1, 2*\x2+2*\y2+4*\z2) circle (1.3pt);

\filldraw[black] (2*\x1+2*\y1+5*\z1, 2*\x2+2*\y2+5*\z2) circle (1.3pt);
\draw 		   (2*\x1+2*\y1+5*\z1, 2*\x2+2*\y2+5*\z2) circle (1.3pt);


\filldraw[black] (2*\x1+3*\y1+3*\z1, 2*\x2+3*\y2+3*\z2) circle (1.2pt);
\draw 		   (2*\x1+3*\y1+3*\z1, 2*\x2+3*\y2+3*\z2) circle (1.2pt);
\filldraw[black] (2*\x1+3*\y1+4*\z1, 2*\x2+3*\y2+4*\z2) circle (1.2pt);
\draw 		   (2*\x1+3*\y1+4*\z1, 2*\x2+3*\y2+4*\z2) circle (1.2pt);

\filldraw[black] (2*\x1+3*\y1+5*\z1, 2*\x2+3*\y2+5*\z2) circle (1.2pt);
\draw 		   (2*\x1+3*\y1+5*\z1, 2*\x2+3*\y2+5*\z2) circle (1.2pt);


\filldraw[black] (2*\x1+4*\y1+4*\z1, 2*\x2+4*\y2+4*\z2) circle (1.1pt);
\draw 		   (2*\x1+4*\y1+4*\z1, 2*\x2+4*\y2+4*\z2) circle (1.1pt);

\filldraw[black] (2*\x1+4*\y1+5*\z1, 2*\x2+4*\y2+5*\z2) circle (1.1pt);
\draw 		   (2*\x1+4*\y1+5*\z1, 2*\x2+4*\y2+5*\z2) circle (1.1pt);


\filldraw[black] (2*\x1+5*\y1+5*\z1, 2*\x2+5*\y2+5*\z2) circle (1.0pt);
\draw 		   (2*\x1+5*\y1+5*\z1, 2*\x2+5*\y2+5*\z2) circle (1.0pt);


\filldraw[white] (3*\x1+3*\y1+3*\z1, 3*\x2+3*\y2+3*\z2) circle (1.2pt);
\draw[thick, blue] (3*\x1+3*\y1+3*\z1, 3*\x2+3*\y2+3*\z2) circle (1.8pt); 

\filldraw[black] (3*\x1+3*\y1+4*\z1, 3*\x2+3*\y2+4*\z2) circle (1.2pt);
\draw 		   (3*\x1+3*\y1+4*\z1, 3*\x2+3*\y2+4*\z2) circle (1.2pt);

\filldraw[black] (3*\x1+3*\y1+5*\z1, 3*\x2+3*\y2+5*\z2) circle (1.2pt);
\draw 		   (3*\x1+3*\y1+5*\z1, 3*\x2+3*\y2+5*\z2) circle (1.2pt);


\filldraw[black] (3*\x1+4*\y1+4*\z1, 3*\x2+4*\y2+4*\z2) circle (1.1pt);
\draw 		   (3*\x1+4*\y1+4*\z1, 3*\x2+4*\y2+4*\z2) circle (1.1pt);

\filldraw[black] (3*\x1+4*\y1+5*\z1, 3*\x2+4*\y2+5*\z2) circle (1.1pt);
\draw 		   (3*\x1+4*\y1+5*\z1, 3*\x2+4*\y2+5*\z2) circle (1.1pt);


\filldraw[black] (3*\x1+5*\y1+5*\z1, 3*\x2+5*\y2+5*\z2) circle (1.0pt);
\draw 		   (3*\x1+5*\y1+5*\z1, 3*\x2+5*\y2+5*\z2) circle (1.0pt);


\filldraw[white] (4*\x1+4*\y1+4*\z1, 4*\x2+4*\y2+4*\z2) circle (1.1pt);
\draw[thick, blue]   (4*\x1+4*\y1+4*\z1, 4*\x2+4*\y2+4*\z2) circle (1.7pt); 

\filldraw[black] (4*\x1+4*\y1+5*\z1, 4*\x2+4*\y2+5*\z2) circle (1.1pt);
\draw 		   (4*\x1+4*\y1+5*\z1, 4*\x2+4*\y2+5*\z2) circle (1.1pt);


\filldraw[black] (4*\x1+5*\y1+5*\z1, 4*\x2+5*\y2+5*\z2) circle (1.0pt);
\draw 		   (4*\x1+5*\y1+5*\z1, 4*\x2+5*\y2+5*\z2) circle (1.0pt);


\filldraw[white] (5*\x1+5*\y1+5*\z1, 5*\x2+5*\y2+5*\z2) circle (1.0pt);
\draw[thick, blue]  (5*\x1+5*\y1+5*\z1, 5*\x2+5*\y2+5*\z2) circle (1.6pt);
\end{scope}
\end{tikzpicture}
\caption{The graph $K_5$ (considered as a 1-dimensional simplicial complex) realized as a subcomplex of the grid complex $G[5]^3$ via the generic chain map given in Lemma \ref{l:K_generic_grid_minor}}
\label{fig:picasso}
\end{center}
\end{figure}

\begin{corollary} \label{t:cubical_and_excluded}
Fix a $d$-dimensional simplicial complex $K$ on $n$ vertices and a 
simplicial complex $\univ$ such that $K\not\prec_H \univ$. 
Let $Y$ denote the $d$-skeleton of $G[n]^m$, where $m > \mu(K)$. 
For every nontrivial chain map $f_\bullet \colon C_\bullet(Y) \to C_\bullet(\univ)$ 
there exist disjoint cells $\sigma$ and $\tau$ in $Y$ satisfying:
\begin{enumerate}
\item[(i)] \label{i:generic} $\sigma$ and $\tau$ are not contained in any 
common axis-parallel hyperplane.
\item[(ii)] \label{i:intersect} The supports of $f_\bullet(\sigma)$ and 
$f_\bullet(\tau)$ are overlapping. 
\end{enumerate}
\end{corollary}

\begin{proof}
Fix a nontrivial chain map $f_\bullet \colon C_\bullet(Y) \to C_\bullet(\univ)$. 
By Lemma \ref{l:K_generic_grid_minor}, 
there exists a generic nontrivial chain map 
\[g_\bullet \colon C_\bullet(K) \to C_\bullet(G[n]^m).\] 
Since $K$ is $d$-dimensional, 
the image of $g_\bullet$ is contained in the $d$-skeleton of 
$G[n]^m$ (which is $Y$). Thus, the composition $h_\bullet = f_\bullet \circ g_\bullet$ is a well-defined, nontrivial chain map:
\[\begin{tikzcd}
  & C_\bullet(Y) \arrow[dr, "f_\bullet"] \\
C_\bullet(K) \arrow[ur, "g_\bullet"] \arrow[rr, "h_\bullet"] & & C_\bullet(\univ) 
\end{tikzcd}\]
By the hypothesis $K\not\prec_H \univ$, there exist disjoint simplices $\sigma'$ and $\tau'$ in $K$ such that the supports of $h_\bullet(\sigma')$ and $h_\bullet(\tau')$ overlap. 
Because $h_\bullet$ factors through $C_\bullet(Y)$, 
there must exist cells $\sigma \in \supp(g_\bullet(\sigma'))$ and $\tau \in \supp(g_\bullet(\tau'))$ such that the supports of $f_\bullet(\sigma)$ and $f_\bullet(\tau)$ overlap. 

Since $g_\bullet$ is generic and $\sigma'$ and $\tau'$ are disjoint, the pair of chains 
$\{g_\bullet(\sigma'), g_\bullet(\tau')\}$ is generic. By definition, this means that the cells
$\sigma$ and $\tau$ are not contained in a common axis-parallel hyperplane.

Finally, since $g_\bullet$ preserves dimension, we have:
\[\dim \sigma + \dim \tau = \dim \sigma' + \dim \tau' \leq \mu(K) < m.\]
By Observation~\ref{o:hyperplane}, if $\sigma$ and $\tau$ were to intersect, 
they would have to be contained in a common axis-parallel hyperplane. 
Since we have just established they are not, $\sigma$ and $\tau$ must be disjoint.
\end{proof}

\subsection{Stair convex chains}
\label{ss:stair convex chains}
The stair convex hull was introduced by Bukh et al.~\cite{stairConv} as a tool for analyzing point configurations related to extremal problems in discrete geometry such as lower bounds on the size of weak $\varepsilon$-nets. 
Here we define a particular class of chains in the grid complex $G[n]^m$ that resembles their
recursive definition.

We fix some integer $n \ge 2$ and work implicitly in the grid
{complex}~$G[n]^m$. For any integers  $m \ge k \ge 0$ (with $m\geq 1$)
and any sequence $1\leq
a_1 < \cdots < a_{k+1} \leq n$ we define the \emph{stair convex}
$k$-chain $\stc^{m}_k(a_1,\ldots, a_{k+1}) \in C_k(G[n]^m)$.
The definition is recursive:

\renewcommand{\arraystretch}{1.4}
\begin{center}
\begin{tabular}{lrcl}
  $\mathbf{(m>k=0)}$ & $\stc_0^m(a) $ &  $\coloneqq$ & $\overbrace{(a,\dots, a)}^{m\text{-fold}}$\\
  $\mathbf{(m=k>0)}$ & $\stc_k^m(a_1, \dots, a_{m+1})$ & $\coloneqq$ & $[a_1,a_2] \times [a_2,a_3] \times \ldots \times [a_m,a_{m+1}]$\\
  $\mathbf{(m>k>0)}$ & $\stc_k^m(a_1, \ldots, a_{k+1})$ & $\coloneqq$ & $\stc^{m-1}_{k-1}(a_1, \ldots, a_k) \times [a_k,a_{k+1}]$\\  &&& \quad $+\  \stc^{m-1}_k(a_1, \ldots, a_{k+1}) \times \{a_{k+1}\}$\\
\end{tabular}
\end{center}

\subsubsection*{Examples}\hfill\par
At one end, $\stc_0^m(a)$ is a vertex on the main diagonal of
$G[n]^m$. At the other end, $\stc_m^m(a_1, \dots, a_{m+1})$
is an $m$-dimensional box. 
Let us examine some simple examples for intermediate values of $k$. 
For $m=2$ and $k=1$ we have
\[ \stc_1^{2}(a,b)  = \  \stc^1_0(a) \times [a,b]  + \stc^1_1(a,b)  \times \{b\} = \{a\} \times [a,b] + [a,b] \times \{b\}\]
which is a rectilinear path from $(a,a)$ to $(b,b)$ with a bend at
$(a,b)$. Similarly, $\stc_1^3(a, b)$ is a rectilinear path from
$(a,a,a)$ to $(b,b,b)$ with bends at $(a,a,b)$ and $(a,b,b)$:

\noindent
\begin{minipage}{7.6cm}
  \[\begin{aligned}
  & \stc_1^3(a, b)\\
  & = \stc^{2}_0(a) \times [a,b] + \stc^{2}_1(a,b) \times \{b\}\\
  & = \{a\} \times \{a\} \times [a,b]  \\
  & \quad \quad + \pth{\{a\} \times [a,b] + [a,b] \times \{b\}} \times \{b\}\\
  & = \{a\}\times \{a\} \times [a,b]  \\
  & \quad \quad + \{a\} \times [a,b] \times \{b\} \\
  & \quad \quad \quad \quad+ [a,b] \times \{b\} \times \{b\}
  \end{aligned}\]
\end{minipage}
\hfill
\begin{minipage}{5.6cm}
\begin{tikzpicture}
\tikzmath{\x1 = 7/8; \x2 =-1/4; \y1 = 3/4; \y2 =3/8; \z1= 0; \z2=15/16; } 

\begin{scope}[scale = .8]
\draw[-latex] (0,0) -- (4*\x1, 4*\x2);
\draw[-latex] (0,0) -- (4*\y1, 4*\y2);
\draw[-latex] (0,0) -- (4*\z1, 4*\z2);

\draw[line width = .2, opacity =.2]
(1*\x1+0*\y1,1*\x2+0*\y2) --++ (0*\x1+3*\y1,0*\x2+3*\y2)
(2*\x1+0*\y1,2*\x2+0*\y2) --++ (0*\x1+3*\y1,0*\x2+3*\y2)
(3*\x1+0*\y1,3*\x2+0*\y2) --++ (0*\x1+3*\y1,0*\x2+3*\y2);

\draw[line width = .2, opacity =.2]
(0*\x1+1*\y1,0*\x2+1*\y2) --++ (3*\x1+0*\y1,3*\x2+0*\y2)
(0*\x1+2*\y1,0*\x2+2*\y2) --++ (3*\x1+0*\y1,3*\x2+0*\y2)
(0*\x1+3*\y1,0*\x2+3*\y2) --++ (3*\x1+0*\y1,3*\x2+0*\y2);

\draw[thin, dotted] (\x1+\y1, \x2+\y2) --++ (\z1, \z2);

\draw[thin, dotted] (3*\x1+3*\y1, 3*\x2+ 3*\y2) 
--++ (3*\z1,3*\z2);
\draw[thick, blue] (\x1+\y1+\z1, \x2+\y2+\z2) --++ (2*\z1, 2*\z2) --++ (2*\y1, 2*\y2) --++ (2*\x1, 2*\x2);

\filldraw[white] (\x1+\y1+\z1, \x2+\y2+\z2) circle (1.4pt);
\draw (\x1+\y1+\z1, \x2+\y2+\z2) circle (1.4pt);

\filldraw[white] (\x1+\y1+3*\z1, \x2+\y2+3*\z2) circle (1.4pt);
\draw (\x1+\y1+3*\z1, \x2+\y2+3*\z2) circle (1.4pt);

\filldraw[white] (\x1+3*\y1+3*\z1, \x2+3*\y2+3*\z2) circle (1.4pt);
\draw (\x1+3*\y1+3*\z1, \x2+3*\y2+3*\z2) circle (1.4pt);

\filldraw[white] (3*\x1+ 3*\y1+ 3*\z1, 3*\x2+ 3*\y2+ 3*\z2) circle (1.4pt);
\draw (3*\x1+ 3*\y1+ 3*\z1, 3*\x2+ 3*\y2+ 3*\z2) circle (1.4pt);

\node[left] at (1*\x1+1*\y1+1*\z1, 1*\x2+1*\y2+1*\z2) {\footnotesize $(a,a,a)$};
\node[left] at (1*\x1+1*\y1+3*\z1, 1*\x2+1*\y2+3*\z2) {\footnotesize $(a,a,b)$};
\node[above] at (1*\x1+3*\y1+3*\z1, 1*\x2+3*\y2+3*\z2) {\footnotesize $(a,b,b)$};
\node[right] at (3*\x1+3*\y1+3*\z1, 3*\x2+3*\y2+3*\z2) {\footnotesize $(b,b,b)$};
\node[below] at (4*\x1,4*\x2) {\footnotesize $x_1$};
\node[above] at (4*\y1,4*\y2) {\footnotesize $x_2$};
\node[left] at (4*\z1,4*\z2) {\footnotesize $x_3$};
\end{scope}
\end{tikzpicture}
\end{minipage}

\bigskip

\noindent
For $m=3$ and $k=2$, we get the following $2$-chain:

\bigskip
\noindent
\begin{minipage}{5.6cm}
\begin{tikzpicture}
\tikzmath{\x1 = 7/8; \x2 =-1/4; \y1 = 3/4; \y2 =3/8; \z1= 0; \z2=15/16; } 
\begin{scope}[scale=.8]

\draw[-latex] (0,0) -- (4*\x1, 4*\x2);
\draw[-latex] (0,0) -- (4*\y1, 4*\y2);
\draw[-latex] (0,0) -- (4*\z1, 4*\z2);

\draw[line width = .2, opacity =.2]
(1*\x1+0*\y1,1*\x2+0*\y2) --++ (0*\x1+3*\y1,0*\x2+3*\y2)
(2*\x1+0*\y1,2*\x2+0*\y2) --++ (0*\x1+3*\y1,0*\x2+3*\y2)
(3*\x1+0*\y1,3*\x2+0*\y2) --++ (0*\x1+3*\y1,0*\x2+3*\y2);

\draw[line width = .2, opacity =.2]
(0*\x1+1*\y1,0*\x2+1*\y2) --++ (3*\x1+0*\y1,3*\x2+0*\y2)
(0*\x1+2*\y1,0*\x2+2*\y2) --++ (3*\x1+0*\y1,3*\x2+0*\y2)
(0*\x1+3*\y1,0*\x2+3*\y2) --++ (3*\x1+0*\y1,3*\x2+0*\y2);

\draw[thin, dotted] (\x1+\y1, \x2+\y2) --++ (\z1,\z2);

\draw[thin, dotted] (2*\x1+2*\y1, 2*\x2+2*\y2) --++ (2*\z1,2*\z2);

\draw[thin, dotted]  (3*\x1+3*\y1, 3*\x2+ 3*\y2) --++ (3*\z1,3*\z2);

\draw (\x1+\y1+2*\z1, \x2+\y2+2*\z2) --++ (1*\y1, 1*\y2) --++ (1*\x1, 1*\x2) --++ (\z1,\z2) --++ (\y1,\y2) (\x1+2*\y1+2*\z1, \x2+2*\y2+2*\z2) --++ (\z1,\z2) --++ (\x1,\x2);

\draw (\x1+\y1+2*\z1, \x2+\y2+2*\z2) --++ (\z1, \z2) --++ (2*\y1, 2*\y2) --++ (\x1, \x2);

\draw[thin, dotted] (\x1+\y1+\z1, \x2+\y2+\z2) --++ (\z1, \z2);

\draw[thin, dotted] (2*\x1+3*\y1+3*\z1, 2*\x2+3*\y2+3*\z2) --++ (\x1, \x2);

\fill[blue, opacity=.15] (1*\x1+1*\y1+2*\z1,1*\x2+1*\y2+2*\z2) --++ (\y1,\y2) --++ (\z1,\z2) --++(-\y1,-\y2);

\fill[blue, opacity=.23] (1*\x1+2*\y1+2*\z1,1*\x2+2*\y2+2*\z2) --++ (\x1,\x2) --++ (\z1,\z2) --++(-\x1,-\x2);

\fill[blue, opacity=.31] (2*\x1+2*\y1+3*\z1,2*\x2+2*\y2+3*\z2) --++ (\y1,\y2) --++ (-\x1,-\x2) --++(-\y1,-\y2);

\filldraw[white] (\x1+\y1+\z1, \x2+\y2+\z2) circle (1.4pt);
\draw (\x1+\y1+\z1, \x2+\y2+\z2) circle (1.4pt);

\filldraw[white] (2*\x1+ 2*\y1+ 2*\z1, 2*\x2+ 2*\y2+ 2*\z2) circle (1.4pt);
\draw (2*\x1+ 2*\y1+ 2*\z1, 2*\x2+ 2*\y2+ 2*\z2) circle (1.4pt);

\filldraw[white] (3*\x1+ 3*\y1+ 3*\z1, 3*\x2+ 3*\y2+ 3*\z2) circle (1.4pt);
\draw (3*\x1+ 3*\y1+ 3*\z1, 3*\x2+ 3*\y2+ 3*\z2) circle (1.4pt);

\node[left] at (1*\x1+1*\y1+1*\z1, 1*\x2+1*\y2+1*\z2) {\footnotesize $(a,a,a)$};

\node[right] at (2*\x1+2*\y1+2*\z1, 2*\x2+2*\y2+2*\z2) {\footnotesize $(b,b,b)$};

\node[right] at (3*\x1+3*\y1+3*\z1, 3*\x2+3*\y2+3*\z2) {\footnotesize $(c,c,c)$};

\node[below] at (4*\x1,4*\x2) {\footnotesize $x_1$};

\node[above] at (4*\y1,4*\y2) {\footnotesize $x_2$};

\node[left] at (4*\z1,4*\z2) {\footnotesize $x_3$};

\end{scope}
\end{tikzpicture}
\end{minipage}
\hfill
\begin{minipage}{7cm}
  \[\begin{aligned}
  & \stc_2^3(a, b,c) \\
  & = \stc^{2}_1(a, b)\times [b,c] + \stc^{2}_2(a, b,c) \times \{c\}\\
  & = \{a\}\times [a,b] \times [b,c] \\
  & \quad \quad + [a,b] \times \{b\} \times [b,c]\\
  & \quad\quad\quad\quad + [a,b] \times [b,c] \times \{c\},
  \end{aligned}\]
\end{minipage}

\begin{remark*}
We will see that stair convex chains behave like simplices in several occasions (e.g. Proposition~\ref{p:cubical-simplices} and Lemma~\ref{l:K_generic_grid_minor}). 
\end{remark*}

\subsection{Non-recursive formulation}
\label{subsec:nonrecursive}
From the examples above, we glimpse a general pattern for ``unwrapping''
the recursive definition of $\stc^m_k$. We extend the definition
of a stair convex chain to the case $m=k=0$ by setting, for any integer
$a$ and any chain $\sigma$, $\stc^0_0(a) \times \sigma = \sigma \times
\stc^0_0(a) = \sigma$. Note that $\stc^0_0(a)$ acts as the identity 
with respect to $\times$.

\begin{lemma}\label{l:unwrapped}
For every $m \ge k \ge 1$ and any $1 \le a_1 < a_2 < \ldots <
a_{k+1} \le n$ we have
\[\stc^m_k(a_1, \dots, a_{k+1}) = \sum_{\stackrel{t_1, t_2, \ldots, t_{k+1} 
\in \NN_0}{t_1+t_2+\ldots + t_{k+1} = m-k}}
\stc^{t_1}_0(a_1) \times \prod_{i=1}^{k} \pth{ [a_{i}, a_{i+1}]
\times \stc^{t_{i+1}}_0(a_{i+1})}.\]
\end{lemma}

\begin{proof}
We define $s^m_0(a) \coloneqq \stc^m_0(a)$, and for $m \geq k\geq 1$ 
we denote by $s^m_k(a_1, \dots, a_{k+1})$ the right-hand term 
of the identity to prove. 
Note that $s^m_k$ immediately satisfies the first two relations 
that define $\stc^m_k$:
\renewcommand{\arraystretch}{1.4}
\begin{center}
\begin{tabular}{lrcl}
$\mathbf{(m> k =0)}$ & $s^m_0(a)$ & $=$ & 
$\overbrace{(a,\dots, a)}^{m\text{-fold}}$\\
$\mathbf{(m= k > 0)}$ & $s^m_k(a_1, \dots, a_{m+1})$ & $=$ & $[a_1,a_2] 
\times [a_2,a_3] \times \ldots \times [a_m,a_{m+1}]$\\
\end{tabular}
\end{center}
\noindent
It therefore suffices to prove that $s^m_k$ also satisfies the third relation:
\renewcommand{\arraystretch}{1.4}
\begin{center}
\begin{tabular}{lrcl}
$\mathbf{(m > k > 0)}$ & $ s_k^m(a_1, \ldots, a_{k+1})$ & $=$ & 
$s^{m-1}_{k-1}(a_1, \ldots, a_k) \times [a_k,a_{k+1}]$\\ 
&&& \quad $+\  s^{m-1}_k(a_1, \ldots, a_{k+1}) \times \{a_{k+1}\}$\\
\end{tabular}
\end{center}
For $m \ge k \ge 1$ let us define
\[I(m,k) \coloneqq \{(t_1, t_2, \ldots, t_{k+1}) \in {\NN_0}^{k+1}
\colon t_1+t_2+\ldots + t_{k+1} = m-k\},\]
  and note that these sets of vectors are pairwise disjoint. For
  $(t_1, \ldots, t_{k+1}) \in I(m,k)$, let us define
  \[ \phi(t_1, \ldots, t_{k+1}) \coloneqq \stc^{t_1}_0(a_1) \times
  \prod_{i=1}^{k} \pth{ [a_{i}, a_{i+1}] \times
    \stc^{t_{i+1}}_0(a_{i+1})}.\]
  In this notation, we have 
  \[ s^m_k(a_1, \dots, a_{k+1}) = \sum_{(t_1, t_2, \ldots, t_{k+1}) \in I(m,k)} \phi(t_1, \ldots, t_{k+1}).\]
  Note that for $(t_1, t_2, \ldots, t_{k+1}) \in I(m,k)$, we have
\[ \phi(t_1, \ldots, t_{k+1}) = 
\left\{\begin{array}{ll}\phi(t_1, \ldots, t_{k}) \times [a_k,a_{k+1}] & 
\text{ if } t_{k+1}=0,\\ 
\phi(t_1, \ldots, t_{k+1}-1) \times \{a_{k+1}\} & 
\text{otherwise}. \end{array} \right.
\]
The desired identity now follows by splitting $I(m,k)$ into
\[\begin{aligned}
I_0(m,k) \coloneqq & \{(t_1, t_2, \ldots, t_{k+1}) \in I(m,k) \colon t_{k+1}=0\}, &
\text{and}\\
I_>(m,k) \coloneqq&  I(m,k) \setminus I_0(m,k),
\end{aligned}\]
and observing that $(t_1, t_2, \ldots, t_{k+1}) \mapsto (t_1, t_2,
\ldots, t_{k+1}-1)$ defines a bijection between $I_>(m,k)$ and $I(m-1,k)$,
and that $(t_1, t_2, \ldots, t_{k+1}) \mapsto (t_1, t_2, \ldots, t_{k})$
defines a bijection between $I_0(m,k)$ and $I(m-1,k-1)$.
\end{proof}

\begin{corollary}\label{c:hyperplane}
Let $k < m$. An axis-parallel hyperplane 
$\Pi_j(a) = \{(x_1, \dots, x_m)\in \mathbb{R}^m : x_j=a\}$ contains a $k$-dimensional
face of the support of $\stc^m_k(a_1, \dots, a_{k+1})$ if and only if $a \in
\{a_1, \dots, a_{k+1}\}$.
\end{corollary}
\begin{proof}
By  Lemma~\ref{l:unwrapped} we have
  \[\stc^m_k(a_1, \dots, a_{k+1}) = \sum_{\stackrel{t_1, t_2, \ldots, t_{k+1} \in \NN_0}{t_1+t_2+\ldots + t_{k+1} = m-k}}
  \stc^{t_1}_0(a_1) \times \prod_{i=1}^{k} \pth{ [a_{i}, a_{i+1}]
    \times \stc^{t_{i+1}}_0(a_{i+1})}.\]
 Note that the support of $\stc^m_k(a_1, \dots, a_{k+1})$ is the union of the supports of the summands of the right-hand side. The support of a summand is contained in $\Pi_j(a)$ if and only if for some $j$ we have $t_j \neq 0$ and~$a_{j} = a$.
\end{proof}

\subsection{Stair convex chains and the boundary operator}
\label{s:boundary}

{We start by proving that stair convex chains behave like $k$-dimensional simplices with respect to the boundary operator on grid complexes. Figure~\ref{fig:boundaries} illustrates this phenomenon in $2$ dimensions.} To formalize this claim, let us define 
\[(a_1, \dots,
\widehat{a_i}, \dots, a_{k+1}) \coloneqq (a_1, \dots, a_{i-1}, a_{i+1},
\dots, a_{k+1}),\] 
that is,~$\widehat{\hspace{1ex}}$~denotes the coordinate to be
omitted. (Recall that all homology in this paper has coefficients in~$\ZZ_2$.)

\begin{figure}[ht]
 \centering
\begin{tikzpicture}[scale=1.1]

\begin{scope}
\draw[-latex] (0,0) -- (4,0);
\draw[-latex] (0,0) -- (0,4);
\draw[dotted] (1,0) -- (1,2) (2,0) -- (2,2) (0,2) -- (1,2) (0,3) -- (1,3);
\fill[blue, opacity=.15] (1,2) -- (2,2) --(2,3) --(1,3) ;
\draw[line width=.43mm, blue!50!black] (1,2) -- (2,2) --(2,3) --(1,3) --cycle;
\node[below] at (1,-.1) {\small $a$};
\node[below] at (2,-.02) {\small $b$};
\node[left] at (0,2) {\small $b$};
\node[left] at (0,3) {\small $c$};
\node[right] at (2.1,2.5) {\small $\stc^2_2(a,b,c)$};
\end{scope}

\begin{scope}[xshift = 7cm]
\draw[-latex] (0,0) -- (4,0);
\draw[-latex] (0,0) -- (0,4);
\fill[blue, opacity=.15] (1,2) -- (2,2) --(2,3) --(1,3) ;
\draw[line width=.43mm, blue!70!black] (1-.025,1) -- (1-.025,3+.025) --(3,3+.025);
\draw[line width=.43mm, green!50!black] (1+.025,1) -- (1+.025,2) --(2,2);
\draw[line width=.43mm, red!70!black] (2,2) -- (2,3-.025) --(3,3-.025);
\filldraw[white] (1,1) circle (1.4pt);
\draw (1,1) circle (1.4pt);
\filldraw[white] (2,2) circle (1.4pt);
\draw (2,2) circle (1.4pt);
\filldraw[white] (3,3) circle (1.4pt);
\draw (3,3) circle (1.4pt);
\node[below] at (1,1) {\small $(a,a)$};
\node[right] at (2-.04,2-.1) {\small $(b,b)$};
\node[right] at (3,3) {\small $(c,c)$};
\node[right] at (2,2.5) {\color{red!70!black} \footnotesize $\stc^2_1(b,c)$};
\node[right] at (1,1.5) {\color{green!50!black} \footnotesize $\stc^2_1(a,b)$};
\node[above] at (1,3) {\color{blue!70!black} \footnotesize $\stc^2_1(a,c)$};
//\end{scope}

\end{tikzpicture}
 \caption{On the left: An illustration of the 2-chain
   $\stc^2_2(a,b,c)$ with highlighted boundary. On the right: An
   illustration of the boundary of $\stc^2_2(a,b,c)$ decomposed into
   the sum of $\stc^2_1(a,b), \stc^2_1(b,c)$ and $\stc^2_1(a,c)$,
   respectively. Note that both $\{a\}\times [a,b]$ and
   $[b,c]\times \{c\}$ cancel out since we work with $\mathbb
   Z_2$ coefficients}
 \label{fig:boundaries}
\end{figure}

\begin{proposition}\label{p:cubical-simplices}
  For integers $m\geq k \geq 1$ and any sequence $a_1 < a_2 < \ldots <
  a_m$ of elements from $[n]$ we have
  $\displaystyle\partial\stc^m_k(a_1, \dots, a_{k+1}) =
  {\sum_{i=1}^{k+1}}\: \stc^m_{k-1}(a_1, \dots, \widehat{a_i}, \dots,
  a_{k+1})$.
\end{proposition}

The proof is a somewhat lengthy, but straight-forward, 
calculation that takes up most of this section. 
We set up an induction on $m$ by using the recursive
definition of $\stc^m_k$ (for $k<m$) or by applying the product rule
after singling out the factor $[a_m,a_{m+1}]$ (for $k=m$). One
important trick is to handle the factors $[a_{i-1},a_{i+1}]$ arising
from $\stc^m_{k-1}(a_1, \dots, \widehat{a_i}, \dots, a_{k+1})$ using
the identity $[a_{i-1},a_{i+1}] = [a_{i-1},a_{i}] + [a_{i},a_{i+1}]$
between $1$-chains. Our first step is to establish a simpler identity.

\begin{lemma}\label{l:boundary}
For any $m \geq 1$ and $1 \le a_1 < a_2 < \ldots < a_{m+2} \leq n$ we have
\begin{equation} \label{eq:d+2}
  \sum_{i=1}^{m+2}\stc_m^m(a_1, \dots, \widehat{a_i},\dots, a_{m+2}) = 0.
\end{equation}
\end{lemma}
\begin{proof}
For $m=1$, this is a reformulation of $[a_1,a_3]= [a_1,a_2] +
[a_2,a_3]$ which holds by definition whenever $a_1, a_2$ and $a_3$ are distinct. 
For $m \ge 2$, we have
\[\begin{array}{ccl}
\sum\limits_{i=1}^{m+2} \stc_m^m(a_1, \ldots, \widehat{a_i},\ldots, a_{m+2}) & = &
\pth{{\sum\limits_{i=1}^{m}}\  \stc_{m-1}^{m-1} (a_1, \ldots, \widehat{a_i}, \ldots, 
a_{m+1})  \times [a_{m+1}, a_{m+2}] } \\[1ex]
&&  \qquad +\stc^{m-1}_{m-1}(a_1, \dots, a_m )\times [a_{m}, a_{m+2}] \\[1ex]
&& \qquad + \stc^{m-1}_{m-1}(a_1, \dots, a_m )\times [a_{m}, a_{m+1}]\\[2ex]
&=& \pth{ {\sum\limits_{i=1}^{m+1}}\stc_{m-1}^{m-1} 
(a_1, \ldots, \widehat{a_i}, \ldots, a_{m+1}) } \times [a_{m+1}, a_{m+2}],
\end{array}\]
since $[a_{m}, a_{m+2}]+[a_{m}, a_{m+1}]= [a_{m+1}, a_{m+2}]$ (here we
use that the $a_i$ are distinct). The statement therefore
follows by induction on $m$.
\end{proof}

Let us return to the proof of the claim that stair convex chains behave
as simplices for the boundary operator on $G[n]^m$.

\begin{proof}[Proof of Proposition~\ref{p:cubical-simplices}]
  The case $m=1$ follows directly from the definition of boundary
  operator.

\subsubsection*{Case $m=k\geq 2$}\hfill\par
Here, we proceed by induction on $m$, and assume that the statement
holds for $m-1$:
  \[\partial \stc^{m-1}_{m-1}(a_1, \dots, a_{m}) = 
  \sum_{i=1}^{m} \stc^{m-1}_{m-2}(a_1, \dots, \widehat{a_i},
  \dots, a_{m}).\]
From the definition of $\partial$ we get 
  \[\begin{aligned}
  & \ \partial \stc^m_m(a_1, \dots, a_{m+1})\\
  = & \ \  \partial \pth{\stc^{m-1}_{m-1}(a_1, \dots, a_m)\times [a_m, a_{m+1}]}\\
  = & \ \underbrace{\pth{\partial \stc^{m-1}_{m-1}(a_1, \dots, a_m)} 
  \times [a_m, a_{m+1}]}_\alpha + \underbrace{\stc^{m-1}_{m-1}(a_1, \dots, a_m)
  \times \pth{\{a_m\} + \{a_{m+1}\}}}_\beta
  \end{aligned}
  \]
We then use the induction hypothesis to rewrite $\alpha$ as
  \[ \alpha =  \sum_{i=1}^{m} \underbrace{\stc^{m-1}_{m-2}(a_1, \dots, \widehat{a_i},
    \dots, a_{m}) \times [a_m, a_{m+1}]}_{\alpha_i}.\]
  and use Lemma~\ref{l:boundary} to partially expand $\beta$ into
  \[ \beta = \underbrace{\stc^{m-1}_{m-1}(a_1, \dots, a_m) \times \{a_{m}\}}_{\beta_0} + \sum_{i=1}^m\underbrace{\stc^{m-1}_{m-1}(a_1, \dots,
    \widehat{a_i}, \dots, a_{m+1}) \times \{a_{m+1}\}}_{\beta_i}\]
  For $i \le m-1$, we have $\alpha_i + \beta_i = \stc^{m}_{m-1}(a_1,
  \dots, \widehat{a_i}, \dots, a_{m+1})$ by the recursive definition
  of~$\stc^{m}_{m-1}$. For the same reason, for any $a \in [n]$,
  \begin{equation}\label{eq:proof boundary stc}
    \begin{aligned}
      \stc^{m}_{m-1}(a_1, \dots, a_{m-1},a) = & \ 
      \stc^{m-1}_{m-2}(a_1, \dots, a_{m-1}) \times [a_{m-1},a] \\
          & \ + \stc^{m-1}_{m-1}(a_1, \dots, a_{m-1},a) \times \{a\}.
    \end{aligned}
  \end{equation}
  Since the $a_i$ are pairwise distinct, we have $[a_m,a_{m+1}] =
  [a_{m-1},a_m] + [a_{m-1},a_{m+1}]$. 
  Using Identity~\eqref{eq:proof boundary stc} once with $a=a_m$ 
  and once with $a=a_{m+1}$, we obtain
  \[ \alpha_m = \stc^{m}_{m-1}(a_1, \dots,a_m) + \beta_0 + 
  \stc^{m}_{m-1}(a_1, \dots, a_{m-1},a_{m+1}) + \beta_m.\]
  Altogether,
  \[ \partial \stc^m_m(a_1, \dots,
  a_{m+1}) = \alpha+\beta = \sum_{i=1}^{m+1} \stc^{m}_{m-1}(a_1, \dots,
  \widehat{a_i}, \dots, a_{m+1})\]
  as claimed.

\subsubsection*{General case}\hfill\par

We now prove the general case by induction on $m$.  
So assume the statement holds for $m-1$ and consider $1 \le k \le m$. 
We already handled the case $k=m$, so let us consider the case $k<m$, 
for which we can use the recursive definition of $\stc^m_k$:
  \[\stc^m_k(a_1, \ldots, a_{k+1}) = \stc^{m-1}_{k-1}(a_1, \ldots,
  a_k) \times [a_k,a_{k+1}]+\ \stc^{m-1}_k(a_1, \ldots, a_{k+1})
  \times \{a_{k+1}\}.\]
  We thus have
  \[\begin{aligned}
  \partial \stc^m_k(a_1, \ldots, a_{k+1})  = & 
  \pth{\partial \stc^{m-1}_{k-1}(a_1, \ldots, a_k)} \times [a_k,a_{k+1}]\\
  & + \stc^{m-1}_{k-1}(a_1, \ldots, a_k) \times \pth{\{a_k\} + \{a_{k+1}\}}\\
  & + \pth{\partial \stc^{m-1}_k(a_1, \ldots, a_{k+1})} \times \{a_{k+1}\}
  \end{aligned}\]
  Using the induction hypothesis for $\stc^{m-1}_{k-1}$ and
  $\stc^{m-1}_{k}$, we obtain
  \[\begin{aligned}
  \partial \stc^m_k(a_1, \ldots, a_{k+1}) = & \ \sum_{i=1}^{k} 
  \stc^{m-1}_{k-2}(a_1, \dots, \widehat{a_i}, \dots,
  a_{k})\times [a_k,a_{k+1}] \\
  & + \sum_{i=1}^{k+1} \stc^{m-1}_{k-1}(a_1, \dots, \widehat{a_i}, \dots,
  a_{k+1}) \times \{a_{k+1}\}\\
  & + \stc^{m-1}_{k-1}(a_1, \ldots, a_k) \times \pth{\{a_k\} + \{a_{k+1}\}}\\
  = & \ \sum_{i=1}^{k} \underbrace{\stc^{m-1}_{k-2}(a_1, \dots, \widehat{a_i}, \dots,
  a_{k})\times [a_k,a_{k+1}]}_{\alpha_i} \\
  & + \sum_{i=1}^{k} \underbrace{\stc^{m-1}_{k-1}(a_1, \dots, \widehat{a_i}, \dots,
  a_{k+1}) \times \{a_{k+1}\}}_{\beta_i}\\
  & + \underbrace{\stc^{m-1}_{k-1}(a_1, \ldots, a_k) \times \{a_k\}}_{\beta_0}\\
  \end{aligned}\]
  For $i \le k-1$, we recognize the recursive definition of~$\stc^{m}_{k-1}$ in
  \[ \alpha_i + \beta_i = \stc^{m}_{k-1}(a_1, \dots, \widehat{a_i}, \dots,
  a_{k+1}).\]
  Since the elements $a_1, \ldots, a_{k+1}$ are pairwise distinct, 
  $[a_k,a_{k+1}] = [a_{k-1},a_k] + [a_{k-1},a_{k+1}]$ and we can split
  $\alpha_k$ and use the recursive definition of~$\stc^{m}_{k-1}$ on each part:
  \[ \begin{aligned}
    \alpha_k = & \ \stc^{m-1}_{k-2}(a_1, \dots, a_{k-1}) \times 
    [a_{k-1},a_{k+1}] + \stc^{m-1}_{k-2}(a_1, \dots, a_{k-1})\times 
    [a_{k-1},a_{k}]\\
    = & \ \pth{\stc^{m}_{k-1}(a_1, \dots, a_{k-1},a_{k+1}) + \beta_k} +  
    \pth{\stc^{m}_{k-1}(a_1, \dots, a_{k-1},a_{k}) + \beta_0}.
  \end{aligned}\]
  Altogether,
  \[ \stc^m_k(a_1, \ldots, a_{k+1})  = 
  \sum_{i=1}^{k+1}   \stc^{m}_{k-1}(a_1, \dots, \widehat{a_i}, \dots, a_{k+1}),\]
  as claimed.
\end{proof}

Finally, we are now in position to prove the ``Picasso Lemma''.

\begin{proof}[Proof of Lemma \ref{l:K_generic_grid_minor}]
Label the vertices of $K$ as $v_1, v_2, \dots, v_n$. 
For every simplex $\{v_{i_1}, \dots, v_{i_{k+1}}\} \in K$ 
with $1\leq i_1 < i_2 <\cdots < i_{k+1}\leq n$, we define:
\[g(\{v_{i_1}, \dots, v_{i_{k+1}}\}) \;  \coloneqq \; \stc^m_k(i_1, \dots, i_{k+1}). \]
We extend $g$ linearly into a map 
$g_\bullet \colon  C_\bullet(K) \to C_\bullet(G[n]^m)$. 
Proposition~\ref{p:cubical-simplices} guarantees that $g_\bullet$ 
commutes with the boundary operator, making it a valid chain map. 
Furthermore, since every vertex $v_a\in V(K)$ is mapped to a 
single vertex $\{a\}\times \cdots \times \{a\}\in G[n]^m$, the map is nontrivial. 

To show that $g_\bullet$ is generic, let $\sigma$ and $\tau$ be a pair of 
disjoint simplices in $K$, and consider arbitrary cells 
$\sigma' \in \supp(g_\bullet(\sigma))$ and $\tau' \in \supp(g_\bullet(\tau))$. 
Since $m > \mu(K)$ (implies $\dim \sigma + \dim \tau < m$), 
it follows that $\dim \sigma' + \dim \tau' < m$. 
By Observation~\ref{o:hyperplane}, there is at least one coordinate, $x_i$, 
that is constant for both cells. 
However, they cannot share the same constant value in this coordinate 
(i.e., $x_i = a$ is impossible for both), because Corollary~\ref{c:hyperplane} 
would imply that both $\sigma$ and $\tau$ contain the vertex $v_a$, 
contradicting their disjointness. 
Thus, the pair $\{g_\bullet(\sigma), g_\bullet(\tau)\}$ is generic.
\end{proof}

\begin{remark*} The reader may note (as can be formalized using \cite[Lemma 1.4]{stairConv}) that mapping simplices to stair convex chains generated by points on the main diagonal of the grid complex closely resembles the canonical geometric realization of a simplicial complex via the moment curve. \end{remark*}

\section{Wrapping up}
\label{s:wrapup}

We now have all the ingredients to prove Theorems~\ref{t:weak colors} and~\ref{t:basic_propogation}. We start with our weak colorful Helly theorem, which we restate here for convenience.

\begin{customthm}{\ref{t:weak colors}}[Weak colorful Helly theorem]
For any finite simplicial complex $K$ and integers $b\geq 1$ and $m> \mu(K)$, there exists an integer $t = t(b,K,m)$ with the following property: If $\FF = \FF_1 \sqcup \cdots \sqcup \FF_m$ is an $m$-colored $(K,b)$-free cover where each color class has size~$t$ and every colorful subfamily has nonempty intersection, then $\FF$ contains some $2m-\mu(K)$ members with nonempty intersection.
\end{customthm}

\begin{proof}
Let $n$ be the number of vertices in $K$ and let $d$ be the dimension of $K$.
We will show that the theorem holds with $t \leq t(b,d,m,n)$ as given by Lemma~\ref{l:colored constrained}.

To this end, let $\FF$ be an $m$-colored $(K,b)$-free cover where 
each color class has size~$t = t(b,d,m,n)$ and 
every colorful subfamily has nonempty intersection. 
Then $\FF$ satisfies the conditions of Lemma~\ref{l:colored constrained}, and so
there exists a nontrivial chain map 
\[f_{\bullet} : C_{\bullet}( Y ) \to C_{\bullet}(\univ)\]
that is constrained by~$\FF$, 
where $Y$ denotes the $d$-skeleton of $G[n]^m$. 
This means there is a subgrid $\Gamma : V(G[n]^m) \to V(\grid_\FF)$ such that 
$\supp f_\bullet(\sigma)\subset 
\bigcap \psi\pth{\Gamma_{\bullet}(\sigma)}$ for every $\sigma \in Y$.

Since $m > \mu(K)$ the chain map $f_\bullet$ also satisfies the conditions of 
Corollary~\ref{t:cubical_and_excluded} (where $Y$ still denotes the $d$-skeleton of $G[n]^m$). 
Therefore we can find two disjoint cells $\sigma$ and $\tau$ in $Y$ 
which satisfy:
\begin{enumerate}[(i)]
\item[(i)] $\sigma$ and $\tau$ are not contained in a common axis-parallel hyperplane.
\item[(ii)] The supports of $f_\bullet(\sigma)$ and $f_\bullet(\tau)$ overlap. 
\end{enumerate}
Since $f_\bullet$ is constrained by $\FF$, 
condition~(ii) implies that there is a point contained in every member of the union
$\psi\pth{\Gamma_\bullet(\sigma)} \cup  \psi\pth{{\Gamma_\bullet(\tau)}}$. 
From~(i) and Claim \ref{c:no trick}~(iii) it follows that the colorful families
$\psi\pth{\Gamma_\bullet(\sigma)}$ and $\psi\pth{{\Gamma_\bullet(\tau)}}$ are disjoint. Finally, we compute the size of this union:
\begin{align*}
|\psi\big(\Gamma_\bullet(\sigma)\big) \cup
\psi\big(\Gamma_\bullet(\tau) \big)| &  =  |\psi\big( \Gamma_\bullet(\sigma) \big)| + |\psi \big(\Gamma_\bullet(\tau)\big)|  \\
  & =  (m-\dim \sigma) + (m-\dim \tau)   \ge 2m-\mu(K).
\end{align*}
(Note that the last inequality uses that $\sigma$ and $\tau$ are disjoint.)
  \qedhere
\end{proof}

Last, we spell out the use of supersaturation outlined in Section~\ref{subsec:strategy} to prove stepping-up theorem using our weak colorful Helly theorem. 

Let us recall some (standard) terminology on hypergraphs. An \emph{$m$-uniform hypergraph} is a pair $H=(V,E)$ where $V=V(H)$ is a finite set of vertices and $E=E(H)\subset \binom{V}{m}$ is the edge set. An $m$-uniform hypergraph is
\emph{$m$-partite} if the vertex set can be partitioned into disjoint sets
(parts) $V(H) = V_1 \cup \cdots \cup V_m$ such that every
edge contains exactly one vertex from each $V_i$. Given integers
$m\geq 2$ and $t \geq 1$, we let $K^m(t)$ denote the complete
$m$-partite $m$-uniform hypergraph with parts $V_1, \dots, V_m$
where $|V_i|=t$. That is, the edge set of $K^m(t)$ consists of all
$m$-tuples of $V_1 \cup \cdots \cup V_m$ that contain exactly one
element from each~$V_i$. 

A hypergraph $H$ \emph{contains} a hypergraph $H'$ if there is an
injection $f: V(H') \to V(H)$ such that for every $e' \in E(H')$,
$f(e') \in E(H)$. (In particular, we do \emph{not} require that $H'$ is
an \emph{induced} sub-hypergraph of $H$.) We use the following \emph{supersaturation}
theorem of Erd\H{o}s and Simonovits:

\begin{theorem*}[{Erd{\H os}-Simonovits~\cite[Corollary 2]{Erdos-Simonovits}}]\label{erdos simonovits}
  For any positive integers $m$ and $t$ and any $\varepsilon>0$ there
  exists $\rho = \rho(\varepsilon, m, t)>0$ such that any $m$-uniform
  hypergraph $H = (V,E)$ with $|E|\geq \varepsilon \binom{|V|}{m}$
  contains at least $\rho |V|^{mt}$ copies of
  $K^m(t)$.
\end{theorem*}

For convenience, we restate our final theorem.

\begin{customthm}{\ref{t:basic_propogation}}[Stepping-up theorem]
Fix a simplicial complex $K$, a value $\delta \in (0,1]$, and integers $b\geq 1$ and $m > \mu(K)$. If $\mathcal{F}$  is a sufficiently large $(K,b)$-free cover such that $\pi_m(\mathcal{F})\geq \delta \binom{ | \mathcal{F} | }{m}$, then $\pi_{m+1}(\FF)\geq \gamma\binom{ | \mathcal{F} |}{m+1}$, where $\gamma > 0$ is a constant that depends only on $\delta$, $b$, $m$ and $K$. 
\end{customthm}

\begin{proof}
Fix a simplicial complex $K$ and integers $b\geq 1$ and $m > \mu(K)$. 
Let $t=t(b,K,m)$ be the constant from Theorem \ref{t:weak colors}. 
Consider some $(K,b)$-free cover $\FF$. 

For a subfamily $\FF' \subseteq \FF$, let $H[\FF']$ be the $m$-uniform hypergraph 
whose vertices are the members of $\FF'$ and whose edges are the 
$m$-tuples of $\FF'$ with nonempty intersection. 
By assumption, our hypergraph $H[\FF]$ contains at least  
$\delta\binom{|\FF|}{m}$ edges. 
By the Erd{\H os}--Simonovits theorem, 
for some constant $\rho>0$ depending only on $m$, $t$, and $\delta$, 
there are at  least $\rho\binom{|\FF|}{mt}$ distinct $mt$-element 
subfamilies $\FF'$ of $\FF$ such that $H[\FF']$ contains a copy of $K^m(t)$. 

Our choice of $t$ ensures that Theorem~\ref{t:weak colors} 
applies to every such subfamily $\FF'$, 
and therefore each $\FF'$ contributes 
some $2m-\mu(K)\ge m+1$ members with non-empty intersection. 
Each $(m+1)$-element subset of $\FF$ with non-empty intersection 
is contained in $\binom{|\FF|-(m+1)}{mt-(m+1)}$ distinct 
$(mt)$-tuples $\FF'$. There are  therefore at least
  \[\frac{\rho\binom{|\FF|}{mt}}{\binom{|\FF|-(m+1)}{mt-(m+1)}} = 
  \frac{\rho}{\binom{mt}{m+1}}\binom{|\FF|}{m+1}\]
$(m+1)$-tuples of $\FF$ with nonempty intersection. In other words,  
$\pi_{m+1}(\FF)$ is at least $\delta' \coloneqq \rho/\binom{mt}{m+1}$, 
where $\rho$ depends only on $m$, $t$, and $\delta$, 
that is on $m$,  $b$, $K$ and $\delta$. That concludes the proof.
\end{proof}

\section{Discussion and outlook}\label{s:outlook}

In conclusion, let us comment on some related problems and future works.

\subsubsection*{Topological generalizations of Helly's theorem}\hfill\par

Relaxing the convexity assumption in Helly's theorem is a classical question that goes back to the \emph{topological Helly theorem} \cite{helly1930, deb1970}, which asserts that Helly's theorem holds for finite good covers. The fractional Helly theorem and the $(p,q)$-theorem also hold for finite good covers~\cite{AKMM}. 

These results were, in turn, extended to families in which intersections have a bounded number of connected components, all contractible~\cite{matouvsek1997helly,acyclic,KM2008}, then to families in which intersections have bounded Betti numbers~\cite{hb17, rb20}. At their core, these extensions rely either on variants of the Nerve Theorem and the Vietoris-Begle Theorem~\cite{acyclic,KM2008} or on variants of the Van Kampen-Flores Theorem~\cite{matouvsek1997helly,hb17,rb20} (see also the survey~\cite{tophelly}). The latter works readily generalize to the setting of complexes with a forbidden homological minor, with the Van Kampen-Flores Theorem being one specific forbidden homological minor for $\RR^d$.

In particular, our $(p,q)$-theorem~\ref{t:pq for Kbfree} implies that for every $p\geq q \geq d+1$ and $b\geq 1$, the assertion of the $(p,q)$-theorem holds for any finite family of sets in $\RR^d$ such that the intersection of every subfamily has $i$th reduced Betti number less than $b$ for $0 \le i < \lceil d/2\rceil$. Note that the constant number of points given by the $(p,q)$-theorem in this case depends not only on $p$, $q$, and $d$, but also on $b$.

A natural question is whether the $(p,q)$-theorem generalizes to families of sets in manifolds. For surfaces, this is the case as we can improve~\cite[Theorem~2.6]{rb20}. Specifically, our $(p,q)$-theorem~\ref{t:pq for Kbfree} implies that for every $p\geq q \geq 3$ and $b\geq 1$, the assertion of the $(p,q)$-theorem holds for any finite family of sets on a compact, $2$-dimensional real manifold, such that the intersection of every subfamily has at most $b$ path-connected components. (Here the constant number of points depends not only on $p$, $q$, $d$ and $b$, but also on the surface considered.) We conjecture that it extends to higher-dimensional manifolds as well.

\begin{conjecture}\label{c:fbd}
    For every $d$-dimensional manifold $M$, there exists a simplicial complex $K_M$ with $\mu(K_M) = d$ such that for any triangulation $\univ$ of $M$, we have $K_M\not\prec_H \univ$.
\end{conjecture}

\noindent
We remark that some version of Conjecture~\ref{c:fbd} was very recently settled in the affirmative by Avvakumov, Bin and Goaoc~\cite{ABG}.

\subsubsection*{Upper-bound theorems}\hfill\par

In fractional Helly theorems, the main parameter one usually tries to optimize is the fractional Helly number. In that respect, our Theorem~\ref{t:fh} is very satisfying: the bound is independent of $b$ and is sharp for at least some forbidden homological minors. The next parameter that one may try to improve is the dependency between the size of the large intersecting family (the $\beta$) and the proportion of intersecting few-tuples (the $\alpha$). 

In the convex case, the original fractional Helly theorem of Katchalski and Liu~\cite{kl79} gave $\beta \ge \frac{\alpha}{d+1}$. They also showed that one may assume that $\beta(\alpha) \to 1$ when $\alpha \to 1$. Their proof uses the observation, which they dubbed the \emph{stepping-up lemma}, that for any $0 < \alpha \le 1$, any $d < k < \ell$ and any finite family $\FF$ of convex sets in $\RR^d$, 
\begin{equation}\label{eq:stepup convex}
  \pi_{k}(\FF) \ge \alpha\binom{ | \mathcal{F} |}{k} \quad \implies \quad  \pi_\ell(\FF) \ge \pth{1-\pth{1-\alpha}\binom{\ell-1}{k-1}}\binom{| \mathcal{F} |}\ell.
\end{equation}
In particular, if {$\alpha > 1-1/\binom{\ell-1}{k-1}$, then a positive proportion of the $\ell$-element subfamilies of $\FF$ have nonempty intersection.} Our stepping up theorem~\ref{t:basic_propogation} asserts that this propagation of positive densities holds more generally for $(K,b)$-free covers. 

For convex sets, the optimal dependency of $\beta$ on $\alpha$ was established by Kalai to $\beta \ge 1-(1-\alpha)^{\frac1{d+1}}$. Kalai's proof is based on the \emph{upper bound theorem} that he~\cite{kalai1984} and Eckhoff~\cite{eckhoff1985upper} established independently. The upper bound theorem asserts that for any family $\FF$ of $n$ convex sets in $\RR^d$, for any $k$ such that $d < k \le d+r$,
\begin{equation}\label{eq:ubt}
{\pi_{k}(\FF) > \sum_{i=0}^d \binom{n-r}{i}\binom{r}{k-i} \quad \implies \quad \pi_{d+r+1}(\FF) > 0,}
\end{equation}
This was recently extended to more general set systems in $\RR^2$ and in surfaces~\cite[Theorem~2.2 and~2.3]{KP}. Our fractional Helly theorem shows some form of upper bound theorem for $(K,b)$-free covers. We did not try to extract it as we believe it is quantitatively rather weak.

\begin{question}
 What is the best possible upper bound theorem for $(K,b)$-free covers?
\end{question}

\subsubsection*{Collapsibility and Lerayness}\hfill\par

The known proofs of the upper bound theorem~\eqref{eq:ubt} are typically more general than the geometric setting, and deal with certain properties of nerves of families of convex sets. The more elementary proofs apply to \emph{$d$-collapsible} complexes~\cite{eckhoff1985upper,kalai1984,alon1985simple}, that is complexes that can be reduced by discrete homotopy moves (called {\em collapses}) to a $d$-dimensional complex~\cite[Lemma~1]{wegner1975d}. The more general proof, also due to Kalai (a presentation can be found in Hell's PhD thesis~\cite[$\mathsection 5.2$]{hellthesis}), applies to {\em $d$-Leray} complexes, that is complexes in which all induced subcomplexes have vanishing homology in dimension $d$ and above. Theorem~\ref{t:basic_propogation} reveals that nerves of $(K,b)$-free covers enjoy some of the consequences of bounded Lerayness. More generally, we conjecture:

\begin{conjecture}
{For any simplicial complex $K$ and positive integer $b$, the nerve of any $(K,b)$-free cover is $L$-Leray, where $L$ depends only on $K$ and $b$.}
\end{conjecture}

\noindent
We note that $L$ must be at least $b(\mu(K)+2)-1$~\cite[Example 2]{hb17}. A first step in this direction was done by Holmsen et al.~\cite{holmsen2019nerves}, who considered the special case when $\FF = \{G_i\}$ is a family of connected subgraphs of a graph $G$ such that any nonempty intersection of members of $\FF$ is also connected. They conjectured that for any $r\geq 1$, if $K_{r+2}$ is not a {minor} of $G$, then the nerve of $\FF$ is $r$-Leray, and verified this for $r\leq 3$. It should be noted that in \cite{holmsen2019nerves} they deal with \emph{graph minors} and not homological minors.

\subsubsection*{Homological VC dimension}\hfill\par

Deeper connections between discrete geometry and topological combinatorics were suggested by Kalai and Meshulam in a program to develop a theory of \emph{homological VC dimension}. For a positive integer $h$ and a family $\FF$ of sets in $\mathbb{R}^d$, let us call the function
\[ \shat{h}{\FF}: \left\{ \begin{array}{rcl} \NN & \to & \NN \cup \{\infty\} \\ k &
    \mapsto & \sup \{ \tbeta_i\pth{\cap\GG} \colon \GG \subseteq \FF, 1 \le |\GG| \le k, 0 \le
    i < h\} \end{array}\right.\]
the ($h$th) \emph{homological shatter function} of $\FF$. The combination of two conjectures of Kalai and Meshulam suggests that families of open sets in $\RR^d$ with polynomial homological shatter function should satisfy a fractional Helly theorem:

\begin{conjecture}[{Following~\cite[Conjectures~6
        and~7]{kalai_conjectures}}]\label{c:combined conjecture}
  For any integer $0 \le m \le d$ and any constant $C>0$, there exists a
  function $\beta:(0,1) \to (0,1)$ such that the following holds. For
  any $\alpha >0$ and any sufficiently large {finite family $\FF$ of open sets} in
  $\RR^d$ with 
  $\shat{d}{\FF}(k) \le Ck^m$, if
  $\pi_{d+1}(\FF) \ge \alpha\binom{|\FF|}{d+1}$ then some $\beta(\alpha)|\FF|$ members
  of $\FF$ have a point in common.
\end{conjecture}

\noindent
A combination of Conjectures~6 and~7 from~\cite{kalai_conjectures} also appeared in~\cite[Conjecture~17]{Kalai-blog}. Here we took upon ourselves to dissociate the dimension $d$ of the space and the degree~$m$ of the polynomial bounding the homological shatter function. We give in to the temptation to generalize this conjecture to \emph{$K$-free covers}, that is to finite families  $\FF$ of (not necessarily induced) subcomplexes of a simplicial complex $\univ$ such that $K \not\prec_H \univ$. We propose the following more general conjecture:

\begin{conjecture}\label{c:extended conjecture}
  For any simplicial complex $K$, integer $t \ge 0$, and constant $C>0$, there exists a function $\beta:(0,1) \to (0,1)$ such that the following holds. For any $\alpha >0$ and any sufficiently large $K$-free cover $\FF$ with $\shat{\dim K}{\FF}(k) \le Ck^t$, if $\pi_{\mu(K)+1}(\FF) \ge \alpha\binom{|\FF|}{\mu(K)+1}$, then some $\beta(\alpha)|\FF|$ members of $\FF$ have a point in common.
\end{conjecture}

\noindent
In other words, we conjecture a fractional Helly theorem for $K$-free covers whose $(\dim K)th$ homological shatter function is bounded by a polynomial of degree $t$. Our Theorem~\ref{t:fh}, by taking $b > C$, resolves the $t=0$ case of this conjecture. 

The technique we developed to prove our stepping-up theorem~\ref{t:basic_propogation} actually shows that if Conjecture~\ref{c:extended conjecture} holds with some fractional Helly number $h$, then it holds with the fractional Helly number $\mu(K)+1$ as claimed. Indeed our construction of constrained chain map does not use the full power of the assumption that the colored family $\FF$ is $(K,b)$-free. Specifically, to obtain the conclusion of Lemma~\ref{l:colored constrained}, it suffices that the $m$-colored family $\FF$ of sub-complexes of~$\univ$ satisfies the following "sliding window" assumption:

\begin{quote}
{\em The $j$th reduced Betti number $\tilde{\beta}_j(\bigcap_{S \in \GG}S)$ is strictly less than $b$ for all $0\leq j < d$ and all colorful subfamilies $\GG \subseteq \FF$ with $|\GG|=m-j-1$.}
\end{quote}

\noindent
This change does not affect the bound on $t(b,d,m,n)$. This in turn leads to a stepping-up theorem with weaker assumptions:

\begin{theorem} \label{t:propogation_window}
Fix a simplicial complex $K$, a value $\delta \in (0,1]$, and integers $b\geq 1$ and $m > \mu(K)$.  If $\FF$ is a $K$-free cover such that 
\begin{enumerate}[(i)]
    \item[(i)] $\tilde{\beta}_j(\bigcap_{S\in \GG} S) < b$, for all $0\leq j < \dim K$ and $\GG \subseteq \FF$ with $|\GG|=m-j-1$, and
    \item[(ii)] $\pi_m(\FF) \geq \delta\binom{|\FF|}{m}$, 
\end{enumerate}
then $\pi_{m+1}(\FF) \geq \gamma\binom{|\FF|}{m+1}$, where $\gamma$ is a constant that depends only on $\delta$, $b$, $m$, and $K$. 
\end{theorem}

\noindent
Now, suppose the fractional Helly number of $\FF$ is bounded by some number $h$. Starting with the assumption that a positive fraction of the  $(\mu(K)+1)$-tuples of $\FF$ are intersecting, we can apply Theorem~\ref{t:propogation_window} with $m = \mu(K)+1, \ldots, h$ (and in each case $b = Cm^t$) to eventually get that a positive fraction of the  $h$-tuples of $\FF$ are intersecting, and then apply that fractional Helly theorem.

\bibliographystyle{plain}

\end{document}